\documentclass[a4 paper,11pt]{article}

\usepackage{fancyhdr}
\usepackage{latexsym}
\usepackage{mathrsfs}
\usepackage{cancel}
\usepackage{color}
\usepackage{multirow}
\usepackage{eso-pic}
\usepackage{dsfont}
\usepackage{amssymb}
\usepackage{graphicx}
\usepackage{setspace}
\usepackage{geometry}
\usepackage{amsthm}
\usepackage{hyperref}
\usepackage[intlimits]{amsmath}
\usepackage{hyperref}
\usepackage{tikz-cd}
\usepackage{amssymb,amsfonts}
\usepackage[all,arc]{xy}
\usepackage{enumerate}
\usepackage{mathrsfs}
\usepackage{tikz}
\usepackage{bbm}
\usepackage{graphicx}
\usepackage{authblk}
\usetikzlibrary{braids}

\renewcommand{\Im}{\text{Im}}

\newcommand{\ie}{\emph{i.e.}}

\newcommand{\wt}{\widetilde}
\newcommand{\wh}{\widehat}

\newcommand{\Ext}{\mathrm{Ext}}
\newcommand{\Hom}{\mathrm{Hom}}

\newcommand{\Mod}{\mathrm{-Mod}}

\newcommand{\trans}{\mathsf{T}}

\newcommand{\FF}{V_{bc}}
\newcommand{\lp}{\left(}
\newcommand{\rp}{\right)}

\newcommand{\Ind}{\mathrm{Ind}}

\newcommand{\lpp}{(\!(}
\newcommand{\rpp}{)\!)}

\newcommand{\grho}{\fg_\rho}

\newcommand{\pd}{{\partial}}

\newcommand{\lMod}{\Mod_{\mathrm{loc}}}
\newcommand{\FU}{\otimes_\CU}

\newcommand{\C}{\mathbb C}

\newcommand{\Z}{\mathbb Z}

\newcommand{\N}{\mathbb N}

\newcommand{\fgl}{\mathfrak{gl}}
\newcommand{\fg}{\mathfrak{g}}
\newcommand{\fh}{\mathfrak{h}}

\newcommand{\CA}{{\mathcal A}}
\newcommand{\CB}{{\mathcal B}}
\newcommand{\CC}{{\mathcal C}}
\newcommand{\CD}{{\mathcal D}}
\newcommand{\CE}{{\mathcal E}}
\newcommand{\CF}{{\mathcal F}}

\newcommand{\CL}{{\mathcal L}}

\newcommand{\CN}{{\mathcal N}}

\newcommand{\CS}{{\mathcal S}}
\newcommand{\CT}{{\mathcal T}}
\newcommand{\CU}{{\mathcal U}}
\newcommand{\CV}{{\mathcal V}}
\newcommand{\CW}{{\mathcal W}}  

\newcommand{\CY}{{\mathcal Y}}
\newcommand{\CZ}{{\mathcal Z}}

\newcommand{\norm}[1]{{{:\!{#1}\!:}}}
\newcommand{\be}{\begin{equation}}
\newcommand{\ee}{\end{equation}}

\newcommand{\btik}{\begin{tikzcd}}
\newcommand{\etik}{\end{tikzcd}}
\newcommand{\mf}{\mathfrak}

\usepackage{calc}
\newcommand{\grau}{gray!60}
\newenvironment{grform}{\begin{tikzpicture}[intext]}{\end{tikzpicture}}
\tikzset{
norm/.style = {ultra thick, color = #1},
norm/.default = black,
intext/.style = {baseline = (current bounding box.center)},
}

\newcommand{\vLine}[5]{
  \draw[ultra thick, color = #5, rounded corners] (#1 , #2) -- (#1 , #2 * 0.7 + #4 * 0.3 ) -- ( #3 , #2 * 0.3 + #4 * 0.7 ) -- (#3 ,#4);
}
\newcommand{\vLineO}[5]{
  \draw[line width = 6pt , color = white, rounded corners] (#1 , #2) -- (#1 , #2 * 0.7 + #4 * 0.3 ) -- ( #3 , #2 * 0.3 + #4 * 0.7 ) -- (#3 ,#4);
  \draw[ultra thick, color = #5, rounded corners] (#1 , #2) -- (#1 , #2 * 0.7 + #4 * 0.3 ) -- ( #3 , #2 * 0.3 + #4 * 0.7 ) -- (#3 ,#4);
}

\newcommand{\dMult}[5]{
\draw[ultra thick, color = #5] (#1 , #2) -- (#1 , #2 + #4 * 0.2) .. controls (#1
, #2 + #4 * 0.5) .. (#1 + #3 *0.25 , #2 + #4 * 0.5) 
-- (#1 + #3 *0.5 , #2 + #4 * 0.5) -- (#1 + #3 *0.75 , #2 + #4 * 0.5) .. controls
(#1 + #3 , #2 + #4 * 0.5) .. (#1 + #3 , #2 + #4 * 0.2) -- (#1 + #3 , #2);
\draw[ultra thick, color = #5] (#1 + #3 *0.5 , #2 + #4 * 0.5) -- (#1 + #3 *0.5 ,
#2 + #4);
}

\newcommand{\dMultO}[5]{
\draw[line width = 6pt , color = white] (#1 , #2) -- (#1 , #2 + #4 * 0.2) ..
controls (#1 , #2 + #4 * 0.5) .. (#1 + #3 *0.25 , #2 + #4 * 0.5) 
-- (#1 + #3 *0.5 , #2 + #4 * 0.5) -- (#1 + #3 *0.75 , #2 + #4 * 0.5) .. controls
(#1 + #3 , #2 + #4 * 0.5) .. (#1 + #3 , #2 + #4 * 0.2) -- (#1 + #3 , #2);
\draw[ultra thick, color = #5] (#1 , #2) -- (#1 , #2 + #4 * 0.2) .. controls (#1
, #2 + #4 * 0.5) .. (#1 + #3 *0.25 , #2 + #4 * 0.5) 
-- (#1 + #3 *0.5 , #2 + #4 * 0.5) -- (#1 + #3 *0.75 , #2 + #4 * 0.5) .. controls
(#1 + #3 , #2 + #4 * 0.5) .. (#1 + #3 , #2 + #4 * 0.2) -- (#1 + #3 , #2);
\draw[ultra thick, color = #5] (#1 + #3 *0.5 , #2 + #4 * 0.5) -- (#1 + #3 *0.5 ,
#2 + #4);
}

\newcommand{\dAction}[6]{
\draw[ultra thick, color = #5] (#1 , #2) -- (#1 , #2 + #4 * 0.2) .. controls (#1
, #2 + #4 * 0.5) .. (#1 + #3 * 0.8 , #2 + #4 * 0.5) -- (#1 + #3 , #2 + #4 *
0.5);
\draw[ultra thick, color = #6] (#1 + #3 , #2) -- (#1 + #3 , #2 + #4);
}

\makeatletter
\renewcommand*\env@matrix[1][*\c@MaxMatrixCols c]{%
  \hskip -\arraycolsep
  \let\@ifnextchar\new@ifnextchar
  \array{#1}}

\begin{document}

\title{Kazhdan-Lusztig Correspondence for Vertex Operator Superalgebras from Abelian Gauge Theories}

\author[1]{Thomas Creutzig}
\author[2]{Wenjun Niu}

\affil[1]{Department of Mathematics, FAU Erlangen-N\"urnberg, Universitätsstraße 40, 91054, Erlangen,
Germany}
\affil[2]{Perimeter Institute for Theoretical Physics, 31 Caroline St N, Waterloo, ON, Canada, N2L 2Y5}

\newtheorem{Def}{Definition}[section]
\newtheorem{Thm}[Def]{Theorem}
\newtheorem{Prop}[Def]{Proposition}
\newtheorem{Cor}[Def]{Corollary}
\newtheorem{Lem}[Def]{Lemma}
\newtheorem{Rem}[Def]{Remark}
\newtheorem{Asp}[Def]{Assumption}
\newtheorem{Exp}[Def]{Example}
\newtheorem{Conj}[Def]{Conjecture}

\numberwithin{equation}{section}

\maketitle

\begin{abstract}
 We prove the Kazhdan-Lusztig correspondence for a class of vertex operator superalgebras which, via the work of Costello-Gaiotto, arise as boundary VOAs of the topological B twist of 3d $\mathcal{N}=4$ abelian gauge theories. This means that we show equivalences of braided tensor categories of modules of certain affine vertex superalgebras and corresponding quantum supergroups.

    We build on the work of Creutzig-Lentner-Rupert for this large class of VOAs and extend it since, in our case, the categories do not have projective objects, and objects can have arbitrary Jordan-Hölder length.

    Our correspondence significantly improves the understanding of the braided tensor category of line defects associated with this class of TQFT by realizing line defects as modules of a Hopf algebra. In the process, we prove a special case of the conjecture of Semikhatov-Tipunin, relating logarithmic CFTs to Nichols algebras of screening operators.
    
\end{abstract}

\noindent\textbf{Mathematics Subject Classification:} 17B69, 81R10, 20G42, 18M15, 57R56. 

\tableofcontents

\newpage

\section{Introduction}

The objective of this paper is the proof of Kazhdan-Lusztig correspondences for a class of affine Lie superalgebras that arise as (perturbative) boundary vertex algebras of the topological B-twist of 3d $\CN=4$ abelian gauge theories. More precisely, let $G=(\C^\times)^r$ be an abelian group and $V=\C^n$ a representation of $G$ defined by a charge matrix $\rho: \Z^r\to \Z^n$. Throughout this paper, we assume that $\rho$ defines a faithful action, in that it fits into a short exact sequence (see equation \eqref{SES}). Associated with this data is a solvable Lie superalgebra $\grho$ and its affine current VOA $V(\grho)$, which first appeared in \cite{CG} as the perturbative boundary VOA for a holomorphic boundary condition of the aforementioned TQFT. In \cite{BCDN}, this affine current VOA was studied and the Kazhdan-Lusztig category $KL_\rho$ of modules was introduced, which has the structure of a rigid braided tensor category. The Lie superalgebra $\grho$ is not reductive but solvable, and the category $KL_\rho$ is not semisimple; in fact, it doesn't have nonzero projective objects.  

In this paper, we show that we can associate to $\rho$ a Hopf algebra, to be denoted by $U_\rho^E$, which is a deformation of the universal enveloping algebra of $\grho$. Properties of this algebra will be given in Section \ref{sec:QG}. Most importantly, this algebra admits an $R$-matrix in its algebraic completion, which can be evaluated on finite-dimensional modules. This, in particular, endows the category $U_\rho^E\Mod$, the category of finite-dimensional modules of $U_\rho^E$, with the structure of a braided tensor category. Our main result is the following statement.

\begin{Thm}\label{KLrho}

There is an equivalence of braided tensor categories:
\be
KL_\rho\simeq U_\rho^E\Mod.
\ee
\end{Thm}

\subsection{Proof strategy: Schauenburg's functor and Nichols algebras}\label{subsec:proofstrategy}

The strategy of this proof is based on recent work \cite{CLR23}, and can be naturally divided into two parts. The idea behind the first part is that, given a VOA $V$ together with a rigid braided tensor category $\CU$ of $V$-modules, as well as a free field realization $V\to W$ such that the VOA $W$ defines a commutative algebra object $A$ in $\CU$, under a set of technical assumptions, there is a fully-faithful braided tensor functor, known as the Schauenburg's functor
\be
\CS: \CU\to \CZ_{A\lMod}\lp A\Mod\rp,
\ee
where $A\Mod$ is the category of $A$-modules in $\CU$, $A\lMod$ the category of local $A$-modules and $\CZ_{A\lMod}\lp A\Mod\rp$ is the relative Drinfeld center. This relates the complicated category $\CU$ to two potentially easier categories $A\Mod$ and $A\lMod$. In practice, $A\lMod$ is simply the category of modules of $W$, and if $W$ is a Heisenberg VOA, then $A\lMod$ is very well understood. The category $A\Mod$ is still more complicated to deal with. Moreover, fully-faithfulness of $\CS$ requires a good understanding of the category $A\Mod$. This brings us to the second part of the strategy, which is the most challenging part. 

In the second part of the strategy, we relate $A\Mod$ to modules of a Hopf algebra object in $A\lMod$. In general, one certainly doesn't expect such a relation to exist. In the case at hand, the VOA $V$ is the joint kernel of a set of screening operators $S_i: W\to W_i$ for simple current modules $W_i$, and it was first proposed and studied in \cite{semikhatov2012nichols} that screening charges obey the relations of a Nichols algebra, which is a Hopf algebra in $A\lMod$ generated by $\bigoplus_i W_i$. In the case of diagonal braiding, this was rigorously proven by \cite{lentner2021quantum}. Although nonrigorous, the idea of relating Nichols algebra to $A\Mod$ was already present in \cite{semikhatov2012nichols, semikhatov2013logarithmic}, where they dealt with the case of $V=M(p)$ the singlet VOA.  In modern language, see \cite{CLR23}, the proposal of Semikhatov and Tipunin can be rephrased as follows. 
\begin{Conj} \label{Conj:Nicholsintro}
     Let $V\to W$ be an embedding whose image is the joint kernel of a set of screening operators $S_i: W\to W_i$. Let $A$ be the algebra object defined by $W$ in the category of $V$ modules, and let $\CB(M)$ be the Nichols algebra generated by $M=\bigoplus_i W_i$ in $A\lMod$. There is an equivalence of tensor categories:
    \be
\CB (M)\Mod (A\lMod)\simeq A\Mod
    \ee
intertwining the action of $A\lMod$. 

\end{Conj}

For $V= M(p)$ the singlet and $W=H_\phi$ a single Heisenberg, this was proven in  \cite{CLR23} for the category of weighted $V$-modules. The most crucial step of the proof is to produce a fiber functor $\epsilon: A\Mod\to A\lMod$, as well as an abelian equivalence $\CE: A\Mod\to \CB (M)\Mod$ obeying certain properties. In \textit{op.cit.}, the fiber functor was constructed in an ad.hoc manner, using the complete knowledge of the abelian braided tensor category $V\Mod$ established in \cite{Creutzig:singlet1, Creutzig:singlet2}. 

In this paper, we apply this strategy to $V=V(\grho)$. It turns out that our category $KL_\rho$ does not have any non-trivial projective object, therefore one cannot apply the ad hoc construction of  \cite{CLR23}. One of the main technical discoveries of this paper (see Section \ref{sec:KL1proof}) is that one can build this fiber functor using a filtration defined by the monodromy action on $A$-modules. This allows us to prove a simple case of Conjecture \ref{Conj:Nicholsintro} in an abstract manner. More precisely, let $V$ be a VOA and $W$ a VOA containing $V$, such that $V$ is defined by the kernel of a screening operator $W\to \sigma W$, where $\sigma$ is spectral flow automorphism defined by some Heisenberg field $\phi\in W$ of level $1$. 

\begin{Thm}[Theorem \ref{Thm:AmodNmod}]\label{Thm:AmodNmodintro}
    Under a list of technical conditions (contained in Section \ref{sec:AModNMod}), there is an equivalence of tensor categories:
    \be
A\Mod\simeq \CB (\sigma W)\Mod. 
    \ee
\end{Thm}

It turns out that the proof of even this simple case is quite complicated. We leave it to a future project to derive a proof for the general case. Theorem \ref{Thm:AmodNmodintro} allows us to apply this strategy to the case $\rho=1$, in which case $\grho$ is simply $\fgl (1|1)$. We use the Wakimoto free field realization of $\fgl (1|1)$ to be $W$ and show that the relative Drinfeld center is precisely $U_1^E\Mod$. For general $\rho$, we use the ungauging relation from \cite{BCDN} between $V(\grho)$ and $V(\fgl (1|1))$ to derive $U_\rho^E$ as a subquotient of many copies of $U_1^E$. 

\subsection{Physical Motivation}

The physical motivation for this work is the study of a class of non-semisimple TQFTs in 3 dimensions. More precisely, it concerns the topological twists of 3d $\CN=4$ supersymmetric abelian gauge theories. 

In a three-dimensional topological quantum field theory, the set of line operators is expected to carry the structure of a braided tensor category. A classic example of this is Chern-Simons theory with a compact gauge group $G$ and integer level $k$, whose line operators coincide with integrable level-$k$ representations of the WZW VOA $V_k(\mathfrak g)$. Another set of examples comes from subsectors of supersymmetric QFTs that behave topologically, \emph{a.k.a.} topological twists. 

One general class of quantum field theories that admits topological twists is that of 3d $\CN=4$ gauge theories. They in fact have \emph{two} distinct topological subsectors, often called A and B twists. The 3d B twist, introduced by \cite{BlauThompson}, leads to a gauge-theory analog of Rozansky-Witten TQFT \cite{RW,KRS}. The 3d A twist is a dimensional reduction of Witten's archetypical twist of 4d $\CN=2$ Yang-Mills that defines Donaldson-Witten TQFT.

Categorical aspects of line operators in the A and B twists of 3d $\mathcal N=4$ gauge theories were developed in a series of different works, including \cite{BDGH,BF-notes,Webster-tilting,OR-Chern,lineops,CCG,CDGG,hilburn2021tate,BN22,BCDN,HG-Betti,GarnerNiu,GeerYoung, garner2024btwisted}. (See especially the introductions to \cite{hilburn2021tate,BCDN} for further background and references.) The physical constructions of these categories took much inspiration from an early classification of BPS line operators by \cite{AsselGomis}, the classic, general analysis of BPS defects in gauge theory by \cite{GukovWitten}, and earlier work on the categorical structure of Rozansky-Witten theory \cite{KRS,RobertsWillerton}. In addition, we note that another series of works including \cite{Dedushenko:2018bp,3dmodularity,Gang:2021hrd,Cui:2021lyi} extracted data of braided tensor (in fact, modular tensor) categories from 3d supersymmetric gauge theories indirectly, by analyzing partition functions and half-indices, or using string/M-theory dualities; these were recently connected with direct field-theory analysis in \cite{Dedushenko:2023cvd,Ferrari:2023fez}. A notable difference between the 3d $\CN=4$ theories considered in \cite{Dedushenko:2018bp,Dedushenko:2023cvd,Gang:2021hrd,Ferrari:2023fez} and general 3d $\CN=4$ gauge theories is that the former are ``rank zero'' --- they have no moduli space of SUSY vacua , which allows their categories of line operators to be semisimple. In general 3d gauge theories, there is a nontrivial moduli space of vacua, and (correspondingly) the categories of line operators are not semisimple. In a topological twist, they are nontrivial derived categories. 

Despite what is now quite a large body of work on categories of line operators in 3d $\CN=4$ gauge theories, the current state of the art (in the general non-semisimple/derived case) is still remarkably complicated. In this paper and a companion \cite{cdntoappear} we seek to significantly improve this state of affairs by introducing \emph{quantum groups} --- and more generally, quasi-triangular Hopf algebras --- that control line operators. These algebras and their weighted category of modules are recently used in \cite{garner2024btwisted} to construct non-semisimple TQFT via an application of the methods developed in  \cite{costantino2014quantum}. 

The 3d $\CN=4$ abelian gauge theories relevant for this paper are defined by an abelian group $G$ and representation $V$, whose charge matrix $\rho$ is assumed to define a faithful action. This class of theories is closed under 3d mirror symmetry, which swaps A and B twists, so without loss of generality we may also focus on a single twist; we choose the B twist. In \cite{BN22,BCDN} it was explained that the category of line operators in such a B twisted gauge theory may be represented as modules for a certain boundary VOA $\CV_\rho^B$, where $\CV_\rho^B$ is a $\Z^r$ simple current extension of the affine current superalgebra $V(\grho)$. In \cite{BCDN}, the category of modules of $\CV_\rho^B$ is defined using this simple current extension as a de-equivariantization of $KL_\rho$ by the $\Z^r$ lattice; and the full category of line operators is the derived category thereof. Theorem \ref{KLrho} gives a quantum group description to the category of line operators, with explicit tensor product and $R$-matrix, making it more amenable to algebraic studies. 

\begin{Rem}
    Just as the boundary VOA relevant for the physical theory is not $V(\grho)$ but the simple current extension, the quantum group relevant to the physics setup is a subquotient of $U_\rho^E$. We will derive this subquotient in Section \ref{subsec:QGphys}, using de-equivariantization. 
\end{Rem}

\subsection{Conventions and Notations}

Since we are extensively dealing with superalgebras and supercategories, we fix our convention here. A vertex operator superalgebra $V$ is a $\Z$-graded vertex operator superalgebra in the sense of \cite{creutzig2017tensor}. The category $V\Mod$ is the abelian category of super vector spaces with an action of $V$, or in other words, it is a $\Z_2$-graded module of the $\Z_2$-graded algebra $U(V)$ (the universal enveloping algebra of $V$). This is called the ``underlying category" of the supercategory in the sense of  \textit{op.cit.} Similarly, given a superalgebra $U$, the category $U\Mod$ is the category of super vector spaces with an action of $U$. In most of the texts, we will drop the word ``super" unless possible confusion may occur. Statements of \textit{op.cit.} remain true in this setting. 

As has been stated, we denote by $G=(\C^\times)^r$ a product of the abelian group $\C^\times$, and $V=\C^n$ a representation of $G$ defined by a charge matrix $\rho$, which we view as a map:
\be
\rho: \Z^r\to \Z^n.
\ee
The matrix element $\rho^{ia}$ (for $1\leq i\leq n, 1\leq a\leq r$) denotes the weight of the $a$-th copy of $\C^\times$ on the $i$-th basis element of $V$. We will use superscripts to label matrix elements in what follows, and use subscripts for the coefficients of fields of a VOA, so that $Y(A, z)=\sum_n A_n z^n$. As in \cite{BCDN}, we assume that the action of $G$ on $V$ is \textit{faithful}, in the sense that the charge lattice $\rho$ fits into a short exact sequence:
\be 
\btik
0\rar & \Z^r\rar{\rho} & \Z^n \rar{\tau} & \Z^{n-r}\rar & 0
\etik \label{SES} \ee
In this case, one can choose a splitting of this exact sequence $\wt{\rho}: \Z^{n-r}\to \Z^n$, namely $\tau\circ \wt{\rho}=\mathrm{Id}_{n-r}$. Similarly, one can choose a co-splitting $\wt{\tau}: \Z^n\to \Z^r$ such that $\wt{\tau}\circ \wt{\rho}=0$ and $\wt{\tau}\circ\rho=\mathrm{Id}_r$. This is depicted in the following:
\be\label{eqsplitexactsequence}
\begin{tikzcd}
    0\rar & \mathbb{Z}^r\arrow[r, shift right, swap, "\rho"]& \arrow[l, shift right, swap, "\wt{\tau}"] \mathbb{Z}^n \arrow[r, shift right, swap, "\tau"] & \arrow[l, shift right, swap, "\wt{\rho}"]\mathbb{Z}^{n-r}\rar & 0
    \end{tikzcd}
\ee
They satisfy the following equations:
\be
\tau\circ\rho=\wt{\tau}\circ\wt{\rho}=0, \qquad \tau\circ \wt{\rho}=\mathrm{Id}_{n-r}, \qquad \wt{\tau}\circ \rho=\mathrm{Id}_r, \qquad \rho\circ \wt{\tau}+\wt{\rho}\circ \tau=\mathrm{Id}_n.
\ee
We will use this splitting to show that $U_\rho^E$ is a subquotient of many copies of $U_1^E$. 

In the following, given an algebra $A$ in a braided tensor category $\CC$, we will denote by $A\Mod (\CC)$ the category of finite-length $A$ modules in $\CC$. In particular, if $\CC$ is $\mathrm{Vect}$, then $A\Mod$ means the category of finite-dimensional modules of $A$. 

Tensor products over $\C$ will be denoted by $\otimes$. Let $V$ be a VOA and $\CU$ a category of modules admitting a braided tensor structure. Then, the tensor product in $\CU$ will be denoted by $\otimes_\CU$. All tensor products are ordinary tensor products, not derived tensor products. We use Koszul sign rule implicitly in all tensor products, so that whenever two tensor factors are swapped, a sign $(-1)^{F\otimes F}$ should be added, where $F$ is the parity operator.

\subsection{Organization}

\begin{itemize}

\item In Section \ref{sec:VOA}, we review properties of the vertex algebra $V(\grho)$ and its Kazhdan-Lusztig category $KL_\rho$, following \cite{BCDN}.

\item  In Section \ref{sec:QG}, we introduce the quantum group $U_\rho^E$ and derive useful properties, especially its relation to $(U_1^E)^{\otimes n}$, which is parallel to the relationship between $V(\grho)$ and $V(\fgl (1|1))$.

\item  In Section \ref{sec:FF}, we review the first half of the proof strategy, that of the free-field realizations and commutative algebra objects, following \cite{CLR23}. We also derive some useful new results, especially the relationship between local modules and non-local modules of a commutative algebra object. 

\item In Section \ref{sec:YDQG}, we review the second half of the proof strategy, that of connecting $A\Mod$ to a Nichols algebra, again following \cite{CLR23}. We also show that $U_\rho^E$ is the relative double of a Nichols algebra.

\item In Section \ref{sec:KL1proof}, we apply the results of previous sections to prove Theorem \ref{Thm:AmodNmodintro}. As a consequence, we prove Theorem \ref{KLrho} for $\rho=1$. We also extend this to $\rho=\mathrm{Id}_n$. 

\item In Section \ref{sec:ungauge}, we use the ungauging relation to derive Theorem \ref{KLrho} for general $\rho$ from $\rho=\mathrm{Id}_n$. We also use the same strategy to derive the physically-relevant quantum group as a subquotient of $U_\rho^E$. 

\end{itemize}

\subsection*{Acknowledgements}

We would like to thank Tudor Dimofte for collaboration that motivated this work, which will also appear in further companion papers. TC is very grateful to Simon Lentner for collaborations on related problems.
WN is grateful to Niklas Garner and Matthew Young for helpful discussions related to this research. We would like to thank Matthew Rupert and Harshi Yadav for many helpful comments on the paper. WN's research is supported by Perimeter Institute for Theoretical Physics. Research at Perimeter Institute is supported in part by the Government of Canada through the
Department of Innovation, Science and Economic Development Canada and by the Province
of Ontario through the Ministry of Colleges and Universities.

\section{The Vertex Algebra $V(\grho)$ and the Kazhdan-Lusztig Category}\label{sec:VOA}

\subsection{The Vertex Algebra $V(\grho)$}

In \cite{BCDN}, we studied an affine Lie superalgebra $V(\fg_\rho)$ and its simple current extension $\CV_\rho^B$. These VOAs are inspired by the study of 3d $\CN=4$ abelian gauge theories, especially their boundary VOA and category of line operators \cite{CG, CCG}. When $\rho=1$, the VOA $V(\grho)$ is simply the affine vertex superalgebra $V(\fgl(1|1))$, whose representation theory was studied in \cite{Creutzig:2011cu,creutzig2020tensor, BN22}. In this section, we recall the definition and properties of the VOA $V(\grho)$. Consider the Lie superalgebra $\fg_\rho$ over $\mathbb C$ given by
\begin{align} \label{def-grho} \grho &:= T^*(\mathfrak g_\C\oplus \Pi V) \;=\; \mathfrak g_\C \oplus \Pi V \oplus \Pi V^* \oplus \mathfrak g_\C^* \\
& \hspace{.8in} \text{with basis}\;\; (N^a,\;\;\; \psi_{+}^i,\;\;\; \psi_{-}^i,\;\;\; E^a)\,, \notag
\end{align}
for $a=1,...,r$ and $i=1,...,n$. This is a $\Z$-graded Lie superalgebra. Its even/bosonic part is a copy of $\mathfrak g_\C=\C^r$ and a copy of $\mathfrak g_\C^*=\C^r$ (with basis elements $N^a,E^a$). Its odd/fermionic part, indicated by the parity shift `$\Pi$', is $V\oplus V^* = \C^{2n}$ (with basis $\psi_{\pm}^i$). The nonvanishing Lie brackets are given by
\be \label{g-brackets}
[N^a,\psi_{\pm}^i]=\pm \rho^{ia} \psi^i_{\pm}\,,\qquad 
\{\psi_{+}^i,\psi_{-}^j\}=\delta^{ij} \sum_{a=1}^r \rho^{ia}E^a.
\ee
To define the Kac-Moody Lie superalgebra, we need to choose an invariant, even, nondegenerate bilinear pairing $\kappa:\grho\times \grho\to\C$. There are many choices of pairings on $\grho$. The one relevant to the physical setup is derived in \cite{CDG-chiral, garner2022vertex}; \footnote{In fact, one can use a field redefinition to shift this to a more standard bilinear form such that $\kappa_0 (N^a, N^b)=0$ and the other pairings are the same as $\kappa$.} it is given by
\begin{equation}\label{eqbilinearBper}
   \kappa(N^a,N^b)=(\rho^T\rho)^{ab}=\sum\limits_{i=1}^n \rho^{ia}\rho^{ib}\,,\qquad \kappa(N^a,E^b)=\delta^{ab}\,,\qquad\kappa(\psi^{i}_+,\psi^{j}_-)=\delta^{ij}\,.
\end{equation}
This defines a Kac-Moody Lie superalgebra, which we denote by $\wh{\grho}$, which is a central extension of $\grho\otimes \C[t,t^{-1}]$ defined by $\kappa$ in the usual way:
\be
\btik
0\rar & \C K\rar & \wh{\grho}\rar & \grho\otimes \C[t,t^{-1}]\rar & 0
\etik
\ee
where the commutation relation is defined on $a\otimes f(t)$ and $b\otimes g(t)$ by:
\be
[a\otimes f(t), b\otimes g(t)]_{\wh\grho}=[a,b]_{\grho}\otimes f(t)g(t)-\kappa(a,b)K\mathrm{Res}_{t=0} f(t)\pd_t g(t). 
\ee
More specifically, on the generators of $\grho$, the commutation relation reads:
\be
\begin{aligned}
&[N_{n}^a, N_{m}^b]=(\rho^T\rho)^{ab}n\delta_{m+n, 0} K, \qquad [N_{n}^a, E_{m}^b]=n\delta^{ab} \delta_{m+n, 0} K\\  &[N_{n}^a, \psi^{i}_{\pm, k}]=\pm \rho^{ia} \psi^{i}_{\pm, k+n}, \qquad \{\psi^{i}_{+,m}, \psi^{j}_{-,n}\}=m\delta_{m+n, 0} \delta^{ij} K+\delta^{ij}\sum_{a=1}^r\rho^{ia} E^a_{m+n}. 
\end{aligned}
\ee
Denote by $\widehat{\grho}_{\geq 0}=\grho\otimes \C[t]\oplus \C K$, and $\widehat{\grho}_{<0}=\grho\otimes t^{-1}\C[t^{-1}]$, then it is easy to see that $\wh{\grho}$ is a crossed product of $\wh{\grho}_{\geq 0}$ and $\wh{\grho}_{<0}$. In particular, there is a decomposition:
\be
U(\wh{\grho}_{<0})\otimes U(\wh{\grho}_{\geq 0})\to  U(\wh{\grho}).
\ee
For each finite-dimensional module $M$ of $\grho$ and each $k \in \mathbb C$, we define a Verma module of $\wh \grho$ of level $k$ by:
\be
\mathrm{Ind}(M):=U(\wh{\grho})\otimes_{U(\wh{\grho}_{\geq 0})} \lp M\otimes \C_k\rp,
\ee
where $M\otimes \C_k$ is the one-dimensional representation of $\wh{\grho}_{\geq 0}$ where $\grho\otimes t\C[t]$ acts trivially, the zero-mode subalgebra $\grho$ acts on $M$ naturally and $K$ acts by multiplication with $k$. The module $V_k(\grho):=\mathrm{Ind}(\C)$ (where $\C$ is the trivial $\grho$ module) has the structure of a super-VOA, whose OPEs on generators are given by:
\begin{equation}\label{eqperturbVOABgenOPE}
\begin{aligned}
& N^a(z)E^b(w)\sim \frac{k\delta^{ab}}{(z-w)^2}\,,\qquad\;\; \quad N^a(z)N^b(w)\sim \frac{k(\rho^T\rho)^{ab}}{(z-w)^2}\,,\\ &N^a(z)\psi^{i}_+(w)\sim \frac{ \rho^{ia}\psi^{i}_+(w)}{(z-w)}\,,\qquad N^a(z)\psi^{i}_-(w)\sim \frac{- \rho^{ia}\psi^{i}_-(w)}{(z-w)}\,,\\ &\hspace{.3in}\psi^{i}_+(z)\psi^{j}_-(w)\sim \delta^{ij}\left(\frac{k}{(z-w)^2}+\frac{\sum\limits_{a=1}^r \rho^{ia}E^a(w)}{z-w}\right)\,.
\end{aligned}
\end{equation}
Its stress-energy tensor is given by:
\be \label{L-B}
    L(z)\,=\,\frac{1}{2k}\Big(\sum_{1\leq a\leq r}\norm{N^a E^a}+\norm{E^a N^a}-\sum_{1\leq i\leq n} \norm{\psi^{i}_+\psi^{i}_-}+\sum_{1\leq i\leq n} \norm{\psi^{i}_-\psi^{i}_+}\Big)\,.
\ee
It is easy to see that in this case, for different non-zero $k$, the VOAs $V_k(\grho)$ are all isomorphic (albeit with a change of conformal element, which does not matter for the braided tensor category considered in this paper), and therefore we will once and for all fix $k=1$, and denote the VOA by $V(\grho)$. The following statement is true:
\begin{Prop}[ \cite{BCDN} Lemma 8.22]
The VOA $V(\grho)$ is simple. 

\end{Prop}

The affine Lie superalgebra $\wh\grho$ has the following $\Z^{r}\times \C^r$ lattice of automorphisms:
\be\label{eq:spectralflow}
\sigma_{\lambda, \mu}(N^a(z))=N^a(z)-\frac{\mu^a}{z},\qquad \sigma_{\lambda, \mu}(E^a(z))=E^a(z)-\frac{\lambda^a}{z},\qquad \sigma_{\lambda,\mu}\psi^{i}_\pm(z)=z^{\mp \sum_{a}\rho^{ia}\lambda^a}\psi^{i}_\pm(z).
\ee
Here $\lambda\in \C^r$ and $\mu\in \C^r$ such that $\rho(\lambda)\in \Z^n$. We call these morphisms spectral-flow morphisms. Starting from a module $M$ of $\wh\grho$, we can obtain a spectral-flowed module $\sigma_{\lambda, \mu}(M)$, whose underlying vector space is $M$ and whose module structure is given by
\be
v* m=\sigma_{\lambda, \mu}^{-1}(v)m, \qquad v\in \wh\grho, ~m\in M. 
\ee
This functorial assignment maps $V(\grho)$ modules to $V(\grho)$ modules.  This is useful in constructing VOA extensions of $V(\grho)$. For example, the VOA $\CV_\rho^B$ is defined as a module of $V(\grho)$ using a $\Z^r$ lattice of spectral-flowed vacuum modules
\be
\CV_\rho^B:=\bigoplus_{\lambda\in\Z^r} \sigma_{\lambda, \rho^\trans\rho\lambda} V(\grho).
\ee

\subsection{The Kazhdan-Lusztig Category}

Given a VOA $V$ and a category $\CU$ of $V$ modules, the extensive work of Huang-Lepowsky-Zhang \cite{huang2014logarithmic1, huang2010logarithmic2, huang2010logarithmic3, huang2010logarithmic4, huang2010logarithmic5, huang2010logarithmic6, huang2011logarithmic7, huang2011logarithmic8} shows that under certain conditions on $V$ and $\CU$, there is a vertex-tensor category structure on $\CU$: the structure of a braided tensor category with twist, where the tensor product of two modules is defined as the universal object of logarithmic intertwining operators.

Specifically, given three modules $M, N$ and $P$, a logarithmic intertwining operator is a map
\be \label{Y-int}
Y:  M\otimes  N\to P\{z\}[\log z]
\ee
that satisfies a similar set of properties the state-operator map of a VOA needs to satisfy. 
The curly bracket means a formal power series with any possible complex  exponents, \ie\ sums of the type
\[
\sum_{\substack{y \in \mathbb C \\ n \in \mathbb N}} P_{y, n}\, z^y\, (\log(z))^n, \qquad  P_{y, n} \in P.
\]
For a full definition, see \cite[Definition 3.10]{huang2010logarithmic2}. Let $\mathrm{Int}( M\otimes N, P)$ be the set of $ M,  N,  P$ intertwining operators \eqref{Y-int}. For fixed $ M$ and $ N$, 
 a universal object for the space of intertwining operators is a module denoted $ M\otimes_{\CU}  N$, together with an intertwining operator $Y: M\otimes N\to  M\otimes_{\CU} N \{z \}[\log z]$ that induces an isomorphism
\be
\mathrm{Int}( M\otimes N, P)\cong \mathrm{Hom}( M\otimes_{\CU} N,  P),
\ee
for any other module $ P$.

It is difficult in general to verify the existence of such a universal object and, in particular, the existence of tensor category. Usually one can show this existence if the category $\mathcal U$ satisfies certain finiteness conditions \cite[Theorem 3.3.4]{creutzig2021tensor}. 
The best known examples are first the Kazhdan-Lusztig category for an affine Lie algebra associated with a semisimple Lie algebra, which  at most levels satisfies the necessary finiteness conditions and has the structure of a braided tensor category; and second the Virasoro algebra at any central charge has a category of modules satisfying all the necessary finiteness conditions and hence has the structure of a braided tensor category \cite{Creutzig:2020zvv}. In general it is quite a difficult problem to establish  the structure of a braided tensor category on a category of modules of a VOA, e.g. only very recently this result has been established for the category of weight modules of the affine VOA of $\mathfrak{sl}_2$ at admissible levels \cite{Creutzig:2023rlw}.

In \cite{BCDN}, we studied the Kazhdan-Lusztig category for the affine Lie superalgebra $V(\grho)$, and showed that this category satisfies the finiteness conditions and therefore has the structure of a braided tensor category. Let us recall the definition here. A generalized module of $V(\grho)$ (or simply a module of the affine Lie superalgebra $\wh\grho$ where $K = \text{Id}$) is called \emph{finite-length} if it has a finite composition series whose composition factors are simple. It is called \emph{grading-restricted} if it is graded by generalized conformal weights (the generalized eigenvalues of $L_0$), the weight spaces are all finite-dimensional, and the generalized conformal weights are bounded from below. 

\begin{Def}
We define $KL_\rho$, the Kazhdan-Lusztig category of $V(\grho)$, to be the category of finite-length, grading-restricted modules of $V(\grho)$.

\end{Def}
 
The following was proven in \cite[Theorem 8.19]{BCDN}:

\begin{Thm}

The category $KL_\rho$ is a rigid braided tensor category, such that the braided tensor structure is defined by logarithmic intertwining operators and the dual is defined by the contragredient dual. 

\end{Thm}

Since $V(\grho)$ has central elements $E^a_{0}$, the category $KL_\rho$ admits the following decomposition:
\be
KL_{\rho}=\bigoplus_{\nu\in \C^r} KL_{\rho, \nu},
\ee
where $KL_{\rho, \nu}$ is the sub-category of objects in $KL_\rho$ where the generalized eigenvalue of $E^a_{0}$ is $\nu_a$, or in other words, $(E^a_{0}-\nu_a)^n=0$ for some sufficiently large $n$. We proved in \cite{BCDN} that for generic $\nu$, the category $KL_{\rho, \nu}$ is equivalent to the corresponding category of finite-dimensional modules of $\grho$. Namely, let $\grho\Mod$ be the category of finite-dimensional modules of $\grho$, which also admits a decomposition:
\be
\grho\Mod=\bigoplus_{\nu\in \C^r} \grho\Mod_{\nu},
\ee
where $\grho\Mod_{\nu}$ is the sub-category where $E^a$ acts with generalized eigenvalue $\nu_a$. We have used the following induction functor to define $V(\grho)$:
\be
\mathrm{Ind}(M):=U(\wh \grho )\otimes_{U(\wh\fg_{\rho, \geq 0})}\lp M\otimes \C_1\rp,
\ee
where $M$ is a module of $\grho$, viewed as a module of $\wh\fg_{\rho, \geq 0}$ via the morphism $\wh\fg_{\rho, \geq 0}\to \fg_\rho$, and $\C_1$ is a module of the center $\C K$ on which $K$ acts as $1$.  It is clear that $\mathrm{Ind}$ sends an object in $\grho\Mod_\nu$ to a module of $V(\grho)$ where the generalized eigenvalues of $E^a_0$ are $\nu_a$. It would be in $KL_{\rho, \nu}$ if it is finite-length. We proved in \cite[Proposition 8.24]{BCDN} that this is the case. In fact, we have the following statement:

\begin{Prop}\label{Propindgeneric}
    The image of the induction functor $\mathrm{Ind}$ lies in $KL_\rho$. When $\sum_a \rho^{ia}\nu^a\notin \Z$ or $\sum_a \rho^{ia}\nu^a=0$ for all $i$, the induction functor induces an equivalence of abelian categories:
    \be
\mathrm{Ind}: \grho\Mod_\nu\simeq  KL_{\rho, \nu}.
    \ee
\end{Prop}

Namely, for a generic $\nu$, we can understand the abelian category $KL_{\rho, \nu}$ very well by studying $\grho\Mod_\nu$. How do we study $KL_{\rho, \nu}$ when $\nu$ doesn't satisfy the requirement of Proposition \ref{Propindgeneric}? The answer is through spectral flow. Recall the automorphisms in equation \eqref{eq:spectralflow}. For each $(\lambda,\mu)$, the spectral flow induces an equivalence:
\be
\sigma_{\lambda,\mu}: KL_{\rho, \nu}\simeq KL_{\rho, \nu+\lambda}.
\ee
One can show that for each $\nu$, there exists some $\lambda$ such that $\rho (\lambda)\in \Z^r$ and that $\nu+\lambda$ satisfies the requirement of Proposition \ref{Propindgeneric}. Consequently, we have equivalences of abelian categories as follows:
\be
\btik
KL_{\rho,\nu}\rar{\sigma_{\lambda,\nu}} & KL_{\rho, \nu+\lambda} & \lar{\mathrm{Ind}} \grho\Mod_{\nu+\lambda}
\etik
\ee

Notice that the spectral flow automorphism $\sigma_{\lambda,\mu}$ is generated by the action of the  zero mode of the Heisenberg current $\sum_a \lambda^aN^{a}(z)+(\mu^a-\rho^\trans\rho\lambda^a) E^a(z)$, via Li's delta operator \cite{Lidelta}. Consequently, the spectral flows of the vacuum, $\sigma_{\lambda,\mu}V(\grho)$, are simple current modules, and satisfy simple fusion rules:
\be
\sigma_{\lambda,\mu}V(\grho)\otimes_{KL_\rho} \sigma_{\lambda',\mu'}V(\grho)\cong \sigma_{\lambda+\lambda',\mu+\mu'}V(\grho).
\ee
The fusion rule above gives the module $\CV_\rho^B$ the VOA structure. Moreover, given a module $M$, the spectral flow $\sigma_{\lambda,\mu} M$ is isomorphic to $\sigma_{\lambda,\mu}V(\grho)\otimes_{KL_\rho} M$, via the intertwining operator defined by Li's delta operator.

\subsection{Free-field Realization of $V(\grho)$}

One important way to study a VOA is by embedding it into simpler VOAs. In \cite{BCDN}, we constructed an embedding of $V(\grho)$ into  lattice vertex algebras, namely an extension of Heisenberg VOA by Fock modules. This generalized the well-studied special case of $\grho = \mathfrak{gl}(1|1)$ and extensions thereof  \cite{schomerus2006gl, creutzig2013w,creutzig2021duality}.
We recall the construction here. Let $H_{\{X^a,Y^a,Z^i\}}$ denote the Heisenberg VOA associated with the $2r+n$ dimensional vector space with basis $\{X^a,Y^a,Z^i\}$ for $a=1,...,r$ and $i=1,...,n$, and symmetric bilinear form
\begin{equation}
    B( X^a,Y^b)=\delta^{ab}\,,\qquad B( Z^i,Z^j)=\delta^{ij}\,.
\end{equation}
Namely, the VOA $H_{\{X^a,Y^a,Z^i\}}$ is strongly generated by $\pd X^a, \pd Y^a$ and $\pd Z^i$ with the following OPE:
\be
\pd X^a(z)\pd Y^b(w)=\frac{\delta^{ab}}{(z-w)^2},\qquad \pd Z^i(z)\pd Z^j(w)=\frac{\delta^{ij}}{(z-w)^2}.
\ee
Here and below, $\pd$ denotes derivative in the variable $z$. We choose the following conformal element $L$:
\begin{equation} \label{L-B-ff}
\begin{split} 
L =\ &\frac 1 2 \Big(\norm{ \sum_{1\leq a\leq r} (\partial X^a\partial Y^a + \partial Y^a \partial X^a) +\sum_{1\leq i\leq n}(\partial Z^i)(\partial Z^i)}\Big) +\\
& \frac 1 2 \sum_{1\leq i\leq n} \Big(\sum_{1\leq a\leq r}\rho^{ia}\partial^2Y^a-\partial^2Z^i\Big)\,.
\end{split}
\end{equation}
Let
\be \CV_Z^\rho = \bigoplus_{\lambda\in \Z^n} \CF_{\lambda\cdot Z}^{X,Y,Z} \ee
be the lattice VOA obtained by extending $H_{\{X^a,Y^a,Z^i\}}$ by Fock modules generated by $\norm{e^{\lambda Z}}$ for all $\lambda\in \Z^n$, where $\lambda Z=\sum_{i=1}^n \lambda^i Z^i$. It is easy to see that the VOA $\CV_Z^\rho$ is isomorphic to a tensor product:
\be
\CV_Z^\rho\cong H_{\{X^a,Y^a\}}\otimes \FF^{\otimes n},
\ee
where $H_{\{X^a,Y^a\}}$ is the Heisenberg VOA generated by $\pd X^a, \pd Y^a$ and $\FF^{\otimes n}$ is a set of free fermion VOAs, generated by $\norm{e^{\pm Z^i}}$. The assignments:
\begin{equation}\label{ffRealization}
    \begin{aligned}
     &N^a\mapsto \partial X^a +\sum_{1\leq i\leq n} \rho^{ia}\partial Z^i\,,\\
     &E^a\mapsto  \partial Y^a\,,\\
     &\psi^{i}_+(z)\mapsto \norm{e^{Z^i}}\,,\\
     &\psi^{i}_-\mapsto \norm{\sum_{1\leq a\leq r} \rho^{ia}\partial Y^a e^{-Z^i}}+\pd\norm{e^{-Z^i}}
    \end{aligned}
\end{equation}
defines an embedding of the Kac-Moody VOA $V(\grho)$ into the lattice VOA $\CV_Z^\rho$. It is in fact possible to identify the image of this embedding with the kernel of a certain set of screening operators. Recall that the lattice VOA $\CV_Z^\rho$  has Fock modules
\be \CV_{Z,\mu X+\nu Y}^\rho = \bigoplus_{\lambda\in \Z^n} \CF_{\mu \cdot X+\nu\cdot Y+\lambda\cdot Z}^{X,Y,Z} \ee
labelled by $\mu,\nu\in \C^r$, which are lifts of the Fock modules of the Heisenberg VOA $H_{\{X^a,Y^a, Z^i\}}$ to modules of its simple current extension $\CV_Z^\rho$ (see Section \ref{sec:ungauge} for details about simple current extensions and lifting functors). Here $\mu\cdot X=\sum_{a=1}^r \mu^a X^a$. When $\rho(\mu)\in \Z^n$, define intertwiners $S^i(z): \CV_{Z,\mu X+\nu Y}^\rho \to \CV_{Z,\mu X+\nu Y-(\rho Y)^i}^\rho(\!(z)\!)$ given by
\be S^i(z)=\norm{e^{Z^i(z)-\sum_a \rho^{ia}Y^a(z)}}\,, \label{def-B-screen} \ee
with corresponding screening charges $S^i_0 = \frac{1}{2\pi i}\oint S^i(z)\,dz$. Note that this is only well-defined when $\rho(\mu)\in \Z^n$, since this ensures that $S^i$ is integer moded. We have the following proposition, proven in \cite[Proposition 4.2]{BCDN}:

\begin{Prop}\label{PropScreen}
The image of the embedding $V(\grho)\hookrightarrow \CV_Z^\rho$ is the kernel of the screening operators:
\begin{equation}
   V(\grho)\,\cong\, \bigcap_{i=1}^n \mathrm{ker}\, S^i_0\big|_{\CV_Z^\rho}\,.
    \end{equation}
\end{Prop}

One can restrict the modules $\CV_{Z, \lambda X+\mu Y}^\rho$ of $\CV_Z^\rho$ to $V(\grho)$. When $\rho(\lambda)\in \Z^r$, one can identify:
\be
\sigma_{\lambda, \mu+\rho^\trans\rho \lambda}V(\grho)\cong \bigcap_{i=1}^n \mathrm{ker}\, S^i_0\big|_{\CV^\rho_{\lambda X+\mu Y, Z}}\,.
\ee
In fact, for each $\mu,\nu$, if $\rho(\mu)$ satisfies the generic condition of Proposition \ref{Propindgeneric}, then one can show that $\CV_{Z, \mu X+\nu Y}^\rho$ can be identified with $\mathrm{Ind} (V_{\nu-\rho, \mu})$ where $V_{\nu, \mu}$ is the Verma module of $\grho$ generated by a single vector $v$ such that $N^av=\nu^a v, E^av=\mu^a v$ and $\psi_-^i v=0$. Here $\nu-\rho$ is the vector $(\nu^a-\sum_i \rho^{ia})$. For example, the free field algebra $\CV_Z^\rho$ is identified with the module $\mathrm{Ind}(V_{-\rho, 0})$, and is the Wakimoto free-field realization of the Lie superalgebra $\grho$. 

Having such a free-field realization significantly simplifies the study of modules of $V(\grho)$, since one can now construct intertwining operators of modules of $V(\grho)$ using intertwining operators of $\CV_Z^\rho$ modules. For example, we have the following statement, whose proof is identical to the proof of \cite[Appendix B1]{GarnerNiu}.

\begin{Prop}
    Every object in $KL_\rho$ is a subquotient (namely a quotient of a sub-module) of an object in $\CV_Z^\rho\Mod$, restricted to $KL_\rho$. 
    
\end{Prop}

One can use such a statement to compute, for example, monodromy as follows:

\begin{Cor}\label{Cor:monsimple}
    Let $M$ be an object in $KL_\rho$, then for any $\lambda,\mu\in \C^r$ such that $\rho (\lambda)\in \Z^r$, the monodromy:
    \be
    \btik
\sigma_{\lambda, \mu} V(\grho) \otimes_{V(\grho)} M \rar{c} & M \otimes_{V(\grho)} \sigma_{\lambda, \mu} V(\grho)\rar{c} & \sigma_{\lambda, \mu} V(\grho) \otimes_{V(\grho)} M
\etik
    \ee
    is given by $\mathrm{Id}\otimes \exp (2\pi i \sum_a \lambda^a N^a_0+(\mu-\rho^\trans\rho\lambda)^a E^a_0)$. 
\end{Cor}

The precise computation of this statement for $\rho=1$ is contained in the Appendix B of \cite{GarnerNiu}, and the statement for general $\rho$ follows in exactly the same way. 

\begin{Exp}
    Let $\rho=1$, in which case $V(\grho)=V(\fgl (1|1))$, and this free field realization is well-known in the literature; see for instance \cite{schomerus2006gl,creutzig2021duality}. This free field realization embeds $V(\fgl (1|1)$ into $\CV_Z^1=H_{X,Y}\otimes V_{bc}$ and the image is the kernel of the screening operator $S=\oint\mathrm{d}z b(z)e^{-Y(z)}$. 
\end{Exp}

\subsection{Simple-current Extension and its Application to $V(\grho)$}\label{sec:ungaugeVOA}

Another important consequence of the free field realization is that one can realize $V(\grho)$ as a simple current extension of many copies of $V(\fgl (1|1))$. Let $V:=V(\fgl (1|1))^{\otimes n}$, which has $\C^n\times \Z^n$ lattice of spectral flow automorphisms $\sigma_{\lambda,\mu}$ where $\lambda\in \Z^n$ and $\mu\in \C^n$. These are defined as in equation \eqref{eq:spectralflow}, for $\rho=\mathrm{Id}_n$ the $n\times n$ identity matrix. Recall the splitting of the exact sequence \eqref{eqsplitexactsequence}. The following direct sum is an algebra object (in a suitable completion of $KL_{1^n}$, the Kazhdan-Lusztig category of $V$):
\be
\CA_\rho:=\bigoplus_{\lambda, \mu\in \Z^{n-r}} \sigma_{\wt{\rho} (\lambda), 
\tau^\trans (\mu)+\wt{\rho} (\lambda)}V
\ee
Namely, this direct sum defines a simple current VOA extension of $V$. By using the free-field realization of $V$ and field redefinition, it was shown in \cite[Theorem 5.6.]{BCDN} that the following is true:

\begin{Prop}\label{Prop:ungaugeVOA}
    There exists a lattice VOA $\CW_\rho$ associated with a self-dual lattice such that there is a VOA isomorphism:
\be
V(\grho)\otimes\CW_\rho\cong \CA_\rho.
\ee
Consequently $V(\grho)\otimes \CW_\rho$ is a simple current extension of $V=V(\fgl(1|1))^{\otimes n}$. This simple current extension is compatible with the free field realization in the sense that the following diagram is commutative:
\be
\btik
V(\fgl (1|1))^{\otimes n}\rar\dar & (\CV_Z^1)^{\otimes n}\dar\\
V(\grho)\otimes \CW_\rho\rar & \CV_Z^\rho\otimes \CW_\rho
\etik
\ee
Here $\CV_Z^1=H_{X,Y}\otimes V_{bc}$ is the free field realization of $V(\fgl (1|1))$. 

\end{Prop}

\begin{Rem}
    The lattice defining $\CW_\rho$ is the lattice $(\Z^{n-r})^2$ with bilinear form:
    \be
\left(\begin{array}{cc} 0 & 1 \\ 1 & \tau^\trans\tau \end{array}\right)
    \ee
\end{Rem}

What this allows us to do is to relate the category of representations of $V(\grho)$ with that of $V(\fgl(1|1))^{\otimes n}$. We will turn to this point in Section \ref{sec:ungauge} when we discuss un-gauging operation for quantum groups.

\section{The Quantum Group $U_\rho^E$}\label{sec:QG}

In this section, we introduce the quantum group $U^E_\rho$. We will show that this algebra has the structure of a Hopf algebra, and that its category of finite-dimensional modules admits an $R$-matrix, which can be expressed in terms of an element in the algebraic closure of $U^E_\rho\otimes U^E_\rho$. When $\rho=1$, the quantum group $U^E_\rho$ is the unrolled quantum supergroup associated with $\mathfrak{gl}(1|1)$. In fact, we will show that $U^E_\rho$ is a subquotient of many copies of $U^E_1$, the unrolled quantum group of $\fgl(1|1)$, such that the Hopf algebra structure and the $R$-matrix of $U^E_\rho$ can be induced from $U_1^E$.

\begin{Rem}

We remind the readers of the Koszul sign convention. One could include the parity operator $(-1)^F$ in what follows, but we choose not to do so for simplicity.

\end{Rem}

\subsection{Definition and Structures}

We fix the charge matrix $\rho$, assuming it is faithful. The Lie superalgebra $\fg_\rho$ is generated by $N^a, E^a, \psi^{i}_\pm$ with relations:
\be
[N^a, \psi^{i}_\pm]=\pm \rho^{ia} \psi^{i}_\pm,\qquad \{\psi^{i}_+, \psi^{i}_-\}=\sum \rho^{ia}E^a. 
\ee
The quantum supergroup $U_\rho^E$ deforms the above commutation relation into the following:
\be
 \{\psi^{i}_+, \psi^{i}_-\}=\sum \rho^{ia}E^a\qquad \rightsquigarrow\qquad  \{\psi^{i}_+, \psi^{i}_-\}=e^{2\pi i(\sum \rho^{ia}E^a)}-1
\ee
In this formula, we could replace $e^{2\pi i(\sum \rho^{ia}E^a)}$ by $e^{\frac{2\pi i}{k}(\sum \rho^{ia}E^a)}$ for any non-zero $k$. It does not change the underlying algebra since we can always rescale $E^a$. The exponential should be understood as a power series in $E$, and only makes sense on finite-dimensional representations. This leads to the following definition:

\begin{Def}
We define $U_\rho^E$ as the superalgebra generated in even degrees by $N^a, E^a, K_{E^a}^{\pm}$, $K_{N^a}^\pm$, and in odd degrees by $\psi^{i}_\pm$, such that the non-trivial commutation relations are:
\be
[N^a, \psi^{i}_\pm]=\pm \rho^{ia}\psi^{i}_\pm,\qquad  \{\psi^{i}_+, \psi^{i}_-\}=\prod_a K_{E^a}^{\rho^{ia}}-1,\qquad K_{a}^+K_{a}^{-}=K_{a}^{-}K_a^+=1.
\ee
Here $K_a$ can be either $K_{N^a}$ or $K_{E^a}$. We also impose the condition $K_{N^a}=e^{2\pi iN^a}$ and $K_{E^a}=e^{2\pi i E^a}$, which is well-defined on finite-dimensional modules of $U_\rho^E$. 
\end{Def}

\begin{Rem}\label{Rem:findim}

The relation $K_{N^a}=e^{2\pi iN^a}$ is not well-defined on the algebra $U_\rho^E$. It is well-defined on finite-dimensional modules. We impose this condition whenever we consider finite-dimensional modules of $U_\rho^E$, and we \textbf{only} consider finite-dimensional modules in the following. 

\end{Rem}

The motivation behind considering $U_\rho^E$ is the observation made in \cite{niu2023representation} that $U_\rho^E\Mod$ is equivalent to $KL_\rho$ as an abelian category, and that it is a Hopf algebra. We will prove the abelian equivalence in the next section. We first show the following result. 

\begin{Thm}\label{Thm:quasitriUq}
The algebra $U_\rho^E$ has the structure of a Hopf algebra, and moreover has an $R$-matrix in the algebraic closure of $(U_\rho^E)^{\otimes 2}$ which has a well-defined evaluation on finite-dimensional modules. 

\end{Thm} 

\begin{proof}
Let us give the explicit Hopf algebra structure. We define coproduct by:
\be
\begin{aligned}
\Delta (N^a)=& N^a\otimes 1+1\otimes N^a, \qquad \Delta (E^a)=E^a\otimes 1+1\otimes E^a, \qquad \Delta (K_a^\pm)=K_a^\pm\otimes K_a^\pm,\\
&\Delta (\psi_+^i)=\psi_+^i\otimes 1+ 1\otimes \psi_+^i, \qquad \Delta (\psi_-^i)=\psi_-^i\otimes \prod_{1\leq a\leq r} K_{E^a}^{\rho^{ia}}+1\otimes \psi_-^i,
\end{aligned}
\ee    
counit by:
\be
\epsilon(E^a)=\epsilon(N^a)=\epsilon(\psi_\pm^i)=0,\qquad \epsilon(K_{E^a}^{\pm})=\epsilon(K_{N^a}^\pm)=1
\ee
antipode:
\be
\begin{aligned}
&S(N^a)=-N^a,\qquad S(E^a)=-E^a,\qquad S(K_a^\pm)=K_a^\mp,\\
& S(\psi_+^i)=-\psi_+^i,\qquad S(\psi_-^i)=-\prod_{1\leq a\leq r} K_{E^a}^{-\rho^{ia}}\psi_-^i,
\end{aligned}
\ee
as well as the following $R$-matrix:
\be
R=e^{2\pi i \sum_a N^a\otimes E^a}\prod_{1\leq i\leq n}\lp 1-\psi_+^i\otimes \lp \prod_{1\leq a\leq r} K_{E^a}^{-\rho^{ia}}\psi_-^i\rp \rp.
\ee
The fact that these satisfy Hopf algebra axioms as well as the relation with the $R$-matrix will follow from the fact that $U_\rho^E$ can be constructed as a relative Drinfeld double, see Section \ref{subsubsec:Nicholsrho}. This $R$ can be evaluated on any finite dimensional modules, since on such the exponential $e^{2\pi i \sum_a N^a\otimes E^a}$ is well-defined. 

\end{proof}

\begin{Rem}
  There is in fact more than one $R$-matrix. One can always add to $R$ a central element of the form $e^{2\pi i \sum_{a, b} \omega^{ab} E^a\otimes E_b}$ for some symmetric tensor $\omega$. These different choices give equivalent BTC via a Drinfeld twist defined by $e^{\pi i \sum \omega^{ab} E^aE_b}$. 
\end{Rem}

We in fact have the following:

\begin{Prop}\label{Prop:qgungauge}
    The algebra $U_\rho^E$ is a subquotient of $(U_1^E)^{\otimes n}$, such that the above Hopf algebra structure and $R$-matrix can be induced from $(U_1^E)^{\otimes n}$. 
    
\end{Prop}

\begin{proof}
    Let us write the generators of $(U_1^E)^{\otimes n}$ by $N^i, \psi_\pm^i, E^i$ and $K_{N^i}, K_{E^i}$ with $1\leq i\leq n$. Recall the splitting of the short exact sequence in equation \eqref{eqsplitexactsequence}
\be
\begin{tikzcd}
    0\rar & \mathbb{Z}^r\arrow[r, shift right, swap, "\rho"]& \arrow[l, shift right, swap, "\wt{\tau}"] \mathbb{Z}^n \arrow[r, shift right, swap, "\tau"] & \arrow[l, shift right, swap, "\wt{\rho}"]\mathbb{Z}^{n-r}\rar & 0
    \end{tikzcd}
\ee
Let $\overline{U}$ be the quotient of $(U_1^E)^{\otimes n}$ by the following central relations:
    \be
\prod_i K_{N^i}^{\wt{\rho}^{i\alpha}}=1, \qquad \prod_i K_{E^i}^{\tau^{\alpha i}}=1.
\ee
It is clearly a Hopf algebra quotient of $(U_1^E)^{\otimes n}$. Denote by $U$ the subalgebra of $\overline{U}$ generated by the following elements:
    \be
    N^a:=\sum_i \rho^{ia}N^i, \qquad K_{N^a}^{\pm}:=\prod K_{N^i}^{\pm\rho^{ia}},\qquad \psi_\pm^i, \qquad E^a=\sum_i \wt{\tau}^{ai}E^i, \qquad K_{E^i}^\pm 
    \ee
This is a Hopf subalgebra of $\overline{U}$.   We first claim that $U\cong U_\rho^E$ as an algebra. Indeed, the algebra $U$ is generated by $N^a, \psi^{i,\pm}$ and $E^a$, such that $[N^a, \psi^{i,\pm}]=\pm \rho^{ia}\psi^{i,\pm}$. Moreover, the commutation relation of $\psi^{i,\pm}$ reads:
\be
\{\psi^{i, +}, \psi^{i, -}\}=K_{E^i}-1,
\ee
We can rewrite the element $K_{E^i}$ in the following way:
\be\label{eq:KErelation}
K_{E^i}=\prod_j K_{E^j}^{\sum \rho^{ia}\wt{\tau}^{aj}+\sum \wt{\rho}^{i \alpha}\tau^{\alpha j}}=\prod_j K_{E^j}^{\sum \rho^{ia}\wt{\tau}^{aj}}.
\ee
Here the first equality follows from the identity $\rho\wt{\tau}+\wt{\rho}\tau=\mathrm{Id}_n$ and the second equality follows from the fact that $\prod_i K_{E^i}^{\tau^{\alpha i}}=1$ in $U$. Therefore, if we define $E^a=\sum \wt{\tau}^{ai}E^i$ and $K_{E^a}:=\prod_jK_{E_j}^{\wt{\tau}^{aj}}$, then the above becomes:
\be
\{\psi^{i, +}, \psi^{i, -}\}=\prod_a K_{E^a}^{\rho^{ia}}-1,
\ee
exactly the relation of $U_\rho^E$. What we have constructed is an algebra homomorphism $\varphi_\rho: U_\rho^E\to U$. To show that this is an isomorphism, we only need to show that this is an isomorphism of vector spaces. Let $C_\rho$ be the subalgebra generated by $N^a, E^a$ and $K_{N^a}^\pm, K_{E^a}^\pm$. Clearly the restriction of $\varphi_\rho$ to $C_\rho$ is an isomorphism. We can now give a filtration to $U_\rho^E$ and $U$ over $C_\rho$ by
\be
F^nU:=\{u\in U \big \vert u \text{ involves at most } n \text{ numbers of } \psi_\pm^i\}, ~\text{ similarly for } U_\rho^E.
\ee
It is clear that $\varphi_\rho$ is a map of filtered algebras. It is easy to see that the associated graded of $U_\rho^E$ and $U$ are isomorphic, as an algebra over $C_\rho$, to the tensor product $C_\rho\otimes \bigwedge\lp V\oplus V^*\rp$, and $\varphi_\rho$ is simply the identity map of this algebra. 

It is clear that most of the Hopf structure can be induced from $U_1^E$ to $U_\rho^E$. The only non-trivial part is the R matrix. Note that the R matrix in $(U_1^E)^{\otimes n}$ is given by:
\be
R=\prod_i e^{2\pi iN^i\otimes E^i} \prod_i (1-\psi_+^i\otimes K_{E^i}^{-1}\psi_-^i).
\ee
Using equation \eqref{eq:KErelation} and the definition of $K_{E^a}$, the second term becomes:
\be
\prod_i (1-\psi_+^i\otimes K_{E^i}^{-1}\psi_-^i)=\prod_i\lp 1-\psi_+^i\otimes \lp \prod_a K_{E^a}^{-\rho^{ia}}\psi_-^i\rp \rp 
\ee
For the first term, note that the summation $\sum N^i\otimes E^i$ is a summation of dual vectors, and therefore does not change under a change of basis. Consequently, we have an equality:
\be
\sum_i N^i\otimes E^i=\sum_a N^a\otimes E^a+\sum_\alpha N^\alpha\otimes E^\alpha.
\ee
Here we define $E^\alpha=\sum \tau^{\alpha i}E^i$ and $N^\alpha=\sum \wt{\rho}^{i\alpha}N^i$. This is because $\{N^a, N^\alpha\}$ above form a basis of $\fg$ whose dual basis is $\{E^a, E^\alpha\}$. Now when we consider the category of finite-dimensional modules, we impose the condition that $K_{N^\alpha}=e^{2\pi i N^\alpha}$, and therefore the relation $K_{N^\alpha}=1$ implies that $N^\alpha$ acts semi-simply with integer eigenvalues. Consequently, we have $e^{2\pi i N^\alpha\otimes E^\alpha}=1\otimes K_{E^\alpha}^{n^\alpha}=1$ since $K_{E^\alpha}=1$. We find, in the subquotient, the following R matrix:
\be
\overline{R}:=e^{2\pi i\sum_a N^a\otimes E^a}\prod_i\lp 1-\psi_+^i\otimes \lp \prod_a K_{E^a}^{-\rho^{ia}}\psi_-^i\rp \rp.
\ee
This is precisely the $R$-maitrx we give in Theorem \ref{Thm:quasitriUq}. 

\end{proof}

\begin{Rem}
    This proposition should be compared to Proposition \ref{Prop:ungaugeVOA}, both of which we will call the un-gauging, following physics motivation \cite{BCDN}. In Section \ref{sec:ungauge}, we will use these to extend the equivalence of Theorem \ref{KLrho} from $\rho=1$ to general $\rho$. 
    
\end{Rem}

\subsection{An Abelian Equivalence Between $U_\rho^E\Mod$ and $KL_\rho$}

In this section we prove the following result. 

\begin{Prop}\label{Prop:equivab}
    There is an equivalence of abelian categories:
    \be
U_\rho^E\Mod\simeq KL_\rho.
    \ee
\end{Prop}

\begin{proof}
    First, let $\nu$ be such that $\sum\rho^{ia}\nu^a=0$ or $\sum \rho^{ia}\nu^a\notin \Z$ for any $i$. Then we see that $KL_{\rho, \nu}\simeq \grho\Mod_\nu$. We now show that there is an equivalence:
    \be
U_\rho^E\Mod_\nu\simeq \grho\Mod_\nu.
    \ee
Here $U_\rho^E\Mod_\nu$ is the sub-category where the action of $E^a$ has generalized eigenvalue $\nu^a$. The idea is that the analytic function $f(x)=\frac{e^{2\pi i x}-1}{x}$ has an inverse away from any non-zero integer $x$, which we call $g$. Namely $g(x)f(x)=1$ for any $x\notin \Z$ or $x=0$. We now consider the following assignments
\be
\Psi_+^i:=\psi_+^i,\qquad \Psi_-^i:=g(\sum \rho^{ia}E^a)\psi_-^i.
\ee
Here to define $g(\sum \rho^{ia}E^a)$ we use a Taylor expansion of $g$ around $\sum \rho^{ia}\nu^a$ which is well-defined by the assumption on $\nu$ and that $\sum \rho^{ia} (E^a-\nu^a)$ is nilpotent. The new generators satisfy the relation:
\be
\{\Psi_+^i, \Psi_-^i\}=\sum \rho^{ia}E^a
\ee
and therefore together with $N^a, E^a$ generate an action of $\grho$. This induces an equivalence between $U_\rho^E\Mod_\nu$ and $\grho\Mod_\nu$. 

To finish the proof, we recall that for any $\lambda$ with $\rho(\lambda)\in \Z^n$, there is an equivalence $KL_\nu\simeq KL_{\nu+\lambda}$. We are done once we show the same is true for $U_\rho^E\Mod$. This is clear because of the following automorphism of $U_\rho^E$:
\be
N^a\mapsto N^a-\mu^a, \qquad E^a\mapsto E^a-\lambda^a,\qquad \psi_\pm^i\mapsto \psi_\pm^i.
\ee
This is an automorphism thanks to the fact that $e^{2\pi i \sum_a \rho^{ia}\lambda^a}=1$, so that the relation $K_{E^a}=e^{2\pi i E^a}$ is preserved on finite-dimensional modules. 

\end{proof}

In the remainder of this paper, we will upgrade this to an equivalence of braided tensor categories, as stated in Theorem \ref{KLrho}. As we have discussed in Section \ref{subsec:proofstrategy}, the proof strategy is naturally divided into two parts. We will review them in the following two sections, and prove necessary technical results along the way. 

\section{Free-field Realizations and Relative Drinfeld Centers}\label{sec:FF}

Following \cite{Creutzig:2021cpu,CLR23}, one important ingredient in relating modules of VOAs to modules of quantum groups is free-field realizations. In this section, we review the relevant parts of this strategy. We also develop some new technical results that will be applied in Section \ref{sec:KL1proof}.

\subsection{Tensor Categories Arising from Commutative Algebra Objects}

Let $V$ be a (super)VOA and $\CU$ a rigid braided tensor category of generalized modules of $V$. Let $W$ be a VOA containing $V$, such that as a module of $V$, $W$ belongs to $\CU$. Following \cite{Huang:2014ixa} and \cite[Theorem 3.13]{creutzig2020simple}, we can think of $W$ as defining a commutative (super)algebra object $A$ in $\CU$, such that the vertex operator of $W$ gives the multiplication map $m:A\otimes_\CU A\to A$, locality of the vertex operator gives rise to the commutativity $mc=m$ where $c:A\otimes_\CU A\to A\otimes_\CU A$ is the braiding, and Jacobi identity gives rise to the associativity of $m$. 

In this section, we will work under the general setting of $\CU$ being a braided tensor supercategory and $A$ a commutative (super)algebra object. We assume that $\CU$ is locally finite and rigid. We will comment on the parallel to VOA theory when needed. Given such an algebra object internal to $\CU$, one can consider $A\Mod(\CU)$, the category of $A$-modules in $\CU$, and the full-subcategory $A\lMod(\CU)$, the category of local $A$-modules. We give the definitions here, following \cite[Definition 2.29]{creutzig2017tensor}.  

\begin{Def}
   An object in $A\Mod(\CU)$ consists of an object $M\in \CU$, together with a morphism $a: A\otimes_\CU M\to M$ such that the following diagram commutes (in this diagram $a_{A,A,M}$ is the associativity): 
    \be\label{eq:assocA}
    \btik
A\otimes_\CU (A\otimes_\CU M)\dar{a_{A,A,M}}\rar{1\otimes a} & A\otimes_\CU M\arrow[dd, "a"]\\
(A\otimes_\CU A)\otimes_\CU M \dar{m\otimes 1} & \\
A\otimes_\CU M\rar{a} & M
\etik
    \ee
This commutative diagram translates to the graphical calculus as
\be
	\begin{matrix}
		\begin{tikzpicture}[scale=.75, out=up, in=down, line width=0.5pt]
			\node at (0, 4.3) {$\scriptstyle{A}$};
			\node at (1, 4.3) {$\scriptstyle{A}$};
			\node at (2, 4.3) {$\scriptstyle{M}$};
			\node at (1, -0.3) {$\scriptstyle{M}$};
			\node (m) at (1.5, 2.75) [draw,minimum width=25pt,minimum height=10pt,thick, fill=white] {$\scriptstyle{a}$};
			\node (e) at (1 , 1.25) [draw,minimum width=25pt,minimum height=10pt,fill=white] {$\scriptstyle{a}$};
             \draw [blue][line width=1pt] (1 , 3) to (1, 4);
             \draw [line width=1pt] (2 , 3) to (2, 4);
             \draw [red][line width=1pt] (0.75 , 1.5) to (0, 4);
             \draw [line width=1pt] (1.25 , 1.5) to (1.5, 2.5);
             \draw [line width=1pt] (1 , 0) to (1, 1);
		\end{tikzpicture}
	\end{matrix}
=
	\begin{matrix}
		\begin{tikzpicture}[scale=.75, out=up, in=down, line width=0.5pt]
			\node at (0, 4.3) {$\scriptstyle{A}$};
			\node at (1, 4.3) {$\scriptstyle{A}$};
			\node at (2, 4.3) {$\scriptstyle{M}$};
			\node at (1, -0.3) {$\scriptstyle{M}$};
			\node (m) at (0.5, 2.75) [draw,minimum width=25pt,minimum height=10pt,thick, fill=white] {$\scriptstyle{m}$};
			\node (e) at (1 , 1.25) [draw,minimum width=25pt,minimum height=10pt,fill=white] {$\scriptstyle{a}$};
             \draw [blue][line width=1pt] (1 ,3) to (1, 4);
             \draw [red][line width=1pt] (0 ,3) to (0, 4);
             \draw [line width=1pt] (1.25 , 1.5) to (2, 4);
             \draw [violet][line width=1pt] (0.75 , 1.5) to (0.5, 2.5);
             \draw [line width=1pt] (1 , 0) to (1, 1);
		\end{tikzpicture}
	\end{matrix}
\ee
    A morphism between $A$-modules $M$ and $N$ is a morphism $f:M\to N$ in $\CU$ making the following diagram commute:
    \be
\btik
A\otimes_\CU M\rar{a_M}\dar{1\otimes f} & M\dar{f}\\
A\otimes_\CU N\rar{a_N} & N  
\etik
    \ee
 that is graphically
 \be\label{eq:af=fagraph}
 \begin{matrix}
		\begin{tikzpicture}[scale=.75, out=up, in=down, line width=0.5pt]
			\node at (0, 4.3) {$\scriptstyle{A}$};
			\node at (1, 4.3) {$\scriptstyle{M}$};
			\node at (0.5, -0.3) {$\scriptstyle{N}$};
			\node (m) at (0.5, 2.75) [draw,minimum width=25pt,minimum height=10pt,thick, fill=white] {$\scriptstyle{1 \otimes f}$};
			\node (e) at (0.5 , 1.25) [draw,minimum width=25pt,minimum height=10pt,fill=white] {$\scriptstyle{a_N}$};
             \draw [line width=1pt] (1 , 3.1) to (1, 4);
             \draw [red][line width=1pt] (0 , 3.1) to (0, 4);
             \draw [green][line width=1pt] (1 , 1.55) to (1, 2.4);
             \draw [red][line width=1pt] (0 , 1.55) to (0, 2.4);
             \draw [green][line width=1pt] (0.5 , 0) to (0.5, 0.95);
		\end{tikzpicture}
	\end{matrix}
=
	\begin{matrix}
		\begin{tikzpicture}[scale=.75, out=up, in=down, line width=0.5pt]
			\node at (0, 4.3) {$\scriptstyle{A}$};
			\node at (1, 4.3) {$\scriptstyle{M}$};
			\node at (0.5, -0.3) {$\scriptstyle{N}$};
			\node (m) at (0.5, 2.75) [draw,minimum width=25pt,minimum height=10pt,thick, fill=white] {$\scriptstyle{a_M}$};
			\node (e) at (0.5 , 1.25) [draw,minimum width=25pt,minimum height=10pt,fill=white] {$\scriptstyle{f}$};
             \draw [line width=1pt] (1 , 3.05) to (1, 4);
             \draw [red][line width=1pt] (0 , 3.05) to (0, 4);
             \draw [line width=1pt] (0.5 , 1.6) to (0.5, 2.45);
             \draw [green][line width=1pt] (0.5 , 0) to (0.5, 0.9);
		\end{tikzpicture}
	\end{matrix}
 \ee 
The category $A\lMod (\CU)$ is the full subcategory of $A\Mod(\CU)$ whose objects $M$ are such that the following composition is equal to $a$.
\be
\btik
A\otimes_\CU M\rar{c_{A,M}} & M\otimes_\CU A\rar{c_{M,A}} & A\otimes_{\CU} M\rar{a} & M
\etik
\ee
  Namely $ac^2=a$. Graphically 
  \be
  \begin{matrix}
		\begin{tikzpicture}[scale=.75, out=up, in=down, line width=0.5pt]
			\node at (0, 6.3) {$\scriptstyle{A}$};
			\node at (1, 6.3) {$\scriptstyle{M}$};
			\node at (0.5, 2.2) {$\scriptstyle{M}$};
			\node (m) at (0.5, 3.75) [draw,minimum width=25pt,minimum height=10pt,thick, fill=white] {$\scriptstyle{a_M}$};
             \draw [blue][line width=1pt] (0 ,4.05) to (0, 6);
             \draw [line width=1pt] (1 ,4.05) to (1, 6);
             \draw [line width=1pt] (0.5 , 2.5) to (0.5, 3.45);
		\end{tikzpicture}
	\end{matrix}
	=	
	\begin{matrix}
		\begin{tikzpicture}[scale=.75, out=up, in=down, line width=0.5pt]
			\node at (0, 6.3) {$\scriptstyle{A}$};
			\node at (1, 6.3) {$\scriptstyle{M}$};
			\node at (0.5, 2.2) {$\scriptstyle{M}$};
			\node (m) at (0.5, 3.75) [draw,minimum width=25pt,minimum height=10pt,thick, fill=white] {$\scriptstyle{a_M}$};
			\draw [blue][line width=1pt] (1, 5) to (0, 6);
             \draw [white, line width=5pt] (0.3, 5.3) to (0.7, 5.7);
             \draw [line width=1pt] (0, 5) to (1, 6);
             \draw [line width=1pt] (1, 4.05) to (0, 5);
             \draw [white, line width=5pt] (0.3, 4.3) to (0.7, 4.7);
             \draw [blue][line width=1pt] (0, 4.05) to (1, 5);
             \draw [line width=1pt] (0.5 , 2.5) to (0.5, 3.45);
		\end{tikzpicture}
	\end{matrix}
 \ee
\end{Def}

\begin{Rem}
    In the following, when no confusion could occur, we will write $A\Mod$ and $A\lMod$, and dropping the $(\CU)$. 
\end{Rem}

Specializing to the case when $\CU$ is the category of modules of a vertex algebra $V$ such that $A$ defines a VOA extension $W$, then if one traces through the definition, an object in $A\Mod(\CU)$ is a module $M$ of $V$, together with an intertwining operator $\CY_M: W\otimes M\to M\{z\}[\log z]$ such that $\CY_M$ satisfies the following:\footnote{Here the equality is understood as an equality after appropriate re-expansion, and the same holds true for the remainder of the paper when considering equalities of intertwining operators. }
\be
\CY_M(w_1, z_1)\CY_M(w_2, z_2)m=\CY_M(Y_W(w_1, z_1-z_2)w_2,z_2)m. 
\ee
This $\CY_M$ corresponds to the module map $A\otimes_{\CU}M\to M$ and the above equation is the associativity of the action. The category $A\lMod(\CU)$ consists of those objects where $\CY_M$ is a map valued in $M\lpp z\rpp$ (since this requires that $\CY_M(w, e^{2\pi i}z)=\CY_M(w, z)$, so that $\CY_M$ is a series in $z^\pm$). Put simply, objects in the category $A\lMod(\CU)$ are generalized modules of $W$. In fact, we have the following statement (see \cite[Theorem 3.14]{creutzig2020simple}).

\begin{Prop}\label{Prop:AlModWMod}

The category $A\lMod(\CU)$ is the full subcategory of all generalized $W$-modules which, viewed as modules of $V$, lie in $\CU$. 
 
\end{Prop}

The following is a useful lemma that will become important later on. 

\begin{Lem}
    Let $M$ be in $A\Mod$. Then the following action of $A$:
    \be
\btik
A\FU M\rar{c_{A,M}^{-1}} & M\FU A\rar{c_{M,A}^{-1}} & A\FU M\rar{a} & M
\etik
    \ee
    which we denote by $ac^{-2}$, defines another $A$-module structure on $M$. We call the resulting module $M_c$. The assignment $M\mapsto M_c$ defines an endo-functor on $A\Mod$. 
\end{Lem}

\begin{proof}

This proposition follows from the following diagrammatic argument.

\be
\begin{aligned}
	\begin{matrix}
		\begin{tikzpicture}[scale=.75, out=up, in=down, line width=0.5pt]
			\node at (0, 6.3) {$\scriptstyle{A}$};
			\node at (1, 6.3) {$\scriptstyle{A}$};
			\node at (2, 6.3) {$\scriptstyle{M}$};
			\node at (1, -0.3) {$\scriptstyle{M}$};
			\node (m) at (1.5, 3.75) [draw,minimum width=25pt,minimum height=10pt,thick, fill=white] {$\scriptstyle{ac^{-2}}$};
			\node (e) at (1 , 1.25) [draw,minimum width=25pt,minimum height=10pt,fill=white] {$\scriptstyle{ac^{-2}}$};
             \draw [blue][line width=1pt] (1 ,4.1) to (1, 6);
             \draw [line width=1pt] (2 ,4.1) to (2, 6);
             \draw [red][line width=1pt] (0.75 , 1.55) to (0, 6);
             \draw [line width=1pt] (1.25 , 1.55) to (1.5, 3.4);
             \draw [line width=1pt] (1 , 0) to (1, 0.95);
		\end{tikzpicture}
	\end{matrix}
	&=	
	\begin{matrix}
		\begin{tikzpicture}[scale=.75, out=up, in=down, line width=0.5pt]
			\node at (0, 6.3) {$\scriptstyle{A}$};
			\node at (1, 6.3) {$\scriptstyle{A}$};
			\node at (2, 6.3) {$\scriptstyle{M}$};
			\node at (1, -0.3) {$\scriptstyle{M}$};
			\node (m) at (1.5, 3.75) [draw,minimum width=25pt,minimum height=10pt,thick, fill=white] {$\scriptstyle{a}$};
			\node (e) at (1 , 1.25) [draw,minimum width=25pt,minimum height=10pt,fill=white] {$\scriptstyle{a}$};
			\draw [red][line width=1pt] (0,3.5) to (0, 6);
			\draw [line width=1pt] (1, 5) to (2, 6);
             \draw [white, line width=5pt] (1.7, 5.3) to (1.3, 5.7);
             \draw [blue][line width=1pt] (2, 5) to (1, 6);
             \draw [blue][line width=1pt] (1, 4) to (2, 5);
             \draw [white, line width=5pt] (1.7, 4.3) to (1.3, 4.7);
             \draw [line width=1pt] (2, 4) to (1, 5);
             \draw [line width=1pt] (0.5 , 2.5) to (1.5, 3.5);
             \draw [white, line width=5pt] (1.2, 2.7) to (0.8, 3.1);
             \draw [red][line width=1pt] (1.5, 2.5) to (0, 3.5);
             \draw [red][line width=1pt] (0.5 , 1.5) to (1.5, 2.5);
             \draw [white, line width=5pt] (1.2, 1.8) to (0.8, 2.2);
             \draw [line width=1pt] (1.5, 1.5) to (0.5, 2.5);
             \draw [line width=1pt] (1 , 0) to (1, 1);
		\end{tikzpicture}
	\end{matrix}
	&=	
	\begin{matrix}
		\begin{tikzpicture}[scale=.75, out=up, in=down, line width=0.5pt]
			\node at (0, 6.3) {$\scriptstyle{A}$};
			\node at (1, 6.3) {$\scriptstyle{A}$};
			\node at (2, 6.3) {$\scriptstyle{M}$};
			\node at (0.5, -0.3) {$\scriptstyle{M}$};
			\node (m) at (1.5, 1.75) [draw,minimum width=25pt,minimum height=10pt,thick, fill=white] {$\scriptstyle{a}$};
			\node (e) at (0.5 , 0.75) [draw,minimum width=25pt,minimum height=10pt,fill=white] {$\scriptstyle{a}$};
			\draw [red][line width=1pt] (0, 4) to (0, 6);
			\draw [line width=1pt] (1, 5) to (2, 6);
            \draw [white, line width=5pt] (1.7, 5.3) to (1.3, 5.7);
            \draw [blue][line width=1pt] (2, 5) to (1, 6);
            \draw [blue][line width=1pt] (0, 3) to (2, 5);
              \draw [white, line width=5pt] (1.8, 4.2) to (1.4, 4.7);
             \draw [line width=1pt] (2, 4) to (1, 5);
             \draw [line width=1pt] (1, 3) to (2, 4);
             \draw [white, line width=5pt] (0.4, 3.4) to (0, 3.9);
             \draw [white, line width=5pt] (1.7, 3.2) to (1.3, 3.6);
             \draw [red][line width=1pt] (2, 3) to (0, 4);
             \draw [red][line width=1pt] (0 , 2) to (2, 3);
             \draw [white, line width=5pt] (0.9, 2.2) to (0.4, 2.6);
             \draw [white, line width=5pt] (1.6, 2.35) to (1.2, 2.75);
             \draw [blue][line width=1pt] (1, 2) to (0, 3);
             \draw [line width=1pt] (2, 2) to (1, 3);
             \draw [red][line width=1pt] (0 , 1) to (0, 2);
             \draw [line width=1pt] (0.5 , 1) to (1.5, 1.5);
             \draw [line width=1pt] (0.5 , 0) to (0.5, 0.5);
		\end{tikzpicture}
	\end{matrix}
 &=	
	\begin{matrix}
		\begin{tikzpicture}[scale=.75, out=up, in=down, line width=0.5pt]
			\node at (0, 6.3) {$\scriptstyle{A}$};
			\node at (1, 6.3) {$\scriptstyle{A}$};
			\node at (2, 6.3) {$\scriptstyle{M}$};
			\node at (1.5, -0.3) {$\scriptstyle{M}$};
			\node (m) at (0.5, 1.75) [draw,minimum width=25pt,minimum height=10pt,thick, fill=white] {$\scriptstyle{m}$};
			\node (e) at (1.5 , 0.75) [draw,minimum width=25pt,minimum height=10pt,fill=white] {$\scriptstyle{a}$};
			\draw [red][line width=1pt] (0, 4) to (0, 6);
            \draw [line width=1pt] (1, 5) to (2, 6);
            \draw [white, line width=5pt] (1.7, 5.3) to (1.3, 5.7);
			\draw [blue][line width=1pt] (2, 5) to (1, 6);
            \draw [blue][line width=1pt] (0, 3) to (2, 5);
             \draw [white, line width=5pt] (1.8, 4.2) to (1.4, 4.7);
             \draw [line width=1pt] (2, 4) to (1, 5);
             \draw [line width=1pt] (1, 3) to (2, 4);
             \draw [white, line width=5pt] (0.4, 3.4) to (0, 3.9);
             \draw [white, line width=5pt] (1.7, 3.2) to (1.3, 3.6);
             \draw [red][line width=1pt] (2, 3) to (0, 4);
            \draw [red][line width=1pt] (0 , 2) to (2, 3);
             \draw [white, line width=5pt] (0.9, 2.2) to (0.4, 2.6);
             \draw [white, line width=5pt] (1.6, 2.35) to (1.2, 2.75);
             \draw [blue][line width=1pt] (1, 2) to (0, 3);
             \draw [line width=1pt] (2, 2) to (1, 3);
             \draw [violet][line width=1pt] (1 , 1) to (0.5, 1.5);
             \draw [line width=1pt] (1.5 , 1) to (2, 2);
             \draw [line width=1pt] (1.5 , 0) to (1.5, 0.5);
		\end{tikzpicture}
	\end{matrix}
 &=	
	\begin{matrix}
		\begin{tikzpicture}[scale=.75, out=up, in=down, line width=0.5pt]
			\node at (0, 6.3) {$\scriptstyle{A}$};
			\node at (1, 6.3) {$\scriptstyle{A}$};
			\node at (2, 6.3) {$\scriptstyle{M}$};
			\node at (1.5, -0.3) {$\scriptstyle{M}$};
			\node (m) at (0.5, 1.75) [draw,minimum width=25pt,minimum height=10pt,thick, fill=white] {$\scriptstyle{m}$};
			\node (e) at (1.5 , 0.75) [draw,minimum width=25pt,minimum height=10pt,fill=white] {$\scriptstyle{a}$};
			\draw [red][line width=1pt] (0, 4) to (0, 6);
            \draw [line width=1pt] (1, 5) to (2, 6);
            \draw [white, line width=5pt] (1.7, 5.3) to (1.3, 5.7);
			\draw [blue][line width=1pt] (2, 5) to (1, 6);

             \draw [blue][line width=1pt] (0, 3) to (2, 5);
             \draw [white, line width=5pt] (1.8, 4.2) to (1.4, 4.7);
             \draw [line width=1pt] (2, 4) to (1, 5);
             \draw [line width=1pt] (1, 3) to (2, 4);
             \draw [white, line width=5pt] (0.4, 3.4) to (0, 3.9);
             \draw [white, line width=5pt] (1.7, 3.2) to (1.3, 3.6);
             \draw [red][line width=1pt] (2, 3) to (0, 4);

             \draw [blue][line width=1pt] (1, 2) to (0, 3);
             \draw [white, line width=5pt] (0.4, 2.3) to (0.9, 2.7);
             \draw [red][line width=1pt] (0 , 2) to (2, 3);
             \draw [white, line width=5pt] (1.6, 2.35) to (1.2, 2.75);
             \draw [line width=1pt] (2, 2) to (1, 3);

             \draw [violet][line width=1pt] (1 , 1) to (0.5, 1.5);
             \draw [line width=1pt] (1.5 , 1) to (2, 2);
             \draw [line width=1pt] (1.5 , 0) to (1.5, 0.5);
		\end{tikzpicture}
	\end{matrix}.
\end{aligned}
\ee

The first equality is the definition, the second one naturality of braiding, the third one that $M$ is an $A$-module, the fourth one that $A$ is commutative. This agrees with the composition of the multiplication of $A$ and $ac^{-2}$:

\be
\begin{aligned}
	\begin{matrix}
		\begin{tikzpicture}[scale=.75, out=up, in=down, line width=0.5pt]
			\node at (0, 6.3) {$\scriptstyle{A}$};
			\node at (1, 6.3) {$\scriptstyle{A}$};
			\node at (2, 6.3) {$\scriptstyle{M}$};
			\node at (1, -0.3) {$\scriptstyle{M}$};
			\node (m) at (0.5, 3.75) [draw,minimum width=25pt,minimum height=10pt,thick, fill=white] {$\scriptstyle{m}$};
			\node (e) at (1 , 1.25) [draw,minimum width=25pt,minimum height=10pt,fill=white] {$\scriptstyle{ac^{-2}}$};
             \draw [blue][line width=1pt] (1 ,4) to (1, 6);
             \draw [red][line width=1pt] (0 ,4) to (0, 6);
             \draw [line width=1pt] (1.25 , 1.55) to (2, 6);
             \draw [violet][line width=1pt] (0.75 , 1.55) to (0.5, 3.5);
             \draw [line width=1pt] (1 , 0) to (1, 0.95);
		\end{tikzpicture}
	\end{matrix}
	&=	 \begin{matrix}
		\begin{tikzpicture}[scale=.75, out=up, in=down, line width=0.5pt]
			\node at (0, 6.3) {$\scriptstyle{A}$};
			\node at (1, 6.3) {$\scriptstyle{A}$};
			\node at (2, 6.3) {$\scriptstyle{M}$};
			\node at (1.5, -0.3) {$\scriptstyle{M}$};
			\node (m) at (0.5, 4) [draw,minimum width=25pt,minimum height=10pt,thick, fill=white] {$\scriptstyle{m}$};
			\node (e) at (1 , 1.25) [draw,minimum width=25pt,minimum height=10pt,fill=white] {$\scriptstyle{a}$};
			\draw [red][line width=1pt] (0,4.25) to (0, 6);
			\draw [line width=1pt] (1.5, 3.5) to (2, 6);
             \draw [blue][line width=1pt] (1, 4.25) to (1, 6);
             \draw [line width=1pt] (0.5 , 2.5) to (1.5, 3.5);
             \draw [white, line width=5pt] (1.2, 2.9) to (0.8, 3.3);
             \draw [violet][line width=1pt] (1.5, 2.5) to (0.5, 3.75);
             \draw [violet][line width=1pt] (0.5 , 1.5) to (1.5, 2.5);
             \draw [white, line width=5pt] (1.2, 1.8) to (0.8, 2.2);
             \draw [line width=1pt] (1.5, 1.5) to (0.5, 2.5);
             \draw [line width=1pt] (1 , 0) to (1, 1);
		\end{tikzpicture}
	\end{matrix}
 &=	
	\begin{matrix}
		\begin{tikzpicture}[scale=.75, out=up, in=down, line width=0.5pt]
			\node at (0, 6.3) {$\scriptstyle{A}$};
			\node at (1, 6.3) {$\scriptstyle{A}$};
			\node at (2, 6.3) {$\scriptstyle{M}$};
			\node at (1.5, -0.3) {$\scriptstyle{M}$};
			\node (m) at (1.5, 2.25) [draw,minimum width=25pt,minimum height=10pt,thick, fill=white] {$\scriptstyle{m}$};
			\node (e) at (1.5 , 0.75) [draw,minimum width=25pt,minimum height=10pt,fill=white] {$\scriptstyle{a}$};
			\draw [line width=1pt] (0, 2) to (2, 6);
            \draw [white, line width=5pt] (1.4, 4.25) to (1.5, 4.65);
             \draw [white, line width=5pt] (0.5, 3.6) to (0.8, 3.8);
             \draw [blue][line width=1pt] (2, 2.5) to (1, 6);
             \draw [red][line width=1pt] (1, 2.5) to (0, 6);
             \draw [violet][line width=1pt] (1, 1) to (1.5, 2);
             \draw [white, line width=5pt] (1.1, 1.4) to (1.5, 1.6);
             \draw [line width=1pt] (2, 1) to (0, 2);
             \draw [line width=1pt] (1.5 , 0) to (1.5, 0.5);
		\end{tikzpicture}
	\end{matrix}
	&=	
	\begin{matrix}
		\begin{tikzpicture}[scale=.75, out=up, in=down, line width=0.5pt]
			\node at (0, 6.3) {$\scriptstyle{A}$};
			\node at (1, 6.3) {$\scriptstyle{A}$};
			\node at (2, 6.3) {$\scriptstyle{M}$};
			\node at (1.5, -0.3) {$\scriptstyle{M}$};
			\node (m) at (0.5, 1.75) [draw,minimum width=25pt,minimum height=10pt,thick, fill=white] {$\scriptstyle{m}$};
			\node (e) at (1.5 , 0.75) [draw,minimum width=25pt,minimum height=10pt,fill=white] {$\scriptstyle{a}$};
			\draw [red][line width=1pt] (0, 4) to (0, 6);
			\draw [line width=1pt] (1, 5) to (2, 6);
            \draw [white, line width=5pt] (1.7, 5.3) to (1.3, 5.7);
            \draw [blue][line width=1pt] (2, 5) to (1, 6);
            \draw [line width=1pt] (1, 3) to (1, 5);
             \draw [white, line width=5pt] (1.3, 3.2) to (0.9, 3.6);
             \draw [red][line width=1pt] (1.5, 3) to (0, 4);
             \draw [blue][line width=1pt] (2, 2.7) to (2, 5);
            \draw [red][line width=1pt] (0 , 2) to (1.5, 3);
             \draw [blue][line width=1pt] (1, 2) to (2, 2.7);
             \draw [white, line width=5pt] (1.4, 2.35) to (1, 2.75);
             \draw [white, line width=5pt] (1.3, 2.2) to (1.5, 2.5);
             \draw [line width=1pt] (1.5, 2) to (1, 3);
             \draw [violet][line width=1pt] (1 , 1) to (0.5, 1.5);
             \draw [line width=1pt] (1.5 , 1) to (1.5, 2);
             \draw [line width=1pt] (1.5 , 0) to (1.5, 0.5);
		\end{tikzpicture}
	\end{matrix}
 &=	
	\begin{matrix}
		\begin{tikzpicture}[scale=.75, out=up, in=down, line width=0.5pt]
			\node at (0, 6.3) {$\scriptstyle{A}$};
			\node at (1, 6.3) {$\scriptstyle{A}$};
			\node at (2, 6.3) {$\scriptstyle{M}$};
			\node at (1.5, -0.3) {$\scriptstyle{M}$};
			\node (m) at (0.5, 1.75) [draw,minimum width=25pt,minimum height=10pt,thick, fill=white] {$\scriptstyle{m}$};
			\node (e) at (1.5 , 0.75) [draw,minimum width=25pt,minimum height=10pt,fill=white] {$\scriptstyle{a}$};
			\draw [red][line width=1pt] (0, 4) to (0, 6);
			\draw [line width=1pt] (1, 5) to (2, 6);
            \draw [white, line width=5pt] (1.7, 5.3) to (1.3, 5.7);
            \draw [blue][line width=1pt] (2, 5) to (1, 6);
            \draw [blue][line width=1pt] (0, 3) to (2, 5);
             \draw [white, line width=5pt] (1.8, 4.2) to (1.4, 4.7);
             \draw [line width=1pt] (2, 4) to (1, 5);
             \draw [line width=1pt] (1, 3) to (2, 4);
             \draw [white, line width=5pt] (0.4, 3.4) to (0, 3.9);
             \draw [white, line width=5pt] (1.7, 3.2) to (1.3, 3.6);
             \draw [red][line width=1pt] (2, 3) to (0, 4);
             \draw [blue][line width=1pt] (1, 2) to (0, 3);
             \draw [white, line width=5pt] (0.4, 2.3) to (0.9, 2.7);
             \draw [red][line width=1pt] (0 , 2) to (2, 3);
             \draw [white, line width=5pt] (1.6, 2.35) to (1.2, 2.75);
             \draw [line width=1pt] (2, 2) to (1, 3);
             \draw [violet][line width=1pt] (1 , 1) to (0.5, 1.5);
             \draw [line width=1pt] (1.5 , 1) to (2, 2);
             \draw [line width=1pt] (1.5 , 0) to (1.5, 0.5);
		\end{tikzpicture}
	\end{matrix}.
\end{aligned}
\ee

Here the first equality is the definition, while all others are using only naturality of braiding and the hexagon identity. In all these diagrams, associativity isomorphisms are omitted since they don't present obstruction to the equalities of the diagrams. 

To show that $M\to M_c$ is an endofunctor, we need to show that if $f: M\to N$ is a morphism, then $f: M_c\to N_c$ is a morphism. We have:

\be
	\begin{matrix}
		\begin{tikzpicture}[scale=.75, out=up, in=down, line width=0.5pt]
			\node at (0, 6.3) {$\scriptstyle{A}$};
			\node at (1, 6.3) {$\scriptstyle{M}$};
		   \node at (0.5, 1.3) {$\scriptstyle{N}$};
			\node (m) at (0.5, 3.75) [draw,minimum width=25pt,minimum height=10pt,thick, fill=white] {$\scriptstyle{a_M}$};
			\draw [line width=1pt] (0, 5) to (1, 6);
             \draw [white, line width=5pt] (0.7, 5.3) to (0.3, 5.7);
             \draw [blue][line width=1pt] (1, 5) to (0, 6);
              \draw [blue][line width=1pt] (0, 4.05) to (1, 5);
             \draw [white, line width=5pt] (0.7, 4.3) to (0.3, 4.7);
              \draw [line width=1pt] (1, 4.05) to (0, 5);
              \draw [line width=1pt] (0.5 , 2.5) to (0.5, 3.45);
             \node (e) at (0.5 , 2.5) [draw,minimum width=25pt,minimum height=10pt,fill=white] {$\scriptstyle{f}$};
             \draw [line width=1pt] (0.5 , 1.5) to (0.5, 2.15);
		\end{tikzpicture}
	\end{matrix}
 =
 \begin{matrix}
		\begin{tikzpicture}[scale=.75, out=up, in=down, line width=0.5pt]
			\node at (0, 6.3) {$\scriptstyle{A}$};
			\node at (1, 6.3) {$\scriptstyle{M}$};
		   \node at (0.5, 1.3) {$\scriptstyle{N}$};
			\node (m) at (0.5, 3.75) [draw,minimum width=25pt,minimum height=10pt,thick, fill=white] {$\scriptstyle{1\otimes f}$};
			\draw [line width=1pt] (0, 5) to (1, 6);
             \draw [white, line width=5pt] (0.7, 5.3) to (0.3, 5.7);
             \draw [blue][line width=1pt] (1, 5) to (0, 6);
              \draw [blue][line width=1pt] (0, 4.05) to (1, 5);
             \draw [white, line width=5pt] (0.7, 4.3) to (0.3, 4.7);
             \draw [line width=1pt] (1, 4.05) to (0, 5);
             \draw [line width=1pt] (0 , 2.5) to (0, 3.45);
             \draw [line width=1pt] (1 , 2.5) to (1, 3.45);
             \node (e) at (0.5 , 2.5) [draw,minimum width=25pt,minimum height=10pt,fill=white] {$\scriptstyle{a_M}$};
             \draw [line width=1pt] (0.5 , 1.5) to (0.5, 2.15);
		\end{tikzpicture}
	\end{matrix}
 =
 \begin{matrix}
		\begin{tikzpicture}[scale=.75, out=up, in=down, line width=0.5pt]
			\node at (0, 6.3) {$\scriptstyle{A}$};
			\node at (1, 6.3) {$\scriptstyle{M}$};
		   \node at (0.5, 1.3) {$\scriptstyle{N}$};
     \draw [line width=1pt] (0 , 5.5) to (0, 6.1);
             \draw [line width=1pt] (1 , 5.5) to (1, 6.1);
     \node (m) at (0.5, 5.3) [draw,minimum width=25pt,minimum height=10pt,thick, fill=white] {$\scriptstyle{1\otimes f}$};
			\draw [line width=1pt] (0, 4) to (1, 5);
             \draw [white, line width=5pt] (0.7, 4.3) to (0.3, 4.7);
             \draw [blue][line width=1pt] (1, 4) to (0, 5);
               \draw [blue][line width=1pt] (0, 3.05) to (1, 4);
             \draw [white, line width=5pt] (0.7, 3.3) to (0.3, 3.7);
             \draw [line width=1pt] (1, 3.05) to (0, 4);
               \node (e) at (0.5 , 2.8) [draw,minimum width=25pt,minimum height=10pt,fill=white] {$\scriptstyle{a_M}$};
             \draw [line width=1pt] (0.5 , 1.5) to (0.5, 2.5);
		\end{tikzpicture}
	\end{matrix}
 \ee

 Here the first equality follows from the fact that $f$ is a morphism of $A$-modules (equation \eqref{eq:af=fagraph}), and the second follows from naturality of braiding. 
 
\end{proof}


    

Many useful results about $A\Mod$ and $A\lMod$ are proven in \cite{creutzig2017tensor} and \cite{CLR23}, which we will recall here. See \cite[Theorem 2.53, Theorem 2.55]{creutzig2017tensor}. 

\begin{Prop}\label{Prop:FSAMod}
    The following statements are true.

    \begin{enumerate}
        \item The category $A\Mod$ is a tensor category, such that $M\otimes_A N$ is the co-equalizer of the following map:
        \be
\btik
(A\otimes_\CU M)\otimes_\CU N\arrow[rrr, shift left, "a_M\otimes 1"] \arrow[rrr, shift right, swap, "(1\otimes a_N)\circ a^{-1}_{M,A,N}\circ  c_{A,M}"]& & & M\otimes_\CU N
\etik
        \ee

        \item The category $A\lMod$ is a braided tensor category, where the braiding $M\otimes_A N\to N\otimes_A M$ is given by the descent of $c_{M,N}: M\otimes_\CU N\to N\otimes_\CU M$ to the quotient. 

        \item The above braiding makes sense even when, say, $M$ is not local, as long as $N$ is local.  
    \end{enumerate}
\end{Prop}

The third point deserves a little more attention. We first recall the definition of Drinfeld center. See the book \cite[Section 7.13]{egno} for more details. 

\begin{Def}
Let $\CB$ be a tensor category. The \textbf{Drinfeld center} $\CZ(\CB)$ is the category whose objects are pairs $(M, c_{*,M})$ where $M\in\CB$ and $c_{*, M}$ is a natural isomorphism:
\be
c_{N, M}: N\otimes_\CB M\to M\otimes_\CB N
\ee
satisfying the hexagon axiom, which is easiest depicted using the following diagram:
\begin{center}
\begin{tikzpicture}
\pic[braid/strand 1/.style={red}, braid/strand 2/.style={yellow}, thick, name prefix=braid1] at (0,0) {braid={s_2^{-1}s_1^{-1}}};
\node[at=(braid1-1-s), anchor=south] {$X$};
\node[at=(braid1-2-s), anchor=south] {$Y$};
\node[at=(braid1-3-s), anchor=south] {$M$};
\node[at=(braid1-1-e), anchor=north] {$X$};
\node[at=(braid1-2-e), anchor=north] {$Y$};
\node[at=(braid1-3-e), anchor=north] {$M$};
\draw (3,-1.5) node[anchor=west, font=\Huge]{$=$};
\pic[braid/strand 1/.style={green}, thick, name prefix=braid2] at (4,0) {braid={1 s_1^{-1}}};
\node[at=(braid2-1-s), anchor=south] {$X\otimes Y$};
\node[at=(braid2-2-s), anchor=south] {$M$};
\node[at=(braid2-1-e), anchor=north] {$X\otimes Y$};
\node[at=(braid2-2-e), anchor=north] {$M$};
\end{tikzpicture}
\end{center}

Morphisms between $(M, c_{*, M})$ and $(N, c_{*, N})$ are morphisms $f: M\to N$ such that the following diagram commutes for all $P\in \CB$:
\be
\btik
P\otimes_\CB M\rar{c_{P,M}} \dar{1\otimes f}& M\otimes_\CB P\dar{f\otimes 1}\\
P\otimes_\CB N\rar{c_{P,N}} & N\otimes_\CB P
\etik
\ee

\end{Def}

The category $\CZ(\CB)$ is always a braided tensor category. 

\begin{Def}
Let $\CB$ be a tensor category and $\CC$ a braided tensor category. We say that $\CB$ is $\CC$-central if there is a faithful braided tensor functor $\CF: \overline{\CC}\to \CZ(\CB)$, where $\overline{\CC}$ is the category $\CC$ with opposite braiding (namely $c^{-1}$ instead of $c$). 

\end{Def}

\begin{Rem}
    The data of $\CF$ consist of a monoidal functor $\iota: \CC\to \CB$ and a half-braiding $c': M\otimes \iota N\to \iota N\otimes M$ for all $M\in \CB$ and $N\in \CC$, such that the restriction of $c'$ to $\iota N_1\otimes \iota N_2$ is $\iota c^{-1}$, where $c$ is the braiding on $\CC$. 
\end{Rem}

The third point of Proposition \ref{Prop:FSAMod} could be rephrased as follows: 

\begin{Prop}
    The tensor category $A\Mod$ is $A\lMod$-central. 
    
\end{Prop}

\subsubsection{Induction Functor and Properties}\label{subsubsec:Ind}

In this section, we review the relation between $\CU$ and $A\Mod$. The relation is through a pair of adjunctions. 
These functors are also covered in \cite[Section 2.7]{creutzig2017tensor} and in particular there it is shown that they respect the vertex algebra structure. 
\begin{Def}
    The forgetful functor $\mathrm{forget}_A: A\Mod\to \CU$ has a left adjoint, given by $\mathrm{Ind}_A$, which is defined on objects by:
    \be
\mathrm{Ind}_A(M):=A\FU M
    \ee
    such that the action morphism is given by the action of $A$ on itself. It is in fact a tensor functor (\cite[Theorem 1.6]{kirillov2002q}).
\end{Def}

A nice consequence of having such an adjoint pair is that every object $M\in A\Mod$ is the surjective image of $\mathrm{Ind}_A(\mathrm{forget}_A(M))$, where the morphism $\mathrm{Ind}_A(\mathrm{forget}_A(M))\to M$ is given by multiplication map:
\be
\mathrm{Ind}_A(\mathrm{forget}_A(M))\cong A\otimes_\CU \mathrm{forget}_A(M)\to M.
\ee
To simplify notations, we will mostly omit $\mathrm{forget}_A$. One way we will use this is to understand how much a module $M\in A\Mod$ fails to be local. Recall the module $M_c$; it is clear that $M$ is a local module if and only if the identity map $M_c\to M$ in $\CU$ is a morphism of $A\Mod$. Since $\mathrm{forget}_A(M_c)=\mathrm{forget}_A(M)$, we have the following lemma. 

\begin{Lem}
    There is an $A$ module homomorphism:
    \be
\btik
\Ind_A(M)\rar{a\oplus ac^{-2}} & M\oplus M_c
\etik
    \ee
    Moreover, $M$ is local if and only if $\mathrm{Im}(a\oplus ac^{-2})\subseteq \Delta (M)$, where $\Delta: M\to M\oplus M_c$ is the diagonal morphism (which is a morphism in $\CU$). 
\end{Lem}

\begin{proof}
    Since $\mathrm{forget}_A(M)\cong \mathrm{forget}_A(M_c)$, both $a$ and $ac^{-2}$ define $A$-module homomorphisms from $\Ind_A (M)$ to $M$ and $M_c$ respectively. $M$ is local precisely when $a=ac^{-2}$ and that is true when the image of $a\oplus ac^{-2}$ is contained in the diagonal. 
\end{proof}

\begin{Rem}
    Note that since $\mathbbm{1}$ embeds into $A$ and since we assumed that $\CU$ is rigid, as an object in $\CU$, $\Ind_A(M)$ contains $\mathbbm{1}\FU M$. Moreover, the following diagram commutes (as a diagram in $\CU$):
    \be
\btik 
\Ind_A(M)\rar{a\oplus ac^{-2}} & M\oplus M_c\\
\mathbbm{1}\FU M\uar\arrow[ur, swap, "\Delta"]
\etik
    \ee
    Therefore $\mathrm{Im}(a\oplus ac^{-2})$ as an object in $\CU$ always contains $\Delta (M)$. In fact one can show that $\mathrm{Im}(a\oplus ac^{-2})$ is the $A$-submodule of $M\oplus M_c$ generated by $\Delta (M)$.  
\end{Rem}

This leads to the following corollary:

\begin{Cor}
    The quotient $(M\oplus M_c)/(\mathrm{Im}(a\oplus ac^{-2}))$ is local. 
\end{Cor}

\begin{proof}
 Let $M_1 = (M\oplus 0) \cap (\mathrm{Im}(a\oplus ac^{-2}))$, then since $\mathrm{Im}(a\oplus ac^{-2})$ contains the diagonal, there is an isomorphism $M/M_1\cong (M \oplus M_c) / (\mathrm{Im}(a\oplus ac^{-2}))$. We have the following commutative diagram of short exact sequences:
    \be
\btik
0\rar & A\otimes_\CU M_1\rar\dar{a\oplus ac^{-2}} & A\otimes_\CU M\rar \dar{a\oplus ac^{-2}}&A\otimes_\CU M/M_1 \dar{a\oplus ac^{-2}}\rar & 0\\
0\rar & M_1\oplus M_{1, c}\rar \dar& M\oplus M_c\rar \dar & (M/M_1)\oplus (M/M_1)_c\rar \dar& 0\\
 & \overline M_1 & \overline M & \overline{M/M_1}
\etik
    \ee
    Here the overlined modules are the cokernels of the vertical arrows. The short exact sequences above induce a long exact sequence on the kernels and cokernels, and in particular we have an exact sequence:
    \be
    \btik
\overline{M}_1\rar & \overline{M}\rar & \overline{M/M_1}\rar & 0
    \etik
    \ee
  Now since $M_1\oplus M_{1, c}$ is contained in the image of $A\otimes_\CU M$ under $a\oplus ac^{-2}$ (since this image contains $M_1\oplus 0$ and $\Delta (M)$), the morphism $\overline M_1\to \overline M$ is zero, and therefore the morphism $\overline M\to \overline{M/M_1}$ is an isomorphism. Note that $\overline{M}=(M\oplus M_c)/\mathrm{Im}(a\oplus ac^{-2})=M/M_1$, this means that $M/M_1\cong \overline{M/M_1}$. Since the image of $a\oplus ac^{-2}$ contains at least the diagonal, and since objects in $\CU$ are of finite-length, this implies that the image of $A\otimes_\CU M/M_1$ under $a\oplus ac^{-2}$ is equal to the diagonal and $M/M_1$ is local. 
\end{proof}

The above shows that any object in $A\Mod$ has a canonical sub-object (the module $M_1$) such that $M/M_1$ is local. If one can show that $M_1$ is always strictly smaller than $M$, then one can construct a filtration for any object such that the associated graded is a local module. We will use this in Section \ref{sec:AModNMod}. 

Another consequence of the induction functor is the following. 

\begin{Prop}[\cite{davydov2013witt} Section 3.4]
    The tensor functor $\Ind_A: \CU\to A\Mod$ extends to a braided tensor functor:
    \be
\CS: \CU\longrightarrow \CZ\lp A\Mod\rp.  
    \ee
\end{Prop}

We will not go into detail, but there are natural transformations $b_{*, M}$ fitting into the following diagram:
\be
\btik
N\otimes_\CU (A\FU M) \rar{\eta} \dar & N\otimes_A \Ind_A (M)\arrow[ddddd, "b_{*, M}"]\\
(N\FU A)\FU M \dar{c^{-1}_{A,N}\otimes \mathrm{id}}& \\
(A\FU N)\FU M\dar & \\
A\FU (N\FU M)\dar{\mathrm{Id}\otimes c_{N,M}} & \\
A\FU (M\FU N)\dar & \\
(A\FU M)\FU N\rar{\eta} & \Ind_A(M)\otimes_A N
\etik
\ee
In the above diagram, the arrows without labels represent associativity, the morphism $\eta$ is the natural quotient map. The functor $\Ind_A$ together with the natural transformation $b_{*, M}$ defines the functor from $\CU$ to $\CZ(A\Mod)$. It turns out that the functor $\CS$ not only maps $\CU$ to the center of $A\Mod$, but in fact to the commutant of $A\lMod$ in $\CZ(A\Mod)$. We turn to this in the following section. 

\subsection{Relative Drinfeld Center from Commutative Algebra Object}

We start with the following definition. 

\begin{Def}
Let $\CB$ be $\CC$-central tensor category. The \textbf{$\CC$-relative Drinfeld center} of $\CB$ is defined as the centralizer of $\CC$ in $\CZ(\CB)$:
\be
\CZ_{\CC}(\CB):=\{X\in \CZ(\CB)\Big\vert c_{\CF(M), X}\circ c_{X, \CF(M)}=\mathrm{Id},~\forall M\in \CC\}. 
\ee

\end{Def}

The proof of Theorem 3.7 of \cite{CLR23} can be stated as follows. 

\begin{Prop}\label{Prop:SrelativeC}
    The functor $\CS:\CU \to \CZ (A\Mod)$ maps to the relative center
    \be
\CS: \CU\to \CZ_{A\lMod}(A\Mod).
    \ee
\end{Prop}

Moreover, under certain conditions, the above functor is fully-faithful. The most favorable conditions applicable to us are contained in a recent work \cite{creutzig2024commutative}, which we recall below. 

\begin{Asp}\label{AspCat}

${}$

\begin{enumerate} 

\item The categories $\CU$, $A\Mod$ and $A\lMod$ are all rigid. 

\item The identity element is simple in $\CU$ such that $\mathbbm{1}\to A$ is an embedding. 

\item For any subobject $S\subseteq A$ in $\CU$ not contained in $\mathbbm{1}$, there exists an object $Z\in \CU$ such that $c_{S, Z}\circ c_{Z, S}\ne 1$. 

\end{enumerate}

\end{Asp}

\begin{Rem}

The assumptions of \textit{op.cit.} are weaker than the above. Particularly, there one does not assume that $\CU$ is rigid, but only a braided $r$-category. However, one then needs a further assumption relating the Grothendieck-Verdier dual of an object to the rigid dual of its induction to $A\Mod$. This is trivially satisfied when $\CU$ is rigid. 

\end{Rem}

As stated, the following is true.

\begin{Thm}[\cite{creutzig2024commutative} Theorem 3.20]
    Under the assumptions in Assumption \ref{AspCat}, the functor $\CS$ in Proposition \ref{Prop:SrelativeC} is fully-faithful. 
\end{Thm}

It was also shown in \cite{CLR23} that when the category $\CU$ has enough projectives and is finite, then $\CS$ is an equivalence of abelian categories. Unfortunately the categories we consider in this paper (namely $KL_\rho$) do not have projective objects nor are finite. Nevertheless, this functor allows us to relate the original category $\CU$, which is the category of modules of $V$, to a relative Drinfeld center. The next step in seeking quantum group corresponding to $V$ is to look for quantum groups corresponding to $A\lMod$ and $A\Mod$. We will see in next section that in the special case at hand, there are natural (quasi-triangular) Hopf algebra candidates. Before we move on, we will show that the category $KL_\rho$ satisfies Assumption \ref{AspCat}, except for the rigidity of $A\Mod$, which will follow from the discussions in the next section. We specify $\CU=KL_\rho$, and denote by $A$ the algebra object defined by the free-field VOA $\CV_Z$. 

\begin{Prop}\label{Prop:aspKL}
    The category $KL_\rho$ is rigid and locally finite. Moreover, the algebra $A$ satisfies the third point of Assumption \ref{AspCat}. Finally, $A\lMod$ is rigid. 
\end{Prop}

\begin{proof}
  
  Proposition \ref{Prop:equivab} shows that $KL_\rho$ is locally finite. By \cite[Theorem 4.2.3]{creutzig2020tensor}, we see that $KL_1$ is rigid. We claim that this implies that $KL_{1^n}$ is rigid, where $KL_{1^n}$ is the Kazhdan-Lusztig category for $V(\fgl(1|1))^{\otimes n}$. Indeed, from Proposition  \ref{Prop:equivab} and \cite[Theorem 3.10.2]{etingof2011introduction}, one can show that simple objects of $KL_{1^n}$ are tensor products of $n$ simple objects in $KL_1$. Since each individual factor is rigid, we just need to show that fusion products in $KL_{1^n}$ distribute over tensor factors, namely
  \be
(\bigotimes_i M_i)\bigotimes_{KL_{1^n}}(\bigotimes_i N_i)=\bigotimes_i (M_i\bigotimes_{KL_1}N_i).
  \ee
  This can be proven, either by repeating the calculations of \cite[Section 3.2]{creutzig2020tensor} (for simple objects), or by repeating the arguments of \cite[Section 4.1]{mcrae2023deligne} for affine Lie superalgebras (in greater generality). As a consequence, simple objects in $KL_{1^n}$ are rigid. Now from \cite[Thm 4.4.1]{Creutzig:singlet2}, we see that $KL_{1^n}$ is rigid. To prove the rigidity of $KL_\rho$, we use \cite[Theorem 8.25]{BCDN}, which shows that $KL_\rho$ is a de-equivariantization of an appropriate subcategory of $KL_{1^n}$, and therefore is rigid.\footnote{The reference \cite{BCDN} was slightly incorrect in using the Deligne product $KL^{\boxtimes n}$, since this category is not the same as $KL_{1^n}$. However, the arguments were correct after using the correct category.} Clearly, $A\lMod$ is rigid. 
    
We are left to check the third point of Assumption \ref{AspCat}. By Corollary \ref{Cor:monsimple}, the monodromy:
\be
\sigma_{\lambda,\mu} V(\grho)\otimes_{KL_\rho} S\longrightarrow \sigma_{\lambda,\mu} V(\grho)\otimes_{KL_\rho} S
\ee
is given by $\mathrm{Id}\otimes \exp (2\pi i (\sum \lambda_aN_0^a+(\mu-\rho^\trans \rho \lambda)_aE_0^a))$. Therefore for $S$ to have trivial monodromy with $\sigma_{\lambda,\mu} V(\grho)$ for arbitrary $\lambda, \mu$, any simple composition factor of $S$ has to be $V(\grho)$ or $\Pi V(\grho)$, its parity shift. Since any extension between $V(\grho)$ with itself will have a non-semisimple action of $N_0^a$ or $E_0^a$, if such $S$ exists, it has to be a direct sum of $V(\grho)$ or $\Pi V(\grho)$. However, by the proof of \cite[Theorem 8.25]{BCDN}, the socle of $A$ is $V(\grho)$. This proves that $S$ can't have non-trivial monodromy with arbitrary $Z$ unless $S\subseteq V(\grho)$. The proof is now complete. 
    
\end{proof}

In the next section, we introduce the Hopf algebra candidates corresponding to $A\lMod$ and $A\Mod$, which lead to a quasi-triangular Hopf algebra candidate for $\CZ_{A\lMod}(A\Mod)$.

\section{Yetter-Drinfeld Modules and Quantum Groups}\label{sec:YDQG}

In this section, we present relevant technical details behind the second part of the proof strategy, relating $A\Mod$ to modules of a Nichols algebra.

\subsection{Modules of Nichols Algebra}\label{subsec:NicholsMod}

\subsubsection{Hopf algebra objects in braided tensor categories}

Let $\CC$ be a braided tensor category. A bialgebra $\mf N$ in $\CC$ is a unital algebra in $\CC$ with multiplication $m: \mf N\otimes_\CC \mf N\to \mf N$ and a counital coalgebra with comultiplication $\Delta:  \mf N\to  \mf N\otimes_\CC  \mf N$, satisfying the same diagrammatic constraints as an ordinary bialgebra. The only difference is the compatibility between multiplication and comultiplication
\be
\btik
 \mf N\otimes_\CC  \mf N\rar{m}\dar{\Delta\otimes \Delta} &  \mf N\dar{\Delta}\\
 \mf N\otimes_\CC  \mf N\otimes_\CC  \mf N\otimes_\CC  \mf N~~\rar{(m\otimes m)\circ c_{23}}& ~~  \mf N\otimes_\CC  \mf N
\etik
\ee
where a braiding $c_{23}$ is the used to commute the second and third factors of $\mf N$. A Hopf algebra object is a bialgebra object with an antipode $S:  \mf N\to  \mf N$ satisfying the same diagram as an ordinary Hopf algebra. See \cite{majid2023algebras} for details. 

Given a Hopf algebra object $ \mf N$, the category $ \mf N\Mod (\CC)$ has the structure of a monoidal category, where the tensor product of two objects is their tensor product in $\CC$, and the $ \mf N$-module structure is provided by $\Delta$. Note that this is different from the monoidal structure on the category of modules of a commutative algebra object. The monoidal category $\mf N\Mod$ (we also omit $\CC$ if no confusion may occur) is $\CC$-central, where the map $\CC\to \mf N\Mod$ is given by the counit $\epsilon: \mf N\to \mathbbm{1}$ and the central structure is provided by the inverse braiding of $\CC$. The category $\CZ_\CC (\mf N\Mod)$ is known to be the category of Yetter-Drinfeld modules, which we recall now. 

\begin{Def}[\cite{egno} Section 7.15]
An $\mf N-\mf N$-\textit{Yetter-Drinfeld} module $X$ in $\CC$ is an object in $\CC$ that is both a module of $\mf N$ and a co-module of $\mf N$, namely multiplication map $m_X: \mf N\otimes X\to X$ and co-multiplication map $\delta_X: X\to \mf N\otimes X$, such that the following compatibility conditions hold:

\begin{center}
                \begin{grform}
                        \begin{scope}[scale = 0.5]
                                \dMult{0.5}{1.5}{1}{-1.5}{\grau}
                                \dMult{0.5}{2.5}{1}{1.5}{\grau}
                                \dAction{2.5}{1.5}{0.5}{-1.5}{\grau}{black}
                                \dAction{2.5}{2.5}{0.5}{1.5}{\grau}{black}
                                \vLine{0.5}{1.5}{0.5}{2.5}{\grau}
                                \vLine{2.5}{1.5}{1.5}{2.5}{\grau}
                                \vLineO{1.5}{1.5}{2.5}{2.5}{\grau}
                                \vLine{3}{1.5}{3}{2.5}{black}
                        \end{scope}
                        \draw (0.5 , -0.3) node {$\mf N$};
                        \draw (0.5 , 2.3) node {$\mf N$};
                        \draw (1.5 , -0.3) node {$X$};
                        \draw (1.5 , 2.3) node {$X$};
                \end{grform}
                =
                \begin{grform}
                        \begin{scope}[scale = 0.5]
                                \vLine{2}{0}{1}{1}{black}
                                \dMultO{0}{1}{1.5}{-1}{\grau}
                                \dAction{0}{1}{1}{1}{\grau}{black}
                                \dAction{0}{3}{1}{-1}{\grau}{black}
                                \dMult{0}{3}{1.5}{1}{\grau}
                                \vLineO{2}{4}{1}{3}{black}
                                \vLine{1.5}{1}{1.5}{3}{\grau}
                        \end{scope}
                        \draw (0.45 , -0.3) node {$\mf N$};
                        \draw (0.45 , 2.3) node {$\mf N$};
                        \draw (1 , -0.3) node {$X$};
                        \draw (1 , 2.3) node {$X$};
                \end{grform}
\end{center}

\end{Def}

All $\mf N-\mf N$-Yetter-Drinfeld modules form a category which we denote by ${}_{\mf N}^{\mf N}\CY\CD\lp\CC\rp$, whose morphisms between objects are given by $\mf N$ module and co-module morphisms. It is a tensor category with tensor product given by the usual module and co-module tensor product. This category admits a braiding given by:
\be\label{eq:braidYD}
c_{(X,m_X, \delta_X),(Y, m_Y,\delta_Y)}=\lp m_Y\otimes \mathrm{id}_X\rp\circ \lp \mathrm{id}_{\mf N}\otimes c_{X,Y}\rp\circ\lp \delta_X\otimes \mathrm{id}_Y\rp
\ee
between objects $(X,m_X, \delta_X)$ and $(Y, m_Y,\delta_Y)$. The braiding is invertible if $\mf N$ has a bijective antipode. If $\CC$ is rigid, then ${}_{\mf N}^{\mf N}\CY\CD\lp\CC\rp$ is rigid, where the dual of an object is its dual in $\CC$ with action and co-action induced from the antipode of $\mf N$. This is summarized in the following statement, proven in \cite{bespalov1995crossed}:

\begin{Thm}
The category ${}_{\mf N}^{\mf N}\CY\CD\lp\CC\rp$ is a braided monoidal category. If $\CC$ is rigid, then so is ${}_{\mf N}^{\mf N}\CY\CD\lp\CC\rp$. 

\end{Thm}

The relation between Yetter-Drinfeld modules and the relative center is given by the following:

\begin{Prop}[\cite{laugwitz2020relative} Proposition 3.36]\label{Prop:centerYD}
Let $\CB:=\mf N\Mod\lp\CC\rp$, then there is an equivalence of braided tensor categories:
\be
\CZ_{\CC}\lp\CB\rp\simeq {}_{\mf N}^{\mf N}\CY\CD\lp\CC\rp.
\ee

\end{Prop}

\subsubsection{Nichols algebras}\label{subsubsec:Nicholsrho}

We now review the construction of Nichols algebras, which is a universal way of obtaining Hopf algebras in braided tensor categories. Let $\CC$ be a braided tensor category, and let $M$ be an object in $\CC$. Denote by $\CT(M)$ the tensor algebra generated by $M$, namely:
\be
\CT(M):=\bigoplus_{n\geq 0} M^{\otimes n},
\ee
with the understanding that $M^{\otimes 0}=\mathbbm{1}$. This has the structure of an $\N$-graded bi-algebra object in $\CC$ whose coproduct is given on $M$ by the diagonal map $M\to \mathbbm{1}\otimes M\oplus M\otimes \mathbbm{1}$, the unit given by the embedding of $\mathbbm{1}$ to $\CT(M)$ and the co-unit given by zero on $M^{\otimes n}$ for $n>0$. 

The \textbf{Nichols algebra} is defined as the universal $\N$-graded bi-algebra quotient of $\CT(M)$. 

\begin{Def}
    The Nichols algebra $\CB(M)$ is the minimal quotient of $\CT(M)$ by graded ideals contained in $\CT (M)_{\geq 2}$. Namely it is defined as a quotient bi-algebra of $\CT(M)$ such that for any other $\N$-graded surjective bi-algebra morphism $\CT(M)\to \CA$, whose composition with $M\to \CT(M)$ gives an injective map $M\to \CA$, there exists a bi-algebra morphism $\CA\to \CB(M)$ making the following diagram commutative:
\be
\btik
\CT(M)\rar \arrow[dr] & \CA\dar\\
 & \CB(M)
\etik
\ee
\end{Def}

Alternatively, one can define the Nichols algebra as the quotient of $\CT(M)$ by the kernel of the \textit{quantum symmetrizer} map, which is a lift of the total symmetrizer from the symmetric group to the braid group. Moreover, it turns out that it is automatically a Hopf algebra quotient of $\CT(M)$. It is difficult in general to understand the structure of Nichols algebra, but things become simpler when we assume that braiding on $M$ has very restricted form. We mention that in the case $\CC=\mathrm{Vect}_\Gamma$ where $\Gamma$ is a vector space with a non-degenerate bilinear form, finite-dimensional Nichols algebra is classified in \cite{heckenberger2009classification}. 

\begin{Exp}\label{Exp:p}
   Let $J$ be a simple current in $\CC$ with inverse $J^{-1}$. Assume that the braiding acting on  $J\otimes J=J^2$ is given by $(-1)$; then $\CB (J)\cong \mathrm{1}\oplus J$ where the product is trivial on $J^{\otimes 2}$.  More generally, if $M$ is an object in $\mathrm{Vect}_\Gamma$ whose braiding is $-1$, then $\CB (M)$ is the exterior algebra generated by $M$ (\cite[Proposition 7]{andruskiewitsch2017introduction}). This simple example is what is relevant to this paper.
   
   If the braiding on $J\otimes J$ is given by $e^{\pi i/p}$ for some positive integer $p$, then $\CB (J)=\oplus_{0\leq n<p}J^{\otimes n}$ such that $J^p=0$. This example is similar to the singlet algebra $V=M(p)$. 
    
\end{Exp}

\begin{Exp}\label{Exp: diagonal}[\cite{CLR23}, Example 5.2]
   Let $\fg$ be a semisimple Lie algebra, with Cartan matrix $\fh$ and Killing form $(-,-)$. The category $\mathrm{Vect}_\Gamma$ where $\Gamma=\fh^*$ and braiding $q^{(\lambda,\mu)}$ define a braided tensor category for $q\in \C^\times$. Let $\alpha_1,\ldots, \alpha_n$ be a choice of simple roots and $X=\bigoplus_i \C_{\alpha_i}$. Then $\CB (X)$ is the quantum Borel part of the small quantum gorup $u_q(\fg)$. This is called ``diagonal braiding". 
\end{Exp}

Recall Proposition \ref{Prop:centerYD} shows that for a Hopf algebra object $\mf N$ in $\CC$, the category $\CZ_\CC (\mf N\Mod)$ is equivalent to the category of Yetter-Drinfeld modules of $\mf N$. When $\CC$ is the category of modules of a quasi-triangular Hopf algebra $C$, and assume $\mf N$ is finite-dimensional as a module of $C$, then the category of Yetter-Drinfeld modules is the same as the category of modules of the relative Drinfeld double $U$ of $\mf N$ over $C$ \cite{majid1999double}. It turns out that $U_\rho^E$ can be realized as the relative Drinfeld double of a Nichols algebra object. 

To this end, let $C$ be a commutative and cocommutative Hopf algebra, whose weight lattice is $\Gamma$. Let $M=\oplus_{1\leq i\leq n} \C_{\gamma_i}$ for some $\gamma_i\in \Gamma$, and $\CB (M)$ the corresponding Nichols algebra. Assume that it is finite-dimensional as a module of $C$. In this case, the work of \cite{CLR23} gives the following explicit realization of the relative Drinfeld double $U$ of $\CB(M)$ over $C$. 

\begin{Lem}[\cite{CLR23} Lemma 6.8]\label{Lem:HopfReal}
Let $C$ be a commutative cocommutative Hopf algebra, whose weight lattice is an abelian group $\Gamma$, such that simple modules of $C$ are objects in $\mathrm{Vect}_\Gamma$. Let $\CB(M)$ be the Nichols algebra of an object $M=\bigoplus_{1\leq i\leq n}\C_{\gamma_i}$ for $\gamma_i\in \Gamma$, whose generators are denoted by $\{x_i\}$. Assume that there are group-like elements in $C$, denoted $g_i, \overline{g}_i$ such that $c_{\gamma_i, N}=g_i\vert_N$ and $c_{N, \gamma_i}=\overline{g}_i\vert_N$. 

Define an algebra $U$ generated by $C$ and $x_i, x_i^*$ with the following commutation relations:
\be
gx_i=\gamma_i(g^{(1)})x_ig^{(2)},\qquad gx_i^*=\gamma_i(Sg^{(1)})x_i^*g^{(2)},\qquad x_ix_j^*-\gamma_j(g_i)x_j^*x_i=\delta_{ij}(1-\overline{g}_ig_i).
\ee
Here $g^{(1)}$ and $g^{(2)}$ are such that $\Delta (g)=g^{(1)}\otimes g^{(2)}$, and $S$ is the antipode in $C$. Extending the coproduct on $x_i,x_i^*$ by:
\be
\Delta(x_i)=x_i\otimes 1+g_i\otimes x_i,\qquad \Delta (x_i^*)=\overline{g}_i\otimes x_i^*+x_i^*\otimes 1,
\ee
and antipode by:
\be
S(x_i)=-g_i^{-1}x_i,\qquad S(x_i^*)=-\overline{g}_i^{-1}x_i^*.
\ee
Then the above makes $U$ into a Hopf algebra, and there exists an $R$-matrix defined in some algebraic closure of $U$ making the category of $U$-modules a braided tensor category. Moreover, we have an equivalence of braided tensor categories:
\be
U\Mod\simeq \CZ_{\CC}(\CD),
\ee
 where $\CC=C\Mod$ and $\CD=\CB(M)\Mod (\CC)$. 

\end{Lem}

\begin{Rem}

If $C$ does not admit group-like elements $g_i$ but only Lie-algebra-like elements $N_i$ such that $c_{\gamma_i, -}=e^{2\pi i N_i}$, then one can always formally add the group-like elements $g_i$ and require that $g_i=\exp (2\pi i N_i)$ on finite-dimensional modules. The $R$-matrix can be read off from \eqref{eq:braidYD}. Let $\{v_k\}$ be a basis for $\CB(M)$ and $\{v_k^*\}$ be a dual basis. Each of $v_k$ has a weight under $C$ given by given by combinations of $\gamma_i$, and there are elements $g_k,\overline{g}_k$ such that $c_{v_k, M}=g_k\vert_M$ and $c_{M, v_k}=\overline{g}_k\vert_M$. Tracing through the definition, the $R$-matrix is given by the following expression:
\be
R=\sum_k R_C\lp R_C^{-1}\lp 1\otimes v_k\rp R_C\rp\lp v_k^*\otimes 1\rp=\sum_k R_C \lp \overline{g}_k^{-1}v_k^*\otimes v_k \rp
\ee
where $R_C$ is the $R$-matrix for the algebra $C$. Here we used the fact that $R_C^{-1}(1\otimes v_k)R_C=\overline{g}_k^{-1}\otimes v_k$. 

\end{Rem}

We now fix $\rho$, and let $C_\rho$ be the commutative and cocommutative Hopf algebra $\C[X^a, Y^a, (-1)^F]$, where $F$ is the parity operator. The $R$-matrix is given by $R=\exp (2\pi i (\sum_a X^a\otimes Y^a))(-1)^{F\otimes F}$.  Let $\C_i$ be the module of $C_\rho$ where $X^a$ acts as $-\rho^{ia}$ and $Y^a$ as zero. Let $M=\bigoplus_i \Pi \C_i$, where $\Pi$ is parity shift. Since braiding acting on $\Pi \C_i\otimes \Pi \C_j$ is equal to $-1$, the Nichols algebra $\CB (M)$ is simply an exterior algebra generated by $M$.  Similarly, the dual of the Nichols algebra $\CB (M)^*$, which is isomorphic to $\CB(M^*)$ (\cite[Section 3.4.2]{andruskiewitsch2017introduction}), is also an exterior algebra for the same reason. 

According to Lemma \ref{Lem:HopfReal}, we choose elements $g_i$ and $\overline{g}_i$ so that the braiding
\be
c_{\C_i, N}: \Pi \C_i\otimes N\to N\otimes \Pi \C_i, \qquad c_{N,\C_i}: N\otimes \Pi \C_i\to \Pi \C_i\otimes N
\ee
is given by $c_{\C_i,N}=\tau\circ g_i\vert_N$ and $c_{N,\C_i}=\tau\circ \overline{g}_i\vert_N$. Using the explicit $R$-matrix, the braiding $c_{\C_i, N}$ is given by $e^{-2\pi i \sum_a \rho^{ia}Y^a}(-1)^F$ and therefore we choose $g_i=e^{-2\pi i \sum_a \rho^{ia}Y^a}(-1)^F$, whereas the braiding $c_{N,\C_i}=(-1)^F$ and therefore $\overline{g}_i=(-1)^F$.

The algebra $U$ will be generated over $C_\rho$ by odd elements $x_\pm^i$, such that:
\be
[X^a,x_\pm^i]=\pm \rho^{ia}x_\pm^i, \qquad [Y^a, x_\pm^i]=0,\qquad \{x_+^i,x_-^i\}=1-e^{-2\pi i \sum \rho^{ia} Y_a}.
\ee
This is the algebra relation defining $U_\rho^E$, via the identification:
\be
X^a\mapsto N^a,\qquad Y^a\mapsto E^a, \qquad x_+^i\mapsto \psi_+^i,\qquad (-1)^Fg_i^{-1}x_-^i\mapsto \psi_-^i.
\ee
 The coproduct is given by:
\be
\Delta (x_-^i)=x_-^i\otimes 1+g_i\otimes x_-^i, \qquad \Delta (x_+^i)=x_+^i\otimes 1+(-1)^F\otimes x_+^i.
\ee
The antipode is $S(x_-^i)=-g^{-1}_ix_-^i$ and $S(x_+^i)=-(-1)^Fx_+^i$. The $R$-matrix is given by:
\be
\begin{aligned}
    R&=e^{2\pi i \sum_a X^a\otimes Y^a}(-1)^{F\otimes F}\prod_i (1+(-1)^Fx_+^i\otimes x_-^i)\\&=e^{2\pi i\sum_a N^a\otimes E^a}(-1)^{F\otimes F} \prod_i(1-\psi_+^i(-1)^F\otimes (-1)^Fg_i\psi_-^i)\\&=e^{2\pi i\sum_a N^a\otimes E^a}(-1)^{F\otimes F} \prod_i(1-\psi_+^i(-1)^F\otimes \prod_a K_{E^a}^{\rho^{ia}}\psi_-^i),
\end{aligned}
\ee
The quasi-triangular Hopf algebra structure coincides with the one defined in Theorem \ref{Thm:quasitriUq}, under the above identification. 

\begin{Rem}

    In this identification, one needs to take into the account the Koszul sign convention of Section \ref{sec:QG}. For example, the action of $1\otimes \psi$ in Section \ref{sec:QG} is given by $(-1)^F\otimes \psi$. If one takes into account this Koszul sign rule, then the identification of the Hopf structure is clear. The reason we include $(-1)^F$ in this section is that results of \cite[Lemma 6.8]{CLR23} are stated for the underlying abelian category. 
    
\end{Rem}

\begin{Cor}\label{Cor:rhocenter}
There is an equivalence of braided tensor categories
\be
U_\rho^E\Mod\simeq \CZ_{C_\rho\Mod}(\CB (M)\Mod). 
\ee

\end{Cor}

\subsection{From screening operators to Nichols algebras}

Recall from last section that we have found a functor $\CS: KL_\rho\to \CZ_{A\lMod}(A\Mod)$. Proposition \ref{Prop:AlModWMod} implies that the category $A\lMod$ is simply the category of modules of the VOA $W$ (which in the case we are interested in, is just the free-field algebra $\CV_Z$) that viewed as a module of $V$ lies in $\CU$ ($KL_\rho$ in our examples). The category $\CV_Z\Mod$ is easily seen to be equivalent to $C_\rho\Mod$, where the $R$-matrix is given by $\exp (\pi i (\sum_a X^a\otimes Y^a+Y^a\otimes X^a))(-1)^{F\otimes F}$. We can perform a twist by $\CF=\exp (\pi i \sum_a Y^a\otimes X^a)$, so that the new $R$-matrix becomes $\exp (2\pi i (\sum_a X^a\otimes Y^a))(-1)^{F\otimes F}$, matching the convention from Section \ref{subsubsec:Nicholsrho}. In conclusion, we have the following statement. 

\begin{Prop}
  There is an equivalence of braided tensor categories:
\be
A\lMod\simeq C_\rho\Mod.
\ee  
\end{Prop}

Comparing this with Corollary \ref{Cor:rhocenter}, we would like to establish a monoidal equivalence $A\Mod\simeq \CB (M)\Mod$ that is compatible with the central action of $A\lMod$. In fact, it is more than a technical speculation that $\CB (M)$ should show up in the study of $A\Mod$. In the case at hand, $V(\grho)\to \CV_Z$ is defined by the kernel of a set of screening operators $S_i=\oint \exp (Z_i-\sum \rho^{ia}Y^a)$. These satisfy relations:
\be
\{S_i, S_j\}=0.
\ee
This is precisely the relation of the Nichols algebra $\CB (M)$. 

Such a relation between Nichols algebras and screening operators was first observed in \cite{semikhatov2012nichols} for $V=M(p)$ the singlet VOA. In this case, the free field realization $V\to H_\phi$ is given by a Heisenberg VOA at level $1$. The screening operator defining $V$ satisfies Nichols algebra relation, which in this case is $S^p=0$ and $[\phi_0,S]=S$ (see example \ref{Exp:p}). The algebra generated by $\phi_0$ and $S$ is precisely the Borel subalgebra $B_q$ of the unrolled restricted quantum group $\overline{U}_q^H(\mathfrak{sl}(2))$ at $q=e^{\pi i/p}$. It is expected that the quantum group corresponding to $M(p)$ is this quantum group, and therefore the authors of \textit{loc.cit} conjectured that there is an equivalence of tensor categories:
    \be
B_q\Mod\simeq A\Mod
    \ee
compatible with the central action of $A\lMod\simeq C_q\Mod$ on both sides. Here $A$ is the algebra object defined by $H_\phi$ in the category of $M(p)$ modules, and $C_q$ is the Cartan subalgebra of $\overline{U}_q^H(\mathfrak{sl}(2))$. In general, let $W$ be a VOA containing Heisenberg fields $\alpha_i$, and let $S_i=\mathrm{Res}\lp e^{\alpha_i}\rp$, then \cite{lentner2021quantum} shows that $S_i$ satisfies Nichols algebra relation of $\oplus_i \C_{\alpha_i}$. Moreover, in \cite{flandoli2022algebras} the authors classified all finite-dimensional Nichols algebras coming from diagonal braiding (namely examples from Example \ref{Exp: diagonal}). Assume now that $W$ decomposes into integer eigenvalues under the zero-modes of $\alpha_i$, and define $M=\bigoplus_i \sigma_{\alpha_i}W$, where $\sigma_{\alpha_i}$ is the spectral flow generated by $\alpha_i$. Assume that $\sigma_{\alpha_i}W$ has diagonal braiding and defines a finite-dimensional Nichols algebra. It is expected that the following conjecture is true

\begin{Conj}\label{Conj:Nichols}
    Let $V=\cap_i\mathrm{Ker}(S_i)$. Let $A$ be the algebra object defined by $W$ in the category of $V$ modules, and let $\CB (M)$ be the Nichols algebra of $M$ in $A\lMod$. There is an equivalence of tensor categories:
    \be
\CB (M)\Mod (A\lMod)\simeq A\Mod
    \ee
compatible with the central structure of $A\lMod$ on both sides. 

\end{Conj}

The work \cite{semikhatov2013logarithmic} provides 
evidence for this conjecture by presenting an explicit logarithmic realization of $\CB (M)$ stemming from singlet \cite{Creutzig:singlet1, Creutzig:singlet2} and triplet algebras \cite{Tsuchiya:2012ru}. For singlets $M(p)$, if one restricts to the category of weighted modules, this conjecture is proven in \cite[Theorem 8.9]{CLR23}. In this case, $M=\CF_{-\frac{2}{p}\phi}$, and the screening operator satisfies that $S^p=0$. 

\begin{Thm}[\cite{CLR23} Theorem 8.9]\label{Thm:MpUq}
    Conjecture \ref{Conj:Nichols} is true for $M(p)\to H_\phi$ and the category of weighted modules of $M(p)$. 
\end{Thm}

Although this work only proved the conjecture for this specific example, we will review the strategy here since many of the steps are applicable to our situation.

\subsubsection{Proof strategy: from abelian equivalence to tensor equivalence}\label{subsubsec:ABtoTen}

The idea of \cite{CLR23} is that if one knows the abelian category of $A\Mod$ as well as $\CB(M)\Mod (A\lMod)$ well enough, then one can uniquely fix a tensor equivalence between them. 

Let $\CC$ be a braided tensor category, and $\mf N$ an algebra object in $\CC$. Let $\CD$ be a $\CC$-central tensor category. If one can understand $\mf N\Mod$ and $\CD$ as abelian categories, and obtain an equivalence of abelian categories $\CD\simeq \mf N\Mod$, then it is tempting to say that $\mf N$ is a Hopf algebra so that the above is an equivalence of tensor categories. This is not always the case, but it is true when there is a \textit{split tensor functor}. 

\begin{Def}
    Let $\CC$ be a braided tensor category and $\CD$ a $\CC$-central category. A split tensor functor is a functor $\epsilon: \CD\to \CC$ that is a tensor functor and the composition $\CC\to \CD\to \CC$ is identity functor.
    
\end{Def}

We have the following statement. 

\begin{Prop}[\cite{CLR23} Lemma 4.12]\label{Prop:Split}
    Let $(\CC, c)$ be a braided tensor category and $\CD$ a $\CC$-central category with $(\iota, c'): \overline{\CC}\to \CZ(\CD)$. Assume that there is a split tensor functor $\epsilon: \CD\to \CC$ such that $\epsilon(c')=c^{-1}$. 

    Suppose $\mf N$ is an algebra object in $\CC$. The category $\mf N\Mod$ is an abelian category with a forgetful functor $\epsilon': \mf N\Mod\to \CC$. There is also a right $\CC$-module category structure on $\mf N\Mod$ such that $M\otimes_\CC N$ is the $\mf N$ module via the action on the left on $M$.   Under this action, $\epsilon'$ is naturally a functor of right $\CC$-module categories. 

    Suppose there is an abelian equivalence $F: \mf N\Mod\to \CD$, together with a natural isomorphism of functors $\eta: \epsilon\circ F\to \epsilon'$, such that $F$ preserves the right $\CC$ module structure in the sense that there are natural isomorphisms:
    \be
   F^T_{M,N}: F(M\otimes_\CC N)\cong F(M)\otimes_\CD \iota (N)
    \ee
    for objects $M\in \mf N\Mod$ and $N\in \CC$. Assume also that $F$ induces trivial right module structure on $\CC$ in the sense that the following diagram commutes:
\be
\btik
\epsilon \lp F(M\otimes_\CC N)\rp \rar{\epsilon (F^T_{M,N})}\dar{\eta} & \epsilon  \lp F(M)\otimes_\CD \iota (N)\rp\dar{\epsilon_{F(M), \iota (N)}}\\
\epsilon'(M\otimes_\CC N)\dar{\epsilon'_{M,N}} & \epsilon (F(M))\otimes_\CC \epsilon (\iota (N))\dar{\eta\otimes \mathrm{Id}}\\
\epsilon'(M)\otimes_\CC N\rar{=} & \epsilon'(M)\otimes_\CC N
\etik
\ee
In this, $M$ is an object in $\mf N\Mod$ and $N\in \CC$, $\epsilon_{F(M),\iota(N)}$ is the natural tensor isomorphism of the tensor functor $\epsilon$, and $\epsilon'_{M, N}=Id$ is the natural isomorphism of the right $\CC$-module functor $\epsilon'$. 

Under these conditions, the algebra $\mf N$ has a bi-algebra structure in $\CC$, such that $F$ is an equivalence of tensor categories. If both $\CC$ and $\CD$ are rigid, then $\mf N$ has the structure of a Hopf algebra such that $F$ is an equivalence of rigid tensor categories. 

\end{Prop}

Namely, if one could produce a tensor functor $\CD\to \CC$ that behaves like a forgetful functor, such that its composition with the abelian equivalence $F:\mf N\Mod\to \CD$ is the forgetful functor, then $\mf N$ has a bi-algebra structure. Of course, since the coproduct is defined by the transportation of structure from $\CD$, it is not easy to really understand the coproduct explicitly. However, in the case when $\mf N$ is isomorphic to a Nichols algebra with diagonal braiding, then it is possible to uniquely fix the coproduct. To make it more precise, let us introduce the notions involved here. Let $\Gamma$ be an abelian group with a quadratic form $Q$, determining a braided tensor category $\CC=\mathrm{Vect}_\Gamma$. Choose $\overline{\alpha_i}\in \Gamma$ and let $M=\C_{\overline{\alpha}_1}\oplus \cdots \C_{\overline{\alpha}_n}$. Let $\N^n$ be the linear $\N$-span of $\overline{\alpha}_i$, whose generators we now denote by $\alpha_i$, so that $\overline{\alpha}_i$ are the images of $\alpha_i$ under the map $\N^n\to \Gamma$. Let $\CB(M)$ be the Nichols algebra of $M$. This algebra is $\N^n$-graded as well as $\Gamma$-graded.  We denote by $\mathrm{Supp}(\CB (M))\subseteq \N^n$ all the degrees appearing in $\CB (M)$. 

\begin{Def}

We say that $\CB(M)$ is sufficiently unrolled if $\overline{\alpha}+\overline{\beta}=\overline{\gamma}$ implies $\alpha+\beta=\gamma$ for any $\alpha, \beta, \gamma\in \mathrm{Supp}(\CB (M))$. 

\end{Def}

When we know that the algebra object $\CB(M)$ is sufficiently unrolled, then the grading fixes the Hopf algebra structure uniquely. 

\begin{Prop}[\cite{CLR23} Lemma 5.6]\label{Prop:sufunroll}
Let $\mf N$ be a bi-algebra object in $\mathrm{Vect}_\Gamma$ that is isomorphic to $\CB(M)$ as an object in $\mathrm{Vect}_\Gamma$. If $\CB(M)$ is sufficiently unrolled, then there is an isomorphism of bi-algebras $\CB(M)\to \mf N$. 

\end{Prop}

\begin{proof}

We can view $M$ as a sub-object of $\mf N$. The $\Gamma$ grading on the algebra and co-algebra structure on $\mf N$ lifts to $\N^n$ grading, thanks to the sufficiently unrolled condition. This fixes the coproduct on $\mf N$ on the generators $M$ to be given by:
\be
\Delta(m)=m\otimes 1+1\otimes m, \qquad m\in M,
\ee
as it is the only part in $\mf N\otimes \mf N$ of the same $\N^n$ degree. Now the embedding $M\to \mf N$ lifts to a map of algebras $\varphi: \CT(M)\to \mf N$, and by the above formula, it is automatically a map of $\N$-graded bi-algebras. Therefore, there is a surjective map from the image of $\varphi$ to $\CB(M)$, due to the universal property. Now because $\CB(M)$ is isomorphic to $\mf N$ as objects, by grading constraints, the image of $\varphi$ must be the entire $\mf N$ and therefore $\CB(M)$ and $\mf N$ are isomorphic as bi-algebras.

\end{proof}

\begin{Rem}
The above also implies that $\mf N$ will have an antipode, since $\CB(M)$ does. Therefore the above is an isomorphism of Hopf algebras. 

\end{Rem}

\begin{Rem}
    Although the above is stated for objects in $\mathrm{Vect}_\Gamma$, it applies to a general braided tensor category $\CC$, as long as $M$ is a direct sum of simple modules whose fusion product generates a category equivalent to $\mathrm{Vect}_\Gamma$. 
\end{Rem}

The significance of Proposition \ref{Prop:sufunroll} for us is as follows. According to Proposition \ref{Prop:Split}, if one can obtain an abelian equivalence $A\Mod\simeq \CB (M)\Mod (A\lMod)$, together with a split tensor functor satisfying some assumptions, one can give $\CB (M)$ an abstract bialgebra structure. This might not be convenient, since it is defined by transporting the monoidal structure. However, Proposition \ref{Prop:sufunroll} shows us that if $\CB (M)$ is sufficiently unrolled, then the grading fixes this bialgebra structure uniquely (and identifies it with the Nichols algebra structure). Consequently, the abelian equivalence we start with can be upgraded to a tensor equivalence. This was used in \cite{CLR23} to prove Theorem \ref{Thm:MpUq}. The proof of \textit{loc.cit} does not apply to our situation, since the category $A\lMod$ in their case is semisimple (as only weighted modules of $H_\phi$ is considered), and therefore it is easy to produce a split tensor functor. In Section \ref{sec:AModNMod}, we extend their proof in a very simple case of diagonal braiding $(-1)$, for which we construct this split tensor functor. 

The proof we will present in the next section applies to $KL_\rho$. However, for general $\rho$ the sufficiently-unrolled condition of Proposition \ref{Prop:sufunroll} is not satisfied. Indeed, in the example of $C_\rho$ and $M=\bigoplus_i \Pi\C_i$, the support of $\CB (M)$ is precisely the subset of $\N^n$ given by:
\be
\mathrm{Supp}(\CB (M))=\{(a_i)_{1\leq i\leq n}\big\vert a_i=0 \text{ or } 1\}.
\ee
The weight lattice $\Gamma$ of $C_\rho$ is $\Z^{r}_X\times \Z^r_Y\times \Z_2$, and the map from $\N^n\to \Gamma$ is given by:
\be
\btik
\N^n\rar{(-\rho^\trans, 0, \tau )} & \Z^r_X\oplus \Z^r_Y\oplus \Z_2, \qquad \tau (n_i)=\overline{\sum_i n_i}, ~(n_i)\in \N^n.   
\etik
\ee
One then sees that for general $\rho$, this won't map $\mathrm{Supp}(\CB (M))$ injectively into $\Gamma$. This is the reason why in the next section, we will use the above strategy only for $\rho=1$, since this satisfies sufficiently-unrolled condition. We will derive the main statement for general $\rho$ using ungauging. 

\section{Kazhdan-Lusztig correspondence for $\fgl (1|1)$}\label{sec:KL1proof}

With all the technical preparations in place, we are ready to prove Theorem \ref{KLrho} for $\rho=1$. in this case, the VOA $V(\grho)$ is simply $V(\fgl (1|1))$ and the quantum group $U_\rho^E$ is $U_1^E$.  The free field algebra $\CV_Z$ is $H_{X, Y}\otimes \FF$, which determines an algebra object $A$ in $KL_{1}$. The category $A\lMod$  consists of modules of the commutative algebra $C=\C[X,Y, (-1)^F]$. The screening operator $S^i$ is valued in the object $\Pi \CF_{-Y}\otimes \FF$. From previous discussions, the logic of the proof is given by the following steps:
\begin{enumerate}
    \item We have seen in Proposition \ref{Prop:aspKL} that $KL_{1}, A\Mod$ and $A\lMod$ satisfy the assumptions in Assumption \ref{AspCat} except the rigidity of $A\Mod$. 

    \item We have seen that $A\lMod\simeq C\Mod$. Let $\mf N$  be the Nichols algebra of $\Pi\C_{-Y}$. We will show that $A\Mod\simeq \mf N\Mod$ as tensor  categories. In fact we will prove a general statement in Theorem \ref{Thm:AmodNmod}, which applies directly to $V(\fgl (1|1)$ and $\CV_Z$. 

    \item This step will also imply that $A\Mod$ is rigid, and therefore we have a fully-faithful functor:
    \be
    \begin{split}
\CS: KL_{1}\longrightarrow\  \CZ_{A\lMod}(A\Mod)&\simeq  \CZ_{C\Mod}(\mf N\Mod)\simeq \\ &\simeq {}_{\mf N}^{\mf N}\CY\CD(C\Mod)\simeq U_1^E\Mod
    \end{split}
    \ee

    \item Since we already know that $KL_1$ is equivalent to $U_1^E\Mod$ as an abelian category, and that $U_1^E$ is finite-dimensional over $C$, by Corollary 3.13 of  \cite{CLR23}, we see that $\CS$ is an equivalence. 
    
\end{enumerate}

The main difficulty above is, of course, the equivalence $A\Mod\simeq \mf N\Mod$, which requires the construction of a split tensor functor. This is the main focus of the next section. After proving Theorem \ref{KLrho} for $\rho=1$, we comment on how the above steps generalize to the case $\rho=\mathrm{Id}_n$. 

\subsection{Realization of $A\Mod$ as Modules of a Nichols Algebra}\label{sec:AModNMod}

Let $V$ be a vertex operator algebra and $W$ a vertex operator algebra containing V. Assume the following assumptions.

\begin{enumerate}
    \item There is a locally-finite rigid braided ribbon category of $V$ modules $\CU$, such that $V$ is a simple object and $W$ defines an indecomposable algebra object in $V$ which we call $A$. 

    \item The category $A\lMod$ is a rigid braided tensor category such that $A$ is simple. 

    \item There exists a Heisenberg generator $\pd\varphi\in W$ of level $1$ generating a $\C^\times$ action on $W$, with which one can define spectral flow $\sigma_\varphi$ via Li's delta operator \cite{Lidelta}. Assume that all objects in $A\lMod$ can be decomposed into generalized eigenvectors of $\varphi_0$. 

    \item Assume that $V=\mathrm{Ker}(S)$ where $S=\mathrm{Res}_{z=0}\exp \lp\varphi(z)\rp$. Consequently, there is an exact sequence $\sigma_\varphi^n W\to \sigma_\varphi^{n+1}W\to \sigma_\varphi^{n+2}W$. Denote by $\sigma_\varphi^n V$ the corresponding kernel in $\sigma_\varphi^n W$. Assume that this is a set of simple currents with no fixed-points in $\CU$. 

    \item Assume the existence of another Heisenberg field $\pd \psi$ generating a $U(1)$ action on $W$ such that $\pd \psi\pd \varphi\sim \frac{1}{(z-w)^2}$. Moreover, assume that $\psi_0$ acts trivially on $W$.  Also assume that all objects in $A\lMod$ decompose into generalized eigenvalues of $\psi_0$. 

    \item Assume that there exists another object $B$ in $\CU$ that is a nontrivial extension of $V$ by $\sigma_\varphi V$. 
    
\end{enumerate}

With these assumptions, it is clear that as an object in $A\lMod$, the spectral flow $\sigma_\varphi W$ has a diagonal braiding with itself. Namely the braiding:
\be
c: \sigma_\varphi W\otimes_W \sigma_\varphi W\to  \sigma_\varphi W\otimes_W \sigma_\varphi W
\ee
is given by $(-1)$. This follows simply from the definition of Li's delta operator \cite{Lidelta}. Therefore, from the classification of Nichols algebras of diagonal type, we can easily see the following.

\begin{Cor}
    Let $\CB(\sigma_\varphi W)$ be the Nichols algebra of $\sigma_\varphi W$ in $A\lMod$, then it is isomorphic to $W\oplus \sigma_\varphi W$ where the multiplication is trivial on $\sigma_\varphi W$. We denote this algebra by $\mf N$. 
\end{Cor} 

In this slightly more general setup, we prove the following statement, which is the conjectural relation between screening operators and Nichols algebras. 

\begin{Thm}\label{Thm:AmodNmod}
    There is an equivalence of tensor categories:
    \be
A\Mod \lp \CU\rp\simeq \mf N\Mod \lp A\lMod\rp. 
    \ee
\end{Thm}

We will then show that $V(\fgl (1|1))\hookrightarrow W$ satisfies these assumptions, which allows us to conclude that $A\Mod\simeq \mf N\Mod (C\Mod)$. We prove Theorem \ref{Thm:AmodNmod} in the following steps:

\begin{enumerate}
    \item We construct a canonical filtration on $A\Mod\lp\CU\rp$ (will be denoted by $A\Mod$ in the following) whose associated graded is in $A\lMod$. We show that this upgrades to a faithful functor $\epsilon$ to $A\lMod$.

    \item We show that the functor $\epsilon$ upgrades to a fully-faithful functor $\CE$ to $\mf N\Mod$ by demonstrating that every object has a canonical action of $\mf N$. 

    \item We show that the functor $\CE$ is surjective  by constructing a preimage of $\mf N$. 

    \item We then conclude that $\CE$ is an abelian equivalence. Now by Proposition \ref{Prop:Split} and Lemma 5.6 of \cite{CLR23}, we conclude that this equivalence must be one of tensor categories. 
\end{enumerate}

\noindent\textbf{Step 1.} In this step, we need to construct a canonical filtration on objects of $A\Mod$. Let $M$ be an object in $A\Mod$, recall the twisted module $M_c$ from Section \ref{subsubsec:Ind}, namely the object whose underlying $V$ module is still $M$ but the action of $A$ becomes:
\be
\btik
A\FU M_c\rar{c^{-2}} & A\FU M_c\rar{a} & M_c
\etik
\ee
Recall we also considered the following morphism of $A$ modules in Section \ref{subsubsec:Ind}:
\be\label{eq:a+ac^{-2}}
\btik
A\otimes_\CU M\rar{a\oplus ac^{-2}} & M\oplus M_c
\etik
\ee
such that the cokernel of this map is always local. Our first step is to see that under the assumptions we have in this section, the cokernel is always non-zero. 

 \begin{Lem}
The image $\mathrm{Im}(a\oplus ac^{-2})$ is strictly smaller than $M\oplus M_c$.
 \end{Lem}

\begin{proof}
    Obviously the restriction of $a\oplus ac^{-2}$ to $V\otimes_\CU M\cong M$ is the diagonal morphism. Quotienting by $V\otimes_\CU M$ and its image, we have an induced morphism (in $\CU$) on the quotient:
    \be
    \btik
(A/V)\otimes_\CU M\rar{a\oplus ac^{-2}} & M\oplus M_c/(\Delta (M))\cong M
\etik
    \ee
    \textbf{Assumption 4} shows that $A/V\cong  \sigma_\varphi V$, and therefore $(A/V)\otimes_\CU M\cong  \sigma_\varphi(V)\otimes_\CU M$. We denote this by $\sigma_\varphi M$. \textbf{Assumption 4} also guarantees that $\sigma_\varphi M$ and $M$ are not isomorphic, and since they have equal number of composition factors (as $\sigma_\varphi V$ is a simple current), there can be no surjection from $\sigma_\varphi M$ to $M$. 

\end{proof}

\begin{Rem}
    The module $\sigma_\varphi M\cong \sigma_\varphi (V)\otimes_\CU M$ is an object in $\CU$, but it also has the structure of an $A$ module via the identification:
    \be
\sigma_\varphi (V)\otimes_\CU M\cong (\sigma_\varphi (V)\otimes_\CU A)\otimes_A M\cong \sigma_\varphi (A)\otimes_A M,
    \ee
    where the latter (as an object in $\CU$) is a quotient of $\sigma_\varphi (A)\otimes_\CU M$ by the image of the two $A$-action maps (see Proposition \ref{Prop:FSAMod}). 
\end{Rem}

By this lemma and \textbf{Assumption 1}, which implies that objects in $A\Mod$ have finite length, we see that all objects are equipped with a filtration such that the associated graded is local. In particular, all simple objects are local. This, together with \cite[Theorem 3.14]{creutzig2024commutative}, implies the following result.

\begin{Cor}\label{cor:exact}
The category $A\Mod$ is rigid. Consequently, tensor product is exact in $A\Mod$. 

\end{Cor}

One would like the split tensor functor to be the associated graded of such a filtration. This is not good enough, though, since taking associated graded is usually not an exact functor. We will use the action of $\psi_0$ to define generalized eigenspaces. To do so however, we need to understand the $A$ module $A\otimes M$. 

\begin{Lem}\label{Lem:Kera}
    Let $K=\mathrm{Ker}(a)$ where $a:A\otimes_\CU M\to M$. Then $K\cong \sigma_\varphi M$ as an object in $A\Mod$. 
\end{Lem}

\begin{proof}
    The object $A\otimes_\CU A$ carries two actions of $A$, one on the left and one on the right. For any object $M$ in $A\Mod$, we have that $A\otimes_\CU M=(A\otimes_\CU A)\otimes_{A_R} M$, using the right action, and the $A$ module structure uses the left action $A_L$. The multiplication morphism $m: A\otimes_\CU A\to A$ is a morphism of both left and right $A$ modules, and its kernel $K_m$ therefore is a module of $A_L\otimes_\CU A_R$. As a module of $A_R$, it is identified with $\sigma_\varphi V\otimes_\CU A\cong \sigma_\varphi A$ since we have the following short exact sequence of $A_R$ modules:
    \be
    \btik
0\rar & V\otimes_\CU A\rar & A\otimes_\CU A\rar & \sigma_\varphi(V)\otimes_\CU A\rar & 0
    \etik
    \ee
     We claim that the left and right action must agree. Indeed, this follows from the fact that as objects in $\CU$ there is an identification $A\otimes_\CU \sigma_\varphi A\cong  \sigma_\varphi A\oplus \sigma_\varphi^2 A$, and therefore the hom space:
    \be
\Hom_{\CU} (A\otimes_\CU \sigma_\varphi A, \sigma_\varphi A)=\Hom_\CU (\sigma_\varphi A, \sigma_\varphi A)
    \ee
is one-dimensional. There is only one morphism whose restriction to $V\otimes  \sigma_\varphi A$ is identity. Clearly, the $A$ module $K$ is isomorphic to $K_m\otimes_{A_R} M$, and since left and right action on $K_m$ agrees, we see that this kernel is $\sigma_\varphi M$ as a module of $A$. This completes the proof.

\end{proof}

\begin{Cor}\label{Cor:Kequalizer}

The embedding $\iota: K\to A\otimes M$ equalizes the left and right actions of $A$ on $A\otimes M$, in the sense that the following diagram commutes
\be
\btik
A\otimes K \rar{1\otimes\iota} \dar{1\otimes \iota} & A\otimes A\otimes M\dar{(1\otimes a)\circ (c_{A, A}\otimes 1)}\\ 
A\otimes A\otimes M\rar{m\otimes 1} & A\otimes M
\etik
\ee

\end{Cor}

\begin{proof}

We have seen that $K=K_m\otimes_A M$ where $K_m$ is the $A-A$-submodule of $A\otimes A$ defined as the kernel of the multiplication map. The result now follows from the proof of Lemma \ref{Lem:Kera}, particularly that $K_m$ equalizes the left and right action of $A$ on $A\otimes A$. 

\end{proof}

Looking at the morphism in equation \eqref{eq:a+ac^{-2}}, we see that the module $M_{1, c}$ is precisely the image of $K\cong \sigma_\varphi M$ under $ac^{-2}$. We would now like to introduce generalized eigenvalues of $\psi_0$. This is dangerous however, since the action of $\psi (z)$ on a logarithmic module $M$ is not a Laurent series. The solution, as in \cite{lentner2021quantum}, is to define $\psi_0:=\mathrm{Res}(\psi(z))$ using the lift of a unit circle to the entire Riemann sphere of logarithmic functions. Therefore, in principle, the image of $\psi_0$ is not $M$, but the algebraic closure of $M$. However, we have seen that every $M$ comes equipped with a filtration such that the associated graded are local, and therefore $\psi_0$ restricted to the local pieces is valued in $M$, and moreover has generalized eigenvalues according to \textbf{Assumption 5}. Consequently, $\psi_0$ has generalized eigenvalues in $M$ since it does in all the local subquotients. 

Let us now decompose $M$ into generalized eigenspaces of $\psi_0$:
\be
M=\bigoplus_{\alpha\in \C} M_{\alpha}.
\ee
If $M$ is a local module, then by \textbf{Assumption 5}, this would be a decomposition of $M$ into $A$ submodules, since a local $A$ module is identified with a module of the mode algebra of $W$, and the OPE guarantees that $\psi_0$ commutes with everything in the mode algebra. In general, we would like to show that this decomposition gives a filtration of $A$ modules. We start with the following statement.

\begin{Lem}\label{Lem:ExtAmod}
    Let $M$ and $N$ be two local $A$ modules, such that the generalized eigenvalue of $\psi_0$ is $\alpha$ and $\beta$ respectively. Let $P$ be an $A$ module fitting into the exact sequence:
    \be
\btik
M\rar & P\rar & N
\etik
    \ee
    Then either $\beta=\alpha-1$ or $P$ is local. 
\end{Lem}

\begin{proof}
    Consider the following diagram:
    \be
\begin{tikzcd}
 \sigma_\varphi P\rar &A\otimes_\CU P\rar{a\oplus ac^{-2}} & P\oplus P_{c}\\
\sigma_\varphi M\rar\uar & A\otimes_\CU M\rar{a}\uar & M\uar{\Delta}
\end{tikzcd}
    \ee
    If $P$ is not local, then this induces a non-trivial morphism $\sigma_\varphi P\to P_{c}$. Moreover by definition, this map is trivial on the submodule $\sigma_\varphi M$ since $M$ is local. This means that we must have a morphism from $ \sigma_\varphi N\to P_{c}$. Note that the generalized eigenvalue of $\sigma_\varphi N$ is $\beta+1$, and having a nontrivial map between them would imply that $\beta+1=\alpha$. 
    
\end{proof}

This lemma shows that in order to build a non-local module of $A$ from local ones, the eigenvalues of $\psi_0$ must differ by one from top to bottom. In fact, we have the following corollary.

\begin{Cor}\label{Cor:Extnonint}
    Let $M,N$ be two objects in $A\Mod$ such that the generalized eigenvalues of $\psi_0$ on $M$ and $N$ differ by a non-integer. Then $\Ext^1(M,N)=0$. 
\end{Cor}

\begin{proof}
    If $M$ and $N$ are local, then this is Lemma \ref{Lem:ExtAmod}. Let us now assume $N$ is local and $M$ fits into a short exact sequence:
    \be
\btik
0\rar & M_1\rar &M\rar & M/M_1\rar & 0
\etik
    \ee
    We apply induction on the length of $M$, and assume that we have proven the statement for $M_1$ and $M/M_1$. We now have a long exact sequence of extension groups:
    \be
\btik
\cdots \rar & \Ext^1(M/M_1, N)\rar & \Ext^1(M,N)\rar & \Ext^1(M_1, N)\rar & \cdots
\etik
    \ee
    Therefore, if the two $\Ext$ groups are trivial for $M/M_1$ and $M_1$ then it is trivial for $M$. We can now induce on the length of $N$ as well and the proof goes in the same way. 
\end{proof}

We now give a partial order on $\C$ by declaring that $\alpha<\alpha+1$. We define $M_{[\geq \alpha]}:=\bigoplus_{n\geq 0}M_{\alpha+n}$. By finite-length property, for each $M$ we can find a finite set of minimal eigenvalues $\alpha_i$ such that:
\be\label{eq:eigdecom}
M=\bigoplus_i\bigoplus_{n\geq 0} M_{\alpha_i+n}.
\ee
Moreover, this is a finite direct sum of vector spaces. If $M$ is local, then this gives a decomposition of $M$ into submodules. Using Lemma \ref{Lem:ExtAmod} and Corollary \ref{Cor:Extnonint}, we can show the following. 

\begin{Prop}
   Each $M_{[\geq \alpha]}$ is a submodule of $M$.
\end{Prop}

\begin{proof}
    Let us first show that each $M_{[\geq \alpha_i]}$ is a $W$ submodule. We induce on the length of $M$. It is true when $M$ is local. Assume now that $M$ fits into an exact sequence: 
    \be
\btik
0\rar & M_1\rar & M\rar & M/M_1\rar & 0
\etik
    \ee
such that $M/M_1$ is local, and that we have proven the statement for $M_1$. This obviously leads to an exact sequence of vector spaces:
\be
\btik
0\rar & (M_1)_{[\geq \alpha_i]}\rar & M_{[\geq\alpha_i]}\rar & (M/M_1)_{[\geq \alpha_i]}\rar & 0
\etik
\ee
To show that $M_{[\geq \alpha_i]}$ is a submodule, we only need to show that there is no non-trivial extension between $(M/M_1)_{[\geq \alpha_i]}$ and $(M_1)_{[\geq \alpha_j]}$ for $j\ne i$. This is true thanks to Corollary \ref{Cor:Extnonint}. 

We can now assume the generalized eigenvalue of $\psi_0$ is contained in $[\geq \alpha]$. We show $M_{[\geq\alpha+n]}$ is a submodule for all $n$. Again we use induction on the length of $M$. It is true if $M$ is local. Suppose now that $M_1$ is the minimal submodule such that $M/M_1$ is local, and assume that we have proven the statement for $M_1$, and assume that $M/M_1$ has a decomposition:
\be
M/M_1=\bigoplus_{k\geq 0} (M/M_1)_{\alpha+k}.
\ee
By choosing a pre-image of $(M/M_1)_{\alpha+k}$ we may assume that $M/M_1$ has a single generalized eigenvalue $\alpha+k$. For $n>k$, it is clear that $M_{[\geq \alpha+n]}=(M_1)_{[\geq \alpha+n]}$, and these are submodules of $M$. When $n<k$, note that $M$ defines an extension between $M_1$ and $M/M_1$, and also that there is a short exact sequence:
\be
\btik
0\rar & (M_1)_{[\geq \alpha+k]} \rar & M_1\rar & M_1/(M_1)_{[\geq \alpha+k]}\rar & 0
\etik
\ee
By Corollary \ref{Cor:Extnonint},  $\mathrm{Ext}^1(M/M_1, M_1/(M_1)_{[\geq \alpha+k]})=0$ since the module $M_1/(M_1)_{[\geq \alpha+k]}$ has generalized eigenvalues smaller than $\alpha+k$. From the long exact sequence of extension groups:
\be
\btik
\cdots\rar & \Ext^1(M/M_1, (M_1)_{[\geq \alpha+k]})\rar &\Ext^1(M/M_1, M_1)  \\
 \rar&  \Ext^1(M/M_1, M_1/(M_1)_{[\geq \alpha+k]})=0\rar & \cdots
\etik
\ee
we see that the class of extension $[M]\in \Ext^1(M/M_1, M_1)$ must have a pre-image in $\Ext^1(M/M_1, (M_1)_{[\geq \alpha+k]})$. Therefore the extension $M$ is determined by an extension between $M/M_1$ and $(M_1)_{[\geq \alpha+k]}$, which we denote by $N$. In other words, $M$ must be given by the following quotient diagram:
\be
\btik
0\rar & (M_1)_{\geq \alpha+k}\rar & M_1\oplus N\rar & M\rar & 0
\etik
\ee
Here the first map is the diagonal embedding of $(M_1)_{[\geq \alpha+k]}$ into $M_1$ and $N$. For each $n\leq k$, we have $N_{[\geq \alpha+n]}=N$ and therefore is a submodule of $N$. Consequently, $M_{[\geq \alpha+n]}$ is a quotient of $N_{[\geq\alpha+n]}\oplus (M_1)_{[\geq\alpha+n]}$ and therefore is a submodule of $M$. This completes the proof. 
 
\end{proof}

The above proof shows that any $M\in A\Mod$ comes equipped with a filtration $M=M_0\supseteq M_1\supseteq M_2\supseteq \ldots \supseteq M_n=0$ such that:
\be\label{eq:filtern}
M_n= \bigoplus_i M_{[\geq \alpha_i+n]}.
\ee
Moreover, it is clear from the proof that the associated graded are all local modules. We now use this to construct the tensor functor $\epsilon$.

\begin{Prop}
    There is a faithful and exact tensor functor:
    \be
\epsilon: A\Mod\to A\lMod.
    \ee
\end{Prop}

\begin{proof}
    We define $\epsilon(M)$ to be the associated graded according to the filtration in equation \eqref{eq:filtern}. It is a faithful and exact functor because any $A$ module homomorphism must preserve generalized eigenvalues of $\psi_0$. We now show that it is a tensor functor. Let $M$ and $N$ be objects in $A\Mod$, assume that $M$ has generalized eigenvalues $[\geq\alpha]$ under $\psi_0$ and $N$ has generalized eigenvalues $[\geq \beta]$ under $\psi_0$. The filtration on $M$ and $N$ gives rise to a filtration on $M\otimes_A N$, given by:
    \be
F_k (M\otimes_A N)=\sum_{i+j=k}F_i(M)\otimes_A F_j(N)
    \ee
thanks to exactness of tensor product in $A\Mod$ (Corollary \ref{cor:exact}). Therefore, there is a surjective morphism:
\be
\bigoplus_{i+j=k}\mathrm{Gr}_i (M)\otimes_A \mathrm{Gr}_j (N)\to \mathrm{Gr}_k (M\otimes_A N).
\ee
Since the left-hand side has generalized eigenvalue $\alpha+\beta+k$, and therefore the right-hand side must also have this generalized eigenvalue. Since there is a unique filtration on $M\otimes_A N$ such that each associated graded piece has different generalized eigenvalue contained in $[\geq (\alpha+\beta)]$, the above must be this filtration. Taking associated graded, we have constructed a functorial map:
\be
\epsilon(M)\otimes_A \epsilon(N)\to \epsilon (M\otimes_A N).
\ee
We now show that this is an isomorphism. It is clearly an isomorphism when $N$ is local, and for general $N$, we have the following commutative diagram of exact sequences:
\be
\btik
0\rar & \epsilon (M)\otimes_A \epsilon(N_1)\rar \dar& \epsilon(M)\otimes_A \epsilon (N)\rar \dar& \epsilon (M)\otimes_A \epsilon(N/N_1)\dar\rar & 0 \\
0\rar & \epsilon (M\otimes_A N_1)\rar & \epsilon (M\otimes_A N)\rar & \epsilon (M\otimes_A (N/N_1))\rar & 0
\etik
\ee
By short five lemma, if the first and third vertical arrows are isomorphisms, then so is the second one. 

Finally, to show that this is a tensor functor, we must show that the ``monoidal structural axiom" (see \cite[Definition 2.4.1]{egno}) for $\epsilon$ and the isomorphism $\epsilon (M\otimes_A N)\cong \epsilon(M)\otimes_A \epsilon(N)$. This will follow from the commutativity of the following diagram:
\be
\btik
\lp F_i(M)\otimes_A F_j(N)\rp\otimes_A F_k(P)\rar{a_{M,N,P}}\dar & F_i(M)\otimes_A (F_j(N)\otimes F_k(P))\dar\\
F_{i+j}(M\otimes_A N)\otimes_A F_k(P)\dar & F_i(M)\otimes_A F_{j+k}(N\otimes_A P)\dar\\
F_{i+j+k}((M\otimes_A N)\otimes_A P)\rar{a_{M,N,P}} & F_{i+j+k}(M\otimes_A (N\otimes_A P))
\etik
\ee
whose commutativity follows from the naturality of $a_{M,N,P}$ and the definition of the filtration and the vertical maps. This completes the proof. 

\end{proof}

\begin{Rem}
    By construction, $\epsilon$ is a split tensor functor of the embedding $\iota:A\lMod\to A\Mod$, since if a module is local then equation \eqref{eq:eigdecom} gives a decomposition of $M$ into objects in $A\lMod$. 
\end{Rem}

\noindent\textbf{Step 2.} Now we upgrade $\epsilon$ to a functor $\CE$ into $\mf{N}\Mod(A\lMod)$. For each $M$, we denote by $\overline{M}_\alpha$ the associated graded piece of $\epsilon(M)$. 
As a consequence of Lemma \ref{Lem:Kera} and Lemma \ref{Lem:ExtAmod}, we see that the morphism of $A$ modules:
\be
\btik
A\otimes_\CU M_{[\geq \alpha]}\rar{a\oplus ac^{-2}} & M_{[\geq \alpha]}\oplus M_{[\geq \alpha], c}
\etik
\ee
induces a morphism of $A$ modules $ \sigma_\varphi M_{[\geq \alpha]}\to M_{[\geq\alpha], c}$. By degree consideration, and taking associated graded of $\epsilon$, this must induce a map:
\be
\sigma_\varphi \overline{M}_{\alpha}\to \overline{M}_{\alpha+1}.
\ee
Namely, there is a canonical morphism in $A\lMod$:
\be
\btik
\sigma_\varphi (A)\otimes_A \overline{M}_{\alpha}\rar{a_{[\alpha]}^{log}} & \overline{M}_{\alpha+1}
\etik 
\ee
The sum $a^{\log}:=\sum_{[\alpha]} a_{[\alpha]}^{log}$ gives us a morphism:
\be
\btik
\sigma_\varphi (A)\otimes_A\epsilon (M)\rar{a^{\log}} & \epsilon (M)
\etik
\ee
which by definition is the same as the structure of a module of $\CT ( \sigma_\varphi A)$, the tensor algebra of $\sigma_\varphi A$. 

\begin{Lem}
    The assignment $M\mapsto \epsilon (M)$ together with the action of $\CT( \sigma_\varphi A)$ is functorial. Namely this gives a functor:
    \be
\wt{\CE}: A\Mod \longrightarrow \CT(\sigma_\varphi A)\Mod \lp A\lMod\rp.
\ee
\end{Lem}

\begin{proof}
We need to show that if there is a morphism $f: M\to N$ of $A$ modules, then the following diagram commutes:
\be
\btik
\sigma_\varphi (A)\otimes_A \overline{M}_{[\alpha]}\rar\dar{1\otimes f_{[\alpha]}} & \overline{M}_{[\alpha+1]}\dar{f_{[\alpha+1]}}\\
\sigma_{\varphi} (A)\otimes_A \overline{N}_{[\alpha]}\rar & \overline{N}_{[\alpha+1]}
\etik
\ee
Here $f_{[\alpha]}$ is the restriction of $f$ on $M_{[\alpha]}$. This will then follow from the commutativity of the diagram (which is commutative because $f$ is a morphism of $A$ modules):
\be
\btik
A\otimes_\CU M_{[\geq \alpha]}\rar{a\oplus ac^{-2}}\dar{1\otimes f} & M_{[\geq \alpha]}\oplus M_{[\geq \alpha], c}\dar{f}\\
A\otimes_\CU N_{[\geq \alpha]}\rar{a\oplus ac^{-2}} & N_{[\geq \alpha]}\oplus N_{[\geq\alpha], c}
\etik
\ee
Indeed, taking the kernel of $a$ in the horizontal morphisms, we find that the following diagram commutes:
\be
\btik
 \sigma_\varphi M_{[\geq \alpha]}\rar{ac^{-2}}\dar{1\otimes f} & M_{[\geq \alpha+1], c}\dar{f}\\
 \sigma_\varphi N_{[\geq \alpha]}\rar{ac^{-2}} & N_{[\geq \alpha+1], c}
\etik
\ee
Taking associated graded we get the desired result. 

\end{proof}

We now show that the image lies in the subcategory of $\mf{N}$ modules. 

\begin{Lem}
    The morphism:
    \be
    \btik
\sigma_\varphi^2(\epsilon(M))\rar{1\otimes a^{log}} &  \sigma_\varphi(\epsilon(M))\rar{a^{log}} & \epsilon(M)
\etik
    \ee
    is zero. Therefore, the functor $\wt{\CE}$ induces a functor:
    \be
\CE: A\Mod\longrightarrow \mf N\Mod \lp A\lMod\rp. 
    \ee
\end{Lem}

\begin{proof}

We identify $\sigma_\varphi M$ with the kernel of $a$, such that the map $\sigma_\varphi M\to M_{c}$ is given by $ac^{-2}$ composed with the embedding $ \sigma_\varphi M\to A\otimes_\CU M$. Consider the following diagram:
\be
\btik
A\otimes_\CU (A\otimes_\CU M)\dar{a_{A,A,M}}\rar{1\otimes ac^{-2}} & A\otimes_\CU M\arrow[dd, "ac^{-2}"]\\
(A\otimes_\CU A)\otimes_\CU M \dar{m\otimes 1}& \\
A\otimes_\CU M\rar{ac^{-2}} & M
\etik
\ee
The commutativity of this diagram follows from the associativity of the action of $A$ on $M_c$. As a morphism in $\CU$, the morphism $\sigma_\varphi^2M\to M$ can be identified with the following morphism:
\be
\btik
\sigma_\varphi^2M\rar & A\otimes_\CU (A\otimes_\CU M)\rar{1\otimes ac^{-2}} & A\otimes M\rar{ac^{-2}} & M
\etik
\ee
Here we view $\sigma_\varphi M=\sigma_\varphi A\otimes_A M$ as a submodule of $A\otimes_\CU M$ (the kernel of $a$), and the module $\sigma_\varphi^2 M=\sigma_\varphi A\otimes_A \sigma_\varphi M$ as a submodule of $A\otimes_\CU (A\otimes_\CU M)$. The commutativity of the above diagram shows that the above composition is equal to the following composition:
\be
\btik
(\sigma_\varphi A\otimes_A \sigma_\varphi A)\otimes_A M\rar & (A\otimes_\CU A)\otimes_\CU M\rar{m\otimes 1} & A\otimes_\CU M\rar{ac^{-2}} & M
\etik
\ee
However, it is clear that $\sigma_\varphi A\otimes_A \sigma_\varphi A$ is in the kernel of $m$, and therefore the composition is zero. This completes the proof. 

\end{proof}

We have thus constructed the functor $\CE$. By construction, it is faithful; we claim that it is full.

\begin{Lem}
    The functor $\CE$ is full. 
\end{Lem}

\begin{proof}
If $N$ is local, then a morphism in $\Hom_{A\Mod} (M, N)$ factors through the quotient $M/M_1$, and therefore is simply an element in $\Hom_{A\lMod}(M/M_1, N)$. Since $N$ is trivial as a module of $\mf N$,  a morphism in $\Hom_{\mf N\Mod}(\CE(M), N)$ factors through the image of $\sigma_\varphi \CE (M)$. Using the fact that $M_1=\Im (ac^{-2})$ and the definition of the $\mf N$ action on $\CE (M)$, the quotient $\CE (M)/\sigma_\varphi \CE (M)=M/M_1$, and therefore $\Hom_{\mf N\Mod}(\CE(M), N)=\Hom_{A\lMod} (M/M_1, N)=\Hom_{A\Mod} (M, N)$. 

Now choose arbitrary $ M, N\in A\Mod$. Since $\CU$ is a locally-finite abelian subcategory, by  \cite[Section 2]{deligne2007categories}, the full abelian subcategory generated by $M\oplus N$ admits enough projective objects. Let $P$ be a projective cover of $M$. Then $\Ind_A (P)$ maps onto $M$ in $A\Mod$. We first show that $\Hom (\CE(A\otimes_\CU P), \CE (N))=\Hom (A\otimes_\CU P, N)$. This is true if $N$ is local. We now use induction. Suppose $N$ fits into the short exact sequence:
\be
\btik
0\rar & N_1\rar & N \rar &N/N_1\rar & 0
\etik
\ee
then we have the following commutative exact sequence:\\[3mm]
\be\nonumber
\btik[transform canvas={scale=0.85}]
&  \Hom (A\otimes_\CU P, N_1)\rar\dar & \Hom (A\otimes_\CU P, N)\rar\dar  & \Hom (A\otimes_\CU P, N/N_1)\dar\rar & 0 \\
0\rar &  \Hom (\CE(A\otimes_\CU P), \CE(N_1))\rar & \Hom (\CE(A\otimes_\CU P), \CE(N))\rar & \Hom (\CE(A\otimes_\CU P), \CE(N/N_1)) &
\etik
\ee
\\[-2mm]

Projectivity of $P$ and the induction-restriction adjunction ensures that the second horizontal map on the first row is surjective.  This, together with snake lemma, ensures that if the first and third vertical maps are isomorphisms, then so is the second vertical map. Applying the result from \cite[Section 2]{deligne2007categories} again, we can choose a projective cover $P_1$ of the kernel of $A\otimes_\CU P\to M$ in the subcategory generated by $A\otimes_\CU P\oplus N$, which gives us an exact sequence $A\otimes_\CU P_1\to A\otimes P\to M\to 0$. This gives rise to the short exact sequence\\[3mm]
\be
\btik[transform canvas={scale=0.95}]
0\rar & \Hom (M, N)\rar{f}\dar{a} & \Hom (A\otimes_\CU P_1, N)\rar{g}\dar{b} & \Hom (A\otimes_\CU P_2, N)\dar{c}\\
0\rar & \Hom (\CE(M), \CE(N))\rar{f'} & \Hom (\CE(A\otimes_\CU P_1), \CE(N))\rar{g'} & \Hom (\CE(A\otimes_\CU P_2), \CE(N))
\etik
\ee
\\[-2mm]

Note that the second and third vertical arrows ($b$ and $c$)  are isomorphisms. We show that this implies that the first one is surjective. Let $\upsilon$ be an element in $\Hom (\CE(M), \CE(N))$, then $f'\upsilon=b(\wt\upsilon)$ for some $\wt\upsilon\in \Hom (A\otimes_\CU P_1, N)$. Since $g'f'\upsilon=0$, we see that $cg\wt \upsilon=0$, and because $c$ is an isomorphism, $g\wt\upsilon=0$. By exactness of the first row, we see that $\wt \upsilon=f\upsilon'$ for some $\upsilon'\in \Hom (M,N)$. Therefore, $f'\upsilon=bf\upsilon'=f'a\upsilon'$. Since $f'$ is injective, $a\upsilon'=\upsilon$ and we are done.

\end{proof}

\noindent\textbf{Step 3.} Now we show that $\CE$ is essentially surjective. To do so, we first need to show that $\CE$ is a functor of $A\lMod$ module categories. 

\begin{Lem}
    Let $M$ be an object of $A\Mod$ and $N$ an object of $A\lMod$. The following diagram commutes
    \be\label{eq:ac2=ac2}
\btik
A\otimes_\CU (M\otimes_\CU N)\rar{a_Mc^{-2}_{A, M}\otimes 1} \dar{1\otimes \eta_{M, N}}& M\otimes_\CU N\dar{\eta_{M, N}}\\
A\otimes_\CU (M\otimes_A N)\rar{ac^{-2}} & M\otimes_A N
\etik
\ee 
\end{Lem}

\begin{proof}
     This follows from the following sequence of identities
    \be
\begin{aligned}
	\begin{matrix}
		\begin{tikzpicture}[scale=.75, out=up, in=down, line width=0.5pt]
			\node at (0, 6.3) {$\scriptstyle{A}$};
			\node at (1, 6.3) {$\scriptstyle{M}$};
			\node at (2, 6.3) {$\scriptstyle{N}$};
			\node at (1, -0.3) {$\scriptstyle{M \otimes_AN}$};
			\node (m) at (1.5, 3.75) [draw,minimum width=25pt,minimum height=10pt,thick, fill=white] {$\scriptstyle{\eta}$};
			\node (e) at (1 , 1.25) [draw,minimum width=25pt,minimum height=10pt,fill=white] {$\scriptstyle{ac^{-2}}$};
             \draw [teal][line width=1pt] (1 ,4.05) to (1, 6);
             \draw [violet][line width=1pt] (2 ,4.05) to (2, 6);
             \draw [red][line width=1pt] (0.75 , 1.55) to (0, 6);
             \draw [blue][line width=1pt] (1.25 , 1.55) to (1.5, 3.45);
             \draw [blue][line width=1pt] (1 , 0) to (1, 0.95);
		\end{tikzpicture}
	\end{matrix}
	&=	
	\begin{matrix}
		\begin{tikzpicture}[scale=.75, out=up, in=down, line width=0.5pt]
			\node at (0, 6.3) {$\scriptstyle{A}$};
			\node at (1, 6.3) {$\scriptstyle{M}$};
			\node at (2, 6.3) {$\scriptstyle{N}$};
			\node at (1, -0.3) {$\scriptstyle{M \otimes_AN}$};
			\node (m) at (1.5, 3.75) [draw,minimum width=25pt,minimum height=10pt,thick, fill=white] {$\scriptstyle{\eta}$};
            \node (e) at (1 , 1.25) [draw,minimum width=25pt,minimum height=10pt,fill=white] {$\scriptstyle{a}$};
			\draw [red][line width=1pt] (0,3.5) to (0, 6);
            \draw [teal][line width=1pt] (1 ,4.05) to (1, 6);
             \draw [violet][line width=1pt] (2 ,4.05) to (2, 6);
          \draw [blue][line width=1pt] (0.5 , 2.5) to (1.5, 3.5);
            \draw [white, line width=5pt] (1.2, 2.7) to (0.8, 3.1);
             \draw [red][line width=1pt] (1.5, 2.5) to (0, 3.5);
             \draw [red][line width=1pt] (0.5 , 1.5) to (1.5, 2.5);
             \draw [white, line width=5pt] (1.2, 1.8) to (0.8, 2.2);
             \draw [blue][line width=1pt] (1.5, 1.5) to (0.5, 2.5);
             \draw [blue][line width=1pt] (1 , 0) to (1, 1);
		\end{tikzpicture}
	\end{matrix}
	&=	
	\begin{matrix}
		\begin{tikzpicture}[scale=.75, out=up, in=down, line width=0.5pt]
			\node at (0, 6.3) {$\scriptstyle{A}$};
			\node at (1, 6.3) {$\scriptstyle{M}$};
			\node at (2, 6.3) {$\scriptstyle{N}$};
			\node at (1.75, -0.3) {$\scriptstyle{M \otimes_AN}$};
			\node (m) at (0.5, 2.75) [draw,minimum width=25pt,minimum height=10pt,thick, fill=white] {$\scriptstyle{a}$};
            \node (e) at (1.75 , 1.25) [draw,minimum width=25pt,minimum height=10pt,fill=white] {$\scriptstyle{\eta}$};
			\draw [teal][line width=1pt] (1 , 3.8) to (1, 6);
             \draw [violet][line width=1pt] (2 ,4.05) to (2, 6);
            \draw [white, line width=5pt] (1.2, 5.1) to (0.8, 5.5);
            \draw [white, line width=5pt] (2.2, 4.9) to (1.8, 5.3);
            \draw [red][line width=1pt] (0.1,3) to (2.5, 4.5);
            \draw [red][line width=1pt] (2.5,4.5) to (0, 6);
           
            \draw [white, line width=5pt] (1.2, 3.5) to (0.8, 3.9);
            \draw [white, line width=5pt] (2.2, 3.7) to (1.8, 4.1);
           \draw [teal][line width=1pt] (1 , 3) to (1, 3.9);
            \draw [violet][line width=1pt] (2 ,1.5) to (2, 4.05);

             \draw [white, line width=5pt] (1.2, 1.8) to (0.8, 2.2);
             \draw [teal][line width=1pt] (1.5, 1.5) to (0.5, 2.5);
             \draw [blue][line width=1pt] (1.75 , 0) to (1.75, 1);
		\end{tikzpicture}
	\end{matrix}
    &=	
	\begin{matrix}
		\begin{tikzpicture}[scale=.75, out=up, in=down, line width=0.5pt]
			\node at (0, 6.3) {$\scriptstyle{A}$};
			\node at (1, 6.3) {$\scriptstyle{M}$};
			\node at (2, 6.3) {$\scriptstyle{N}$};
			\node at (1.75, -0.3) {$\scriptstyle{M \otimes_AN}$};
			\node (m) at (1.75, 2) [draw,minimum width=25pt,minimum height=10pt,thick, fill=white] {$\scriptstyle{a}$};
            \node (e) at (1.5 , 0.75) [draw,minimum width=25pt,minimum height=10pt,fill=white] {$\scriptstyle{\eta}$};
			\draw [teal][line width=1pt] (1 , 3.8) to (1, 6);
             \draw [violet][line width=1pt] (2 ,4.05) to (2, 6);
            \draw [white, line width=5pt] (1.2, 5.1) to (0.8, 5.5);
            \draw [white, line width=5pt] (2.2, 4.9) to (1.8, 5.3);
            \draw [red][line width=1pt] (0.1,3) to (2.5, 4.5);
            \draw [red][line width=1pt] (2.5,4.5) to (0, 6);
            \draw [red][line width=1pt] (1.5, 2.25)  to (0.1,3) ;
          \draw [white, line width=5pt] (1.2, 3.5) to (0.8, 3.9);
            \draw [white, line width=5pt] (2.2, 3.7) to (1.8, 4.1);
             \draw [white, line width=5pt] (1.1, 2.5) to (0.7, 2.9);
             \draw [teal][line width=1pt] (1 , 1) to (1, 3.9);
            \draw [violet][line width=1pt] (2 ,2.25) to (2, 4.05);
            \draw [violet][line width=1pt] (2 ,1) to (2, 1.75);
             \draw [blue][line width=1pt] (1.5 , 0) to (1.5, 0.5);
		\end{tikzpicture}
	\end{matrix}
    &=	
	\begin{matrix}
		\begin{tikzpicture}[scale=.75, out=up, in=down, line width=0.5pt]
			\node at (0, 6.3) {$\scriptstyle{A}$};
			\node at (1, 6.3) {$\scriptstyle{M}$};
			\node at (2, 6.3) {$\scriptstyle{N}$};
			\node at (1.75, -0.3) {$\scriptstyle{M \otimes_AN}$};
			\node (m) at (1.75, 2.75) [draw,minimum width=25pt,minimum height=10pt,thick, fill=white] {$\scriptstyle{a}$};
            \node (e) at (1.5 , 1.25) [draw,minimum width=25pt,minimum height=10pt,fill=white] {$\scriptstyle{\eta}$};
			\draw [teal][line width=1pt] (1 , 3.8) to (1, 6);
             \draw [violet][line width=1pt] (2 ,4.05) to (2, 6);
            \draw [white, line width=5pt] (1.2, 5.1) to (0.8, 5.5);
            \draw [white, line width=5pt] (2.2, 4.9) to (1.8, 5.3);
            \draw [red][line width=1pt] (1.2,3) to (2.5, 4.5);
            \draw [red][line width=1pt] (2.5,4.5) to (0, 6);
            \draw [white, line width=5pt] (2.1, 3.7) to (1.7, 4.1);
           \draw [teal][line width=1pt] (1 , 1.5) to (1, 3.9);
            \draw [violet][line width=1pt] (2 , 3) to (2, 4.05);
             \draw [violet][line width=1pt] (2 ,1.5) to (2, 2.5);
              \draw [blue][line width=1pt] (1.5 , 0) to (1.5, 1);
		\end{tikzpicture}
	\end{matrix} \\
	&=	
	\begin{matrix}
		\begin{tikzpicture}[scale=.75, out=up, in=down, line width=0.5pt]
			\node at (0, 6.3) {$\scriptstyle{A}$};
			\node at (1, 6.3) {$\scriptstyle{M}$};
			\node at (2, 6.3) {$\scriptstyle{N}$};
			\node at (1.75, -0.3) {$\scriptstyle{M \otimes_AN}$};
			\node (m) at (1.75, 2.75) [draw,minimum width=25pt,minimum height=10pt,thick, fill=white] {$\scriptstyle{a}$};
            \node (e) at (1.5 , 1.25) [draw,minimum width=25pt,minimum height=10pt,fill=white] {$\scriptstyle{\eta}$};
			\draw [teal][line width=1pt] (1 , 1.5) to (1, 6);
             \draw [violet][line width=1pt] (2 ,4.05) to (2, 6);
            \draw [white, line width=5pt] (1.2, 4) to (0.8, 4.4);
            \draw [red][line width=1pt] (1.5,3) to (0, 6);
            \draw [violet][line width=1pt] (2 , 3) to (2, 4.05);
             \draw [violet][line width=1pt] (2 ,1.5) to (2, 2.5);
              \draw [blue][line width=1pt] (1.5 , 0) to (1.5, 1);
		\end{tikzpicture}
	\end{matrix}
    &=	
	\begin{matrix}
		\begin{tikzpicture}[scale=.75, out=up, in=down, line width=0.5pt]
			\node at (0, 6.3) {$\scriptstyle{A}$};
			\node at (1, 6.3) {$\scriptstyle{M}$};
			\node at (2, 6.3) {$\scriptstyle{N}$};
			\node at (1.75, -0.3) {$\scriptstyle{M \otimes_AN}$};
			\node (m) at (1.75, 2.75) [draw,minimum width=25pt,minimum height=10pt,thick, fill=white] {$\scriptstyle{a}$};
            \node (e) at (1.5 , 1.25) [draw,minimum width=25pt,minimum height=10pt,fill=white] {$\scriptstyle{\eta}$};
			\draw [teal][line width=1pt] (1 , 5) to (1, 6);
             \draw [violet][line width=1pt] (2 ,4.05) to (2, 6);
             \draw [white, line width=5pt] (1.2, 5.2) to (0.8, 5.6);
             \draw [red][line width=1pt] (1.5,5) to (0, 6);
             \draw [red][line width=1pt] (0,4) to (1.5,5) ;
            \draw [red][line width=1pt] (1.5,3) to (0, 4);
          \draw [white, line width=5pt] (1.3, 4.3) to (0.9, 4.7);
            \draw [white, line width=5pt] (1.2, 3.2) to (0.8, 3.6);
            \draw [teal][line width=1pt] (1 , 1.5) to (1, 5);
            \draw [violet][line width=1pt] (2 , 3) to (2, 4.05);
             \draw [violet][line width=1pt] (2 ,1.5) to (2, 2.5);
              \draw [blue][line width=1pt] (1.5 , 0) to (1.5, 1);
		\end{tikzpicture}
	\end{matrix}
    &=	
	\begin{matrix}
		\begin{tikzpicture}[scale=.75, out=up, in=down, line width=0.5pt]
			\node at (0, 6.3) {$\scriptstyle{A}$};
			\node at (1, 6.3) {$\scriptstyle{M}$};
			\node at (2, 6.3) {$\scriptstyle{N}$};
			\node at (1.75, -0.3) {$\scriptstyle{M \otimes_AN}$};
			\node (m) at (0.5, 2.75) [draw,minimum width=25pt,minimum height=10pt,thick, fill=white] {$\scriptstyle{a}$};
            \node (e) at (1.5 , 1.25) [draw,minimum width=25pt,minimum height=10pt,fill=white] {$\scriptstyle{\eta}$};
			\draw [teal][line width=1pt] (1 , 5) to (1, 6);
             \draw [violet][line width=1pt] (2 , 1.5) to (2, 6);
             \draw [white, line width=5pt] (1.2, 5.2) to (0.8, 5.6);
             \draw [red][line width=1pt] (1.5,5) to (0, 6);
             \draw [red][line width=1pt] (0,4) to (1.5,5) ;
            \draw [red][line width=1pt] (0,3) to (0, 4);
          \draw [white, line width=5pt] (1.3, 4.3) to (0.9, 4.7);
               \draw [teal][line width=1pt] (1 , 3) to (1, 5);
            \draw [teal][line width=1pt] (1 , 1.5) to (0.5, 2.5);
              \draw [blue][line width=1pt] (1.5 , 0) to (1.5, 1);
		\end{tikzpicture}
	\end{matrix}
    .
\end{aligned}
\ee
Most steps in this sequence of equalities follow from the naturality of braiding. The fifth equality follows from the locality of $N$ and the last one follows from the definition of $\eta$. 
\end{proof}

\begin{Prop}\label{Lem:CEAmod}

There are functorial isomorphisms
\be
\CE^T_{M, N}: \CE(M\otimes_A \iota (N))\cong \CE (M)\otimes_A N, \qquad M\in A\Mod, ~N\in A\lMod,
\ee
as $\mf N$-modules,  where $\iota: A\lMod\to A\Mod$ is the canonical embedding. This makes $\CE$ into a functor of $A\lMod$ module categories. 

\end{Prop}

\begin{proof}

By the construction of $\epsilon$, if $N$ is local, then there is a natural isomorphism
\be
\epsilon (M\otimes_A \iota (N))\cong \epsilon (M)\otimes_A N.
\ee
We only need to show that this isomorphism preserves the action of $\mf N$. Let $\eta_{M, N}: M\otimes N\to M\otimes_A N$ be the natural projection, then there is a commutative diagram
\be
\btik
A\otimes_\CU M\otimes N\rar{a_M\otimes 1} \dar{1\otimes \eta_{M, N}}& M\otimes N\dar{\eta_{M, N}}\\
A\otimes_\CU M\otimes_A N\rar{a_{M\otimes_A N}} & M\otimes_A N
\etik
\ee
which defines the $A$-action on $M\otimes_A N$. Let $K_M$ be the kernel of $a_M: A\otimes M\to M$ and $K_{M\otimes_A N}$ the kernel of $a_{M\otimes_A N}$, then the above induces a natural map $K_M\otimes N\to K_{M\otimes_A N}$. This map must be a surjection, since we have the following commutative diagram of split short exact sequences in $\CU$
\be
\btik
K_{M}\otimes_\CU N\dar{\eta_K}\rar & \arrow[r, shift left, "a_{M}\otimes 1"]A\otimes_\CU (M\otimes_\CU N)\dar{1\otimes \eta}& M\otimes_\CU N\arrow[l, shift left, " \mathrm{Id}\otimes 1"] \dar{\eta}\\
K_{M\otimes_A N}\rar &\arrow[r, shift left, "a_{M\otimes_A N}"]A\otimes_\CU M\otimes_A N& M\otimes_A N\arrow[l, shift left, "\mathrm{Id}\otimes 1"]
\etik
\ee
Here $\mathrm{Id}$ is the identity map $\mathbbm{1}\to A$. Moreover, Corollary \ref{Cor:Kequalizer} states that the embedding $K_M\to A\otimes M$ equalizes the two actions of $A$ on $A\otimes M$. This,  combined with the fact that $\eta: M\otimes N\to M\otimes_A N$ coequalizes the left and right actions of $A$, implies that the induced map $\eta_K: K_M\otimes N\to K_{M\otimes_A N}$ coequalizes the following two compositions
\be
\btik
A\otimes_\CU K_M\otimes_\CU N\arrow[rr, shift left, "a_{K_M}\otimes 1"] \arrow[rr, shift right, swap, "(1\otimes a_N)\circ c_{A, K_M}"]&  & K_M\otimes_\CU N\rar{\eta_K} & K_{M\otimes_A N}. 
\etik
\ee
Therefore, it induces a surjection from $K_M\otimes_A N$ to $K_{M\otimes_A N}$. Now this must be an isomorphism since by Lemma \ref{Lem:Kera}, both objects are identified with $\sigma_\varphi (V)\otimes_{\CU}(M\otimes_A N)$, and $\CU$ is of finite length. 

Let us now consider the action of $ac^{-2}: K_{M\otimes_A N}\to (M\otimes_A N)_{c}$. Equation \eqref{eq:ac2=ac2} and the definition of $\eta_K$ imply that the following diagram commutes
\be\label{eq:acacK}
\btik
K_M\otimes_\CU N \rar{ac^{-2}\otimes 1}\dar{\eta_K}&  M_c\otimes_\CU N\dar{\eta}\\
K_{M\otimes_A N}\rar{ac^{-2}} & (M\otimes_A N)_c
\etik
\ee
Although $\eta$ does not identify $M_c\otimes_A N$ with $(M\otimes_A N)_c$, it is nevertheless a map of $A\Mod$, where we use the left $A$ module structure on $M_c\otimes_\CU N$. Let us now take a filtration of $M$ from equation \eqref{eq:filtern}, which induces the filtration on $M\otimes_A N$. We might as well assume that $M=M_{[\geq \alpha]}$ for some $\alpha$. Now $\eta$ clearly makes the following diagrams commute for all $n\geq 0$: 
\be
\btik
M_{[\geq \alpha+n+1], c}\otimes_\CU N\rar \dar{\eta} & M_{[\geq \alpha+n], c}\otimes_\CU N\dar{\eta}\rar & \overline{M}_{\alpha+n}\otimes_\CU N\dar{\eta}\\
(M_{[\geq \alpha+n+1]}\otimes_A N)_c\rar & (M_{[\geq \alpha+n]}\otimes_A N)_c \rar & \overline{M}_{\alpha+n} \otimes_A N 
\etik
\ee
Note that the $\eta$ in the last column does coequalize the two actions of $A$ on $\overline{M}_{\alpha+n}\otimes_\CU N$, thanks to the fact that $\overline{M}_{\alpha+n}$ is local.  Now if we take associated graded in equation \eqref{eq:acacK}, we find that the following diagram commutes
\be
\btik
K_{M_{[\geq \alpha+n]}}\otimes_\CU N\rar{ac^{-2}\otimes 1} \dar{\eta_K}& \overline{M}_{\alpha+n+1}\otimes_\CU N\dar{\eta}\\
K_{M_{[\geq \alpha+n]}\otimes_A N}\rar{ac^{-2}} & \overline{M}_{\alpha+n+1} \otimes_A N 
\etik
\ee
This implies that in the associated graded, the map
\be
a^{log}_{M\otimes_A N}: \sigma_\varphi \lp \epsilon (M\otimes_A \iota (N))\rp\to \epsilon (M\otimes_A \iota (N))
\ee
is given by $a^{log}_M\otimes 1$, as desired. This completes the proof.

\end{proof}

Since every object in $\mf N\Mod(A\lMod)$ is a quotient of an object of the form $\mf N\otimes_A M$ for some local module $M$, we are done if we can show that $\mf N$ is in the image of $\CE$. 

\begin{Lem}
    There is an object $N$ in $A\Mod$ such that $\CE(N)\cong \mf N $. 
\end{Lem}

\begin{proof}
   First, let us comment that $\mf N$ is an extension between $A$ and $\sigma_\varphi A$. It is very clear that there is a unique such extension that has a nontrivial action of $\mf N$. Therefore, we only need to find an object $N$ such that $\CE (N)$ is a nontrivial extension between the two. Let $B$ be the extension of $V$ by $\sigma_\varphi V$ in the \textbf{Assumption 6}, and let $N=\Ind_A (B)$. By construction, $N$ fits in the exact sequence:
\be
\btik
\sigma_\varphi A\rar & N\rar & A
\etik
\ee
We claim that this is a non-split extension of $A$-modules. Otherwise, as an object in $\CU$, $N\cong A\oplus \sigma_\varphi A$. On the other hand, there is an embedding $B\to N$, but by indecomposability of $A$,  there is no embedding from $B$ to $\sigma_\varphi A\oplus A$. Since $N$ is indecomposable as a module of $A$, and since $A$ and $\sigma_\varphi A$ have different generalized eigenvalues of $\psi_0$, this module must be non-local. Consequently, the underlying object in $A\lMod$ of $\CE(N)$ is $A\oplus \sigma_\varphi A$, with a nontrivial action of $\mf N$. By uniqueness mentioned above, we have that $\CE(N)=\mf N$.

\end{proof}

\noindent\textbf{Step 4.} Let us pause here and summarize what the previous steps have shown. We constructed an equivalence of abelian categories:
\be
\CE: A\Mod (\CU)\longrightarrow \mf N\Mod (A\lMod)
\ee
such that the composition of $\CE$ with the forgetful functor $\epsilon': \mf N\Mod \to A\lMod$ is the tensor functor $\epsilon$. Moreover, $\epsilon (c')=c^{-1}$ since $c'$ is functorial, $\epsilon$ is defined as the associated graded of a filtration and $c'=c^{-1}$ on local modules. Finally, the equivalence $\CE$ is one of module categories of $A\lMod$. To apply Proposition \ref{Prop:Split}, we only need to show that $\CE$ induces trivial right module structure on $A\lMod$. 

\begin{Lem}
The functor $\CE$ induces a trivial right module structure on $A\lMod$ in the sense that the following diagram commutes:
\be
\btik
\epsilon' \lp \CE(M\otimes_A \iota(N))\rp \rar{\epsilon' (\CE^T_{M,N})}\dar & \epsilon'  \lp \CE(M)\otimes_A N\rp\dar{\epsilon'_{\CE(M), N}}\\
\epsilon(M\otimes_A\iota( N))\dar{\epsilon_{M,N}} & \epsilon' (\CE(M))\otimes_A N\dar{\eta\otimes \mathrm{Id}}\\
\epsilon(M)\otimes_A N\rar{=} & \epsilon(M)\otimes_A N
\etik
\ee
\end{Lem}

\begin{proof}
    This is a little tautological as soon as one understands the functors involved. The morphism $\epsilon'(\CE (M\otimes_A \iota (N)))\to \epsilon(M)\otimes_A N$ is given by taking associated graded of $M\otimes_A \iota (N)$ and forgetting the $\mf N$ action. The commutativity simply arises because the filtration on $M\otimes_A \iota (N)$ comes from a filtration on $M$ (since the filtration on $\iota (N)$ is really given by a direct sum decomposition of $N$, due to its locality). 
\end{proof}

Therefore, the equivalence $\CE$ induces a bi-algebra structure on $\mf{N}$ such that $\CE$ is an equivalence of tensor categories. Now, since the full subcategory of $A\lMod$ generated by $\sigma_\varphi^\pm A$ is equivalent to $\mathrm{Vect}_\Z$, in which $\mf N\cong \CB (\sigma_\varphi A)$ is sufficiently unrolled, by Lemma 5.6 of \cite{CLR23} (as recalled in Section \ref{subsubsec:ABtoTen}), the bi-algebra structure on $\mf N$ is unique and is given by the structure of $\mf N$ being a Nichols algebra. We thus obtain Theorem \ref{Thm:AmodNmod}.

\subsection{The equivalence for $V(\fgl (1|1))$}

We show that $V(\fgl (1|1))$ and the free-field algebra $\CV_Z^1$ satisfy the assumptions used in the proof of Theorem \ref{Thm:AmodNmod}. In this application, the element $\pd\varphi=\norm{bc}-\pd Y$ and $\pd \psi=\pd X$, both are elements in $\CV_Z^1$.  The spectral flowed module $\sigma_\varphi A$ is precisely $\Pi\C_{-1,0}$ considered in the main example of the previous section. The object $B$ is given by $\sigma_{\varphi}w^*(\CV_Z^1)$, where $w$ is the automorphism of $V(\fg(1|1))$ given by
\be
w(N(z))=-N(z), ~w(E(z))=-E(z), ~w(\psi^+)=\psi^-, ~w(\psi^-)=-\psi^+.
\ee
This $B$ is an extension of $V(\fgl (1|1))$ by $\sigma_\varphi (V(\fgl (1|1)))$ since $w^*\sigma_\varphi (V)=\sigma_{-\varphi}(V)$. This example clearly satisfies all the assumptions listed at the beginning of Section \ref{sec:AModNMod}. 

Theorem \ref{Thm:AmodNmod}, combined with the Schauenburg’s functor, provides us with a fully faithful braided tensor functor
\be
S: KL_1\longrightarrow U_1^E\Mod.
\ee
We claim now that this is an equivalence. Note that by Proposition \ref{Prop:equivab}, we already know that $KL_1$ is equivalent to $U_1^E\Mod$ as an abelian category. By Corollary 3.13 of \cite{CLR23}, the functor $S$ is realized as pull-back along an algebra homomorphism $f: U_1^E\to U_1^E$. Since $S$ clearly is compatible with the functor of restriction to $A\lMod$ in the sense that the following diagram commutes:
\be
\btik
KL_1\rar{S}\dar{\epsilon\circ\Ind_A} & U_1^E\Mod\dar\\
A\lMod \rar{\simeq} & C\Mod
\etik
\ee
the restriction of $f$ to $C$ is identity. Therefore, since $U$ is a free $C$ module of finite dimensions, $f$ must be an isomorphism and therefore $S$ is an equivalence.  We have finally finished the proof of Theorem \ref{KLrho} for $\rho=1$. 

\subsection{Generalizing to many copies of $\fgl (1|1)$}

Let us now generalize the above theorem to $V(\fgl (1|1))^{\otimes n}$. To do so, define $\CV^k:= (\CV_Z^1)^{\otimes k} \otimes V(\fgl (1|1))^{\otimes n-k}$. Each pair $(\CV^k, \CV^{k+1})$ individually satisfies the assumptions needed to apply Theorem  \ref{Thm:AmodNmod} and the Schauenburg’s functor. Here we treat $\CV^{k+1}$ as a commutative algebra object in $\CV^k\Mod$. Note that $\CV^k\Mod$ is rigid from the same argument as the rigidity of $KL_{1^n}$, by applying \cite[Section 4.1]{mcrae2023deligne}. We find fully-faithful braided tensor functors
\be
S_k: \CV^k\Mod\longrightarrow {}_{\mf N_k}^{\mf N_k}\CY\CD\lp \CV^{k+1}\Mod\rp, 
\ee
where $\mf N_k$ is the Nichols algebra corresponding to $\sigma_{\varphi_{k+1}} \CV^{k+1}$, and $\sigma_{\varphi_{k+1}}$ is the spectral flow associated with the $(k+1)$-factor of $\CV_Z^1$.  Combining all these functors, we obtain a fully-faithful functor
\be
S: KL_{1^n}\longrightarrow {}_{\mf N_1}^{\mf N_1}\CY\CD{}_{\mf N_2}^{\mf N_2}\CY\CD\cdots {}_{\mf N_{n}}^{\mf N_n}\CY\CD\lp (\CV_Z^1)^{\otimes n}\Mod\rp.
\ee
We can now apply \cite[Lemma 6.12]{CLR23}, which shows that the category on the right is the category of Yetter-Drinfeld modules of the Nichols algebra generated by the direct sum of $\sigma_{\varphi_{k}} (\CV_Z^1)^{\otimes n}$. We find that the category on the RHS is precisely the category of modules of $(U_1^E)^{\otimes n}$. The same argument as the previous section shows that $S$ must be an equivalence.

\section{Un-gauging Operation and Proof for General 
$\rho$}\label{sec:ungauge}

\subsection{Review of Un-gauging Operation for $V(\grho)$}

We will use the $\rho=\mathrm{Id}_n$ case to prove Theorem \ref{KLrho} for general $\rho$. To do so, we rely on the ungauging relation in Proposition \ref{Prop:ungaugeVOA}. Recall that this proposition shows that there is a lattice VOA $\CW_\rho$ of a self-dual lattice such that the VOA $V(\grho)\otimes \CW_\rho$ is a simple current extension of $V(\fgl (1|1))^{\otimes n}$. It defines a commutative algebra object:
\be
\CA_\rho:=\bigoplus_{\lambda, \mu\in \Z^{n-r}} \sigma_{\wt{\rho} (\lambda), \tau^\trans (\mu)+\wt{\rho} (\lambda)}V(\fgl(1|1))^{\otimes n}
\ee
in the ind-completion $\Ind (KL_{1^n})$. The work of \cite{creutzig2017tensor, creutzig2020direct} shows that one can relate the category of local modules of $\CA_\rho$ in $\Ind(KL_{1^n})$ with modules of the VOA $V(\fg_\rho)\otimes \CW_\rho$. More specifically, the category:
\be
\CA_\rho\Mod_{\text{loc}} \lp \Ind\lp KL_{1^n}\rp\rp
\ee
has the structure of a braided tensor category. Moreover, let $KL_{1^n}^{\rho, [0]}$ denote the full subcategory of $KL_{1^n}$ whose monodromy with $\CA_\rho$ is trivial; then there exists a lifting functor:
\be
\CL_\rho: KL_{1^n}^{\rho, [0]}\longrightarrow \CA_\rho\Mod_{\text{loc}} \lp \Ind\lp KL_{1^n}\rp\rp, M\mapsto \CA_\rho\otimes M
\ee
which is an exact surjective tensor functor. The following statement is true. 

\begin{Thm}[\cite{BCDN} Theorem 8.25]
There is an equivalence between the image of $KL_{1^n}^{\rho, [0]}$ under $\CL_\rho$ and $KL_\rho$. Consequently, there is an equivalence between $KL_\rho$ and the de-equivariantization of $KL_{1^n}^{\rho, [0]}$ by the lattice of simple modules defining the extension $\CA_\rho$. 

\end{Thm}

Using Corollary \ref{Cor:monsimple}, one can show that the category $KL_{1^n}^{\rho, [0]}$ consists of objects in $KL_{1^n}$ where the actions of 
\be
\sum \tau_{\alpha i}E^i_0,\qquad \sum \wt{\rho}_{i\alpha} N^i_0
\ee
are semisimple with integer eigenvalues. The de-equivariantization by the lattice identifies an object $M$ with $\sigma_{\wt{\rho} (\lambda), \tau^\trans(\mu)+\wt{\rho}  (\lambda)}(M)$ for any $\lambda,\mu \in \Z^{n-r}$. We suggestively write this category as $KL_{1^n}^{\rho, [0]}/\Z^{2(n-r)}$. Since we have an equivalence of BTC $KL_{1^n}\simeq (U_1^E)^{\otimes n}\Mod$, one can ask what algebra the de-equivariantization leads to. We will show in this section that this leads naturally to the quantum group $U_\rho^E$ as a subquotient of $(U_1^E)^{\otimes n}$. 

\subsection{Quantum Group for General $
\rho$ and the Physical TQFT}\label{subsec:QGphys}

Note that we have obtained an equivalence:
\be
KL_{1^n}\simeq (U_1^E)^{\otimes n}\Mod\simeq {}_{\mf N}^{\mf N}\CY\CD\lp C\Mod\rp
\ee
where $C$ is the Hopf algebra $\C[x^i,y^i, (-1)^F]$ and $\mf N$ is the Nichols algebra. We must understand the lattice with which we define $\CA_\rho$ in the category on the right hand side. 

\begin{Lem}\label{Lem:qgungauge}
    Let $n, e\in \Z^n$. The object $\CS(\sigma_{e, n+e}\lp V(\fgl (1|1))^{\otimes n}\rp)$ is $\C_{n,e}$ as an object of ${}_{\mf N}^{\mf N}\CY\CD\lp C\Mod\rp$, the 1-dimensional module where $x^i=n^i, y^i=e^i$. 
    
\end{Lem}

\begin{proof}
    It is known that $\sigma_{e, n+e}\lp V(\fgl (1|1))^{\otimes n}\rp$ is the unique sub-module of the module $\CF_{eX+nY}\otimes \FF^{\otimes n}$, restricted from $H_{X^i,Y^i}\otimes \FF^{\otimes n}$. Here recall that $\CF_{eX+nY}$ is the Fock module generated by $\vert eX+nY\rangle$. Under the equivalence:
    \be
A\lMod\simeq C\Mod
    \ee
    this goes to the one-dimensional representation $\C_{n,e}$ where $x^i$ acts as $n^i$ and $y^i$ acts as $e^i$. Note that by the definition of $\CS$, we have the following commutative diagram of functors:
    \be
\btik
A\Mod\dar & \lar{\Ind_A}KL_{1^n}\dar{S}\\
\mf N\Mod  &\lar{\mathrm{Res}} {}_{\mf N}^{\mf N}\CY\CD\lp C\Mod\rp
\etik
    \ee
By Frobenius reciprocity we have $\CF_{eX+nY}\otimes \FF^{\otimes n}\cong \Ind_A (\sigma_{e, n+e}V)$, and $\sigma_{e, n+e}V$ is the unique object in $KL_1$ with this property. Similarly, there is a unique object in ${}_{\mf N}^{\mf N}\CY\CD\lp C\Mod\rp$ whose restriction to $\mf N\Mod$ is $\C_{n,e}$ (because it is one-dimensional, and therefore the action of both $\mf N$ and $\mf N^*$ must be trivial). Therefore, since $\CF_{eX+nY}\otimes \FF$ and $\C_{n,e}$ are identified via the equivalence of $A\lMod\simeq C\Mod$, this shows that $\sigma_{e, n+e}\lp V(\fgl (1|1))^{\otimes n}\rp$ is identified with $\C_{n,e}$ under $\CS$.

\end{proof}

Therefore, the simple direct summands of $\CA_\rho$ can be identified with $\C_{\wt{\rho} (\lambda), \tau^\trans (\mu)}$ for $\lambda,\mu\in \Z^{n-r}$. We have an equivalence:
\be
{}_{\mf N}^{\mf N}\CY\CD\lp C\Mod\rp^{\rho, [0]}/\Z^{2(n-r)}\simeq KL_\rho.
\ee
Since the lattice of simple modules $\Z^{2(n-r)}$ clearly has trivial monodromy with $\mf N$ in $C\Mod$, there is an equivalence
\be
{}_{\mf N}^{\mf N}\CY\CD\lp C\Mod\rp^{\rho, [0]}/\Z^{2(n-r)}\simeq {}_{\mf N}^{\mf N}\CY\CD\lp C\Mod^{\rho, [0]}/\Z^{2(n-r)}\rp
\ee
thanks to \cite[Theorem 10.2]{CLR23}. Here $C\Mod^{[0]}/\Z^{2(n-r)}$ is the de-equivariantization of $C\Mod$ by the lattice defined by $\C_{\wt{\rho} (\lambda), \tau^\trans (\mu)}$. We claim:

\begin{Prop}\label{Prop:ungaugeC}
    There is an equivalence of braided tensor categories:
    \be
C\Mod^{[0]}/\Z^{2(n-r)}\simeq C_\rho\Mod,
    \ee
    where $C_\rho$ is the algebra from Section \ref{subsubsec:Nicholsrho}.
\end{Prop}

\begin{proof}
    Recall in Proposition \ref{Prop:qgungauge} we showed that $U_\rho^E$ is a subquotient of $(U_1^E)^{\otimes n}$. The same proof, ignoring the super-generators $\psi^i_\pm$, shows that $C_\rho$ is a subquotient of $C$. Indeed, let $\overline{C}$ be the quotient of $C$ by the relations:
    \be\label{eq:qgmonodromy}
\prod_i K_{N^i}^{\wt{\rho}_{i\alpha}}=1,\qquad \prod_i K_{E^i}^{\tau_{\alpha i}}=1
    \ee
    then $C_\rho$ is a subalgebra of $\overline{C}$ via the mapps:
    \be
N^a\mapsto \sum \rho^{ia} N^a,\qquad E^a\mapsto \sum_i \wt{\tau}_{ai} E^i, \qquad (-1)^F\mapsto (-1)^F.
    \ee
    By using the explicit $R$-matrix of $C$, it is easy to see that the subcategory of $C\Mod$ whose objects have trivial monodromy with $\C_{\wt{\rho} (\lambda), \tau^\trans (\mu)}$ is precisely the category of modules of $\overline{C}$. Let $C^\perp$ be the subalgebra of $\overline{C}$ generated by $N_\alpha=\sum \wt{\rho}_{i\alpha} N^i$ and $E_\alpha=\sum \tau_{\alpha i} E^i$ (without the parity operator), then we have an algebra decomposition:
\be
\overline{C}=C_\rho\otimes C^\perp.
\ee
This is true since $\wt{\rho}$ provides a splitting and therefore the matrix $(\rho, \wt{\rho})$ as well as $(\tau, \wt{\tau})$ are full rank. Here we put the parity operator $(-1)^F$ into the algebra $C_\rho$. Moreover, according to the proof of Proposition \ref{Prop:qgungauge}, this is a decomposition of $\overline{C}$ into Hopf algebras, such that the $R$-matrix on $\overline{C}$ is a product of $R$-matrices. We therefore have an equivalence of braided tensor categories:
\be
\overline{C}\Mod\simeq C_\rho\Mod\boxtimes C^\perp\Mod.\footnote{Note that in this equality, the category $C^\perp\Mod$ should be the category of ordinary finite-dimensional modules of the algebra $C^\perp$; there is no $\Z_2$ parity operator. Otherwise this equality is false.}
\ee
Because of the relations in equation \eqref{eq:qgmonodromy}, the direct summands of $\CA_\rho$ generates the category $C^\perp\Mod$, and we find, after de-equivariantization:
\be
\overline{C}\Mod/\Z^{2(n-r)}\simeq C_\rho\Mod\boxtimes C^\perp\Mod/\Z^{2(n-r)}\simeq C_\rho\Mod.
\ee
 This completes the proof. 
\end{proof}

We therefore have an equivalence:
\be
KL_\rho\simeq {}_{\mf N}^{\mf N}\CY\CD\lp C_\rho\Mod\rp.
\ee
Now Section \ref{subsubsec:Nicholsrho} shows that the RHS is represented by the Hopf algebra $U_\rho^E$. 

We will also quickly comment on how to obtain the quantum group for 3d $\CN=4$ gauge theories. The relevant VOA $\CV_\rho^B$ is a simple current extension of $V(\grho)$. Under the above equivalence, the simple currents $\sigma_{\lambda, \rho^\trans\rho \lambda}V$ is mapped to $\C_{0, \lambda}$ as an object of ${}_{\mf N}^{\mf N}\CY\CD\lp C_\rho\Mod\rp$. We therefore have:
\be
\CC_\rho^B\simeq {}_{\mf N}^{\mf N}\CY\CD\lp C_\rho\Mod^{[0]}/\Z^r\rp.
\ee
By a similar proof as in Proposition \ref{Prop:ungaugeC}, the category $C_\rho\Mod^{[0]}/\Z^r$ can be represented by a subquotient of $C_\rho$, generated by $X_a, e^{2\pi i Y_a}$, such that $e^{2\pi i X_a}=1$. We denote this algebra by $\wt{C}_\rho$. It is then not difficult to show that the category ${}_{\mf N}^{\mf N}\CY\CD\lp \wt{C}_\rho\rp$ is represented by a Hopf algebra. We summarize this in the following. 

\begin{Thm}
    Let $\wt{U}_\rho$ be the subalgebra of $U_\rho^E$ generated by $N^a, \psi_\pm^i, K_{E^a}^{\pm}$ such that $e^{2\pi i N^a}=1$. Then there is an equivalence of braided tensor categories:
    \be
\CC_\rho^B\simeq \wt{U}_\rho\Mod.
    \ee
    The $R$-matrix on the RHS can be conveniently written as:
    \be
R=\prod_a (1\otimes K_{E^a})^{N^a\otimes 1} \prod_i\lp 1-\psi_+^i\otimes \lp \prod_a K_{E^a}^{-\rho^{ia}}\psi_-^i\rp \rp. 
    \ee
\end{Thm}

\newpage

  \bibliographystyle{ytamsalpha}
   
   
   \bibliography{KLabelian}

\end{document}